\documentclass[12pt,a4paper]{report}

\usepackage[pdftex]{graphicx}
\usepackage{ifpdf}
\ifpdf
   \DeclareGraphicsExtensions{.pdf,.png,.jpg,.mps,.eps}
   \usepackage{hyperref} 
\else
\fi
 
\usepackage{sa_thesis_style}

\begin{document}

\newenvironment{dedication}
{
   \cleardoublepage
   \thispagestyle{empty}
   \vspace*{\stretch{1}}
   \hfill\begin{minipage}[t]{0.66\textwidth}
   \raggedright
}%
{
   \end{minipage}
   \vspace*{\stretch{3}}
   \clearpage
}

% !TEX encoding = UTF-8 Unicode
\begin{titlepage}

\begin{center}

% Upper part of the page. The '~' is needed because \\
% only works if a paragraph has started.
\includegraphics[width=0.5\textwidth]{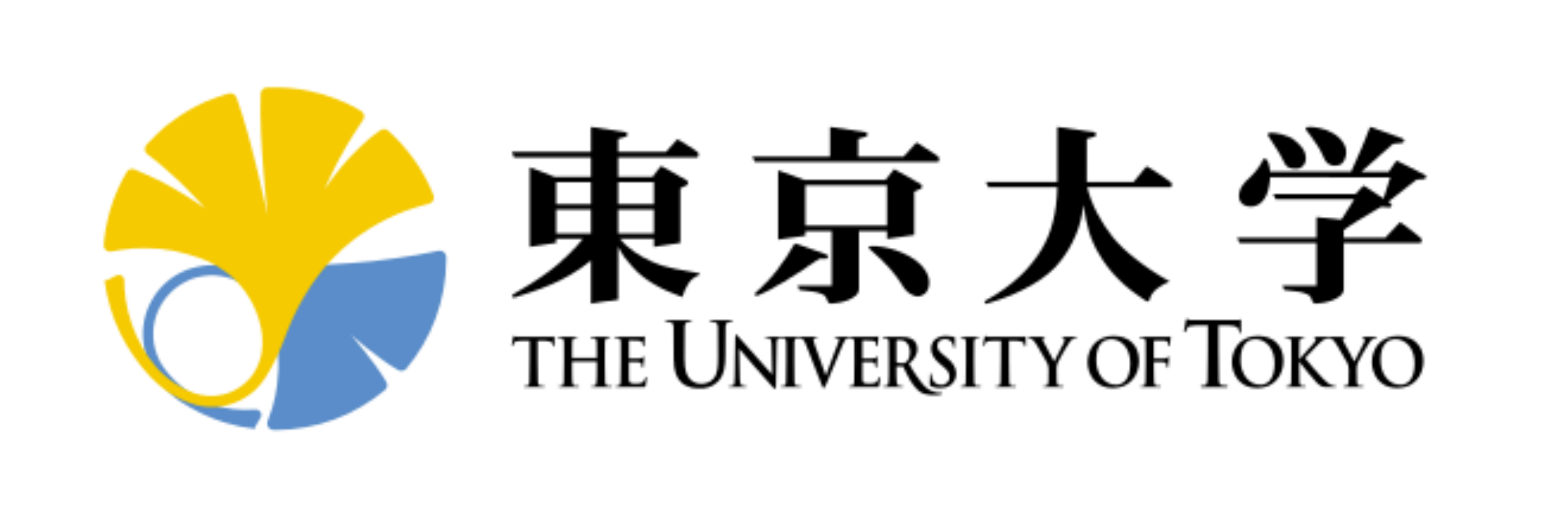}~\\[1cm]

%\textsc{\Large Doctorate thesis}\\[0.5cm]
%{\begin{CJK}{UTF8}{min}%font options: min, goth
{\LARGE Doctorate Thesis\\[1.5cm]}
%\end{CJK}}\\[0.5cm]

% Title
\HRule \\[0.4cm]
{ \LARGE \bfseries Multiple-particle diffusion processes
% with logarithmic interactions
  from the viewpoint of \\Dunkl operators: \\relaxation to the steady state\\[0.4cm] }

\HRule \\[0.5cm]

%\begin{CJK}{UTF8}{min}%font options: min, goth
%{\Large ダンクル演算子の視点から見た多粒子拡散過程：\\定常状態への緩和\\[1.0cm]}
%\end{CJK} \\[1.0cm]

\vfill

%{\begin{CJK}{UTF8}{min}%font options: min, goth
{\large Submitted on December 2013 for the degree of Doctor (Science)\\[1.0cm]}
%\end{CJK}}\\[1.0cm]

%\textsc{\Large The University of Tokyo}\\[1.5cm]
%{\begin{CJK}{UTF8}{min}%font options: min, goth
{\large The University of Tokyo, Graduate School of Science\\[0.5cm]
Department of Physics\\[1.0cm]}
%\end{CJK}}\\[1.0cm]
%\textsc{\Large Graduate School of Science,\\Department of Physics}\\[0.5cm]

% Author and supervisor
%\begin{minipage}{0.35\textwidth}
%\begin{flushleft} 
{\large \textsc{Andraus Robayo}, Sergio Andr{\'e}s\\}
%\end{flushleft}
%\end{minipage}
%\begin{minipage}{0.45\textwidth}
%\begin{flushright}
%\begin{CJK}{UTF8}{min}%font options: min, goth
%{\large アンドラウス　ロバジョ\ ・\ セルヒオ　アンドレス}
%\end{CJK}
%\end{flushright}
%\end{minipage}

\end{center}

\end{titlepage}
%\maketitle

\newpage
\thispagestyle{empty}
\mbox{}

\begin{dedication}
To Katja, Liam and Ana{\"i}s,\\
for reminding me of the merits of simplicity,\\
for creating order out of the chaos around you,\\
and for always making my day with a smile.
\end{dedication}

%\emptypage

% !TEX encoding = UTF-8 Unicode
\begin{abstract}

In this thesis, two families of stochastic interacting particle systems, the interacting Brownian motions and the interacting Bessel processes, are defined as extensions of Dyson's Brownian motion models and the eigenvalue processes of the Wishart and Laguerre processes. This is achieved by considering the parameter $\beta$ from random matrix theory as a real positive number. These systems consist of several particles which evolve as individual Brownian motions and Bessel processes, and that repel mutually through a logarithmic potential. The interacting Brownian motions and Bessel processes are realized as special cases of Dunkl processes, which are a broad family of multivariate stochastic processes defined by using the differential-difference operators known as Dunkl operators. One of the tools provided by Dunkl operator theory, the intertwining operator, relates spatial partial derivatives with Dunkl operators. It also maps multidimensional Brownian motions into Dunkl processes, but its explicit form is unknown in general. Therefore, the properties of all types of Dunkl processes can be examined by studying the characteristics of the intertwining operator. In this thesis, the steady state under an appropriate scaling and and the freezing ($\beta\to\infty$) regime of the interacting Brownian motions and Bessel processes are studied, and it is proved that the scaled steady-state distributions of these processes converge in finite time to the eigenvalue distributions of the $\beta$-Hermite and $\beta$-Laguerre ensembles of random matrices. Moreover, it is shown that the scaled final positions of the particles in these processes become fixed at the zeroes of the Hermite and Laguerre polynomials in the freezing limit. These results are obtained as the consequence of two more general results proved in this thesis. The first is that Dunkl processes in general converge in finite time to a scaled steady-state distribution that only depends on the type of Dunkl process considered. The second is that in the freezing limit, their scaled final position is fixed to a set of points called the peak set, which is the set of points which maximizes their steady-state distribution. In order to obtain these results, previously unknown relations involving the intertwining operator are derived for Dunkl processes in general, and in the case of the interacting Brownian motions and Bessel processes, the effect of the intertwining operator on symmetric polynomials is derived.

\end{abstract}

% !TEX encoding = UTF-8 Unicode
\chapter*{Acknowledgments}

I would like to express my gratitude toward all those people who, in one way or another, made this thesis possible. While this work is intended to be the fruit of my individual efforts, it would have been impossible for me to reach this point without the help of all those who gave me their support throughout my time as a graduate student.

First and foremost, I would like to express my gratitude to Professors Seiji Miyashita and Makoto Katori. I could have never reached this stage in my career without their help and guidance. For all the support, advice and knowledge they have given me, I am forever grateful. 

Secondly, I would like to thank Professors Yutaka Matsuo, Atsuo Kuniba, Tomohiro Sasamoto and Tomio Kobayashi for their insight and comments on this work. In particular, I would like to thank Professor Naomichi Hatano for his detailed reading of this manuscript, which has contributed greatly to its improvement.

I would also like to thank the members of the Miyashita research group who shared their time with me in the last three years. It was a real privilege to be in the company of Dr. Saito, Dr. Mohakud, Dr. Matsui, Mr. Kamatsuka, Mr. Nakada, Mr. Shirai, Mr. Endo and Mr. Futami. In particular I would like to thank Dr. Takashi Mori for the time he gave me for discussions and for his sharp insights on physics and mathematics. Also, I extend my thanks to the secretaries of Miyashita group, Mrs. Hiromi Okuzawa and Ms. Keiko Yashima for all of their help and support throughout these years.

Thanks to Dr. Michal Hajdu{\v s}ek, I had an excuse to have a laugh and an espresso every afternoon. I am tremendously grateful for his advice and his moral support during some of my hardest trials as a doctorate student.

Finally, I would like to thank my friends and extended family, who in their own way have encouraged me in spite of the distance, and in particular my wife and children, who have always supported me in following my dreams.

The author was supported by the Monbukagakusho: MEXT scholarship for foreign graduate students and by the Elements Strategy Initiative Center for Magnetic Materials.

\tableofcontents
\listoffigures
\listoftables
 
% !TEX encoding = UTF-8 Unicode
\chapter{Introduction}
%\addcontentsline{toc}{chapter}{Introduction}

\section{Motivation from random matrix theory in physics}

The topic of random matrix theory pioneered in nuclear physics by Wigner \cite{wigner55} is a widely studied subject that has seen great growth in the last decades (see, e.g., \cite{bohigasweidenmuller11} for a historical review). The  applications in physics of this theory have expanded to areas such as 2D quantum gravity \cite{difrancescoginspargzinnjustin95}, string theory \cite{dijkgraafvafa02, klemmmarinotheisen03, marino04}, quantum chromodynamics (QCD) \cite{thooft74, brezinitzyksonparisizuber78, verbaarschotwettig00}, quantum wires and quantum dots \cite{imry86, altshulershklovskii86, beenakker97, gorkoveliashberg65, dentonmuhlschlegelscalapino71}, resonance scattering \cite{fyodorovsommers97}, quantum and classical optics \cite{mellopereyrakumar88, beenakker98}, quantum entanglement \cite{majumdarbohigaslakshminarayan08}, topological insulators \cite{schnyderryufurusakiludwig08}, directed polymers \cite{degennes68, essamguttmann95, johansson00}, random growth models \cite{baikdeiftjohansson99, prahoferspohn00, baikrains01, sasamoto05}, and quantum chaos and graphs \cite{bohigasgiannonischmit84, gutzwiller90, haake10} among many others. Random matrix theory has also been applied in several fields outside of physics, such as wireless communications \cite{tulinoverdu04}, mathematical finance \cite{pottersbouchaudlaloux05} and RNA folding \cite{rivaseddy99}. While most of the applications of random matrix theory are theoretical, there have been observations of the statistical properties of random matrices in the interface growth of liquid crystals undergoing a phase transition \cite{takeuchisano10}.

By definition, random matrices are matrices whose components are random variables and which obey certain symmetries depending of the physical situation in which they are to be applied \cite{mehta04}. The main idea is that if a quantum system is sufficiently complex, its Hamiltonian may be replaced by a series of random matrices placed in a diagonal block form. Each block represents a set of conserved quantum numbers, and the type of random matrix chosen for each block depends on whether the system has time-reversal invariance and whether the system is in an integer or half-odd integer spin state. 

When the block under consideration corresponds to a time-reversal invariant system with integer spin, one chooses an \emph{orthogonal} random matrix ensemble, labeled by the parameter $\beta=1$, which is defined to be statistically invariant under orthogonal transformations. These matrices are real and symmetric. Similarly, for a time-reversal invariant system with half-odd integer spin, one chooses a \emph{symplectic} random matrix ensemble, labeled by the parameter $\beta=4$. In this case, the ensemble is statistically invariant under symplectic transformations and the matrices in the ensemble have quaternionic entries and are quaternion self-dual. Finally, when the system has no time-reversal symmetry, the corresponding ensemble is a \emph{unitary} ensemble, a set of Hermitian matrices with complex entries, with statistical invariance under unitary transformations. This ensemble is given the parameter $\beta=2$. It must be noted that the labels $\beta=1,2,4$ are not arbitrary, and they arise naturally in the study of the eigenvalue statistics of their corresponding ensembles.

The most famous ensembles of random matrices are the Gaussian ensembles \cite{dyson62B} and the Wishart ensembles \cite{wishart28}. The Gaussian ensembles of random matrices are the sets of real symmetric, complex Hermitian or quaternion self-dual matrices $H$ of size $N\times N$ whose entries are random variables that obey the normal distribution. That is, their entries $\{h_{ij}\}_{1\leq i,j\leq N}$ satisfy the symmetry requirement
\[h_{ij}=h_{ji}^*,\ i\leq j,\]
where ${}^*$ denotes complex or quaternion conjugation for complex or quaternion entries, and $h_{ij}^*=h_{ij}$ for real entries. These ensembles are called the Gaussian orthogonal, unitary and symmetric ensembles (GOE, GUE and GSE, respectively). Similarly, the Wishart ensembles are the sets of matrices $L$ of the form
\[L=Q^\dagger Q,\]
where $Q^\dagger$ is the conjugate transpose of $Q$, and $Q$ itself is a matrix of size $N\times M$ with real, complex or quaternion random variable entries which obey the normal distribution. There are many other ensembles of random matrices that exhibit properties which make them suitable for particular physical applications, and they have been classified by Altland and Zirnbauer \cite{zirnbauer96, altlandzirnbauer97}.

From this point of view, random matrices can basically be applied to any physical system with sufficient complexity, and the ensemble must be chosen taking into account the symmetries of the system. While their properties make random matrices applicable in a wide variety of fields, this applicability can be extended by considering matrix-valued stochastic processes. That is, instead of random-variable entries, one may formulate random matrices with \emph{stochastic processes} as entries, and consider the corresponding eigenvalue processes. 

Perhaps the most well-known example of this idea is Dyson's Brownian motion model (henceforth called Dyson model) \cite{dyson62}, where the entries of the Gaussian ensembles of random matrices are replaced by independent one-dimensional Brownian motions up to symmetry constraints. The particular case $\beta=2$ of this model has found applications in many branches of physics and mathematics due to its relationship with the vicious walker model introduced by Fisher \cite{fisher84}. Specifically, the vicious walker model is a discrete model where multiple random walkers move in a one-dimensional lattice, annihilating each other if they meet at the same lattice point. This model has been useful for the description of interface walls and melting transitions in two dimensions \cite{katori00}. It was proved by Katori and Tanemura that the Dyson model is the scaling limit of the vicious walker model \cite{katoritanemura02}, in the sense that the eigenvalue process of the Dyson model corresponds to a series of Brownian motions in one dimension constrained to never collide. For this reason, this particular case of the Dyson model is called the non-colliding Brownian motion \cite{katoritanemura07}. This problem had been considered also by de Gennes \cite{degennes68} in the context of lipid-water systems where non-crossing chain-like structures appear, and the properties of these systems were modeled as a \emph{fermionic gas} in one dimension evolving in time. In addition, the Dyson model for $\beta=2$ has found applications in polymer physics \cite{essamguttmann95}, the polynuclear growth model \cite{johansson03}, and traffic flow problems \cite{baikborodindeiftsuidan06}.

The Wishart \cite{bru91} (for the $\beta=1$) and Laguerre \cite{konigoconnell01} (for $\beta=2$) processes are also examples of matrix-valued processes of interest in physics. These are the processes obtained by putting independent Brownian motions in the entries of the matrix $Q$ used to define the Wishart ensemble, and the main object of study in this case is the resulting eigenvalue process. These processes are related to the chiral ensembles of random matrices \cite{katoritanemura04}, which themselves are used in QCD \cite{verbaarschotzahed93, verbaarschot94}. Much like the Dyson model, the case $\beta=2$ also has a non-colliding interpretation \cite{katoritanemura11} for which it is called the non-colliding Bessel process.

The objective of this work is to investigate the nature of these multivariate stochastic processes (the Dyson model and the Laguerre and Wishart processes) beyond the discrete values of the parameter $\beta=1,2,4.$ To motivate this extension of these matrix-valued processes to $\beta>0$ continuous, let us introduce some details of random matrix theory. It is a well-known fact \cite{mehta04} that the eigenvalues $\{\lambda_i\}_{i=1}^N$ of the Gaussian ensembles obey the joint distribution
\[\frac{1}{Z_\beta}\exp\Big[-\beta\Big(\sum_{i=1}^N\frac{\lambda_i^2}{2}-\sum_{1\leq i<j\leq N}\log|\lambda_i-\lambda_j|\Big)\Big],\]
where $Z_\beta$ is a partition function. Similarly, the eigenvalues $\{\lambda_i>0\}_{i=1}^M$ of the Wishart ensembles obey the joint distribution
\[\frac{1}{Z_\beta^\prime}\exp\Big[-\beta\Big(\sum_{i=1}^M\frac{\lambda_i}{2}-\frac{a}{2}\sum_{i=1}^M\log \lambda_i -\sum_{1\leq i<j\leq N}\log|\lambda_i-\lambda_j|\Big)\Big],\]
where $a=N-M+1-2/\beta$, and $Z_\beta^\prime$ is the corresponding partition function.

Because the joint eigenvalue densities for the Gaussian and Wishart ensembles have the form of Boltzmann factors, it is common to identify the eigenvalues of these ensembles of random matrices as systems of charged particles of unit charge in a two-dimensional universe that are restricted to move in one dimension. In the case of the Gaussian ensembles, these particles repel each other while being confined by a harmonic background potential, while in the case of the Wishart ensembles there are two background potentials, a linear confinement potential and a logarithmic repulsion potential from the origin. In this electrostatic analogy, both systems of charged particles are in contact with a heat reservoir of inverse temperature $\beta$, meaning that the mathematical parameter $\beta$ can be understood physically as the inverse temperature. Therefore, the matrix ensembles are realizations of these charged particle systems at the inverse temperatures $\beta=1,2$ and 4.

This physical interpretation motivates the extension of these models to continuous values of $\beta>0$. This was achieved by Dumitriu and Edelman \cite{dumitriuedelman02}, who defined a series of ensembles of real tridiagonal random matrices whose eigenvalues obey the joint eigenvalue densities of the Gaussian and Wishart ensembles for $\beta>0$. These ensembles are called the $\beta$-Hermite and $\beta$-Laguerre ensembles. Furthermore, Forrester \cite[Chap.~13]{forrester10} has succeeded in calculating the correlation functions for these $\beta$-ensembles for the case where $\beta$ is an even integer.

In the case of this work, the extension of these systems of charged particles to interacting-particle stochastic processes  for $\beta>0$ is considered. It is known that the dynamics of the eigenvalues of the Dyson model is derived using Bru's theorem \cite{katoritanemura04}. Denoting a vector of $N$ independent one-dimensional Brownian motions by $\bB_t$, the eigenvalue process of the Dyson model is given by the following stochastic differential equations (SDEs) for $i=1,\ldots,N$:
\[\ud\lambda_i=\ud B_{i,t}+\frac{\beta}{2}\sum_{\substack{j=1:\cr j\neq i}}^N\frac{\ud t}{\lambda_i-\lambda_j}.\]
In the case $\beta=2$, this is a system of independent Brownian motions conditioned never to collide. As a consequence of this, their joint probability density is given by the Karlin-McGregor determinant \cite{karlinmcgregor59}. Therefore, this process is determinantal in the sense that its joint probabilities and correlation functions are given by determinants of a single function called a correlation kernel, as shown by Katori and Tanemura \cite{katoritanemura07}. 

For the case of the Wishart and Laguerre processes, the dynamics of the eigenvalues of the real symmetric or complex Hermitian matrix $L$ is given by the SDE \cite{katoritanemura04}
\[\ud \lambda_i=2\sqrt{\lambda_i}\ud B_{i,t}+\beta\Big[(N+M)+\sum_{\substack{j=1:\cr j\neq i}}^N\frac{\lambda_i+\lambda_j}{\lambda_i-\lambda_j}\Big]\ud t.\]
Like in the case of the Dyson model for $\beta=2$, the eigenvalue dynamics of the Laguerre processes ($\beta=2$) is a determinantal process in which several independent one-dimensional processes are conditioned never to collide \cite{katoritanemura11}. However, the component processes are not Brownian motions, but squared Bessel processes. The SDE of a squared Bessel process $Y$ of dimension $D$ is given by \cite{oksendal}
\[\ud Y=2\sqrt{Y}\ud \tilde{B}_t + D\ud t,\]
and it represents the stochastic process realized by the squared distance to the origin of a Brownian motion in $D$ dimensions. These are called Bessel processes because their transition density (probability of arriving at the position $y$ after a time $t$ starting from the position $x$) is given by the function \cite{katoritanemura11}
\[\frac{1}{2t}\Big(\frac{y}{x}\Big)^{\nu/2}\rme^{-(x+y)/2t}I_\nu\Big(\frac{\sqrt{xy}}{t}\Big),\]
where $I_\nu(x)$ is the modified Bessel function of the first kind, and $\nu=D/2-1$ is the Bessel index. Therefore, the Laguerre eigenvalue processes are realized as $N$ squared Bessel processes of index $\nu=N+M-1$ conditioned never to collide.

In the SDEs given above, there is no constraint that should force the parameter $\beta$ to be discrete, save for their matrix-valued nature. As a matter of fact, these processes have no known matrix-valued representation except for the cases $\beta=1,2$ and 4, so the tradeoff for extending $\beta$ to a continuous parameter is that many of the techniques from random matrix theory cannot be used in this case. Because of this fundamental difference, the interacting particle systems considered in this work will be referred to as the interacting Brownian motions and the interacting (squared) Bessel processes.

\begin{figure}
\centering
{\renewcommand{\arraystretch}{1.5}
\renewcommand{\tabcolsep}{0.2cm}
\begin{tabular}{p{0.15\textwidth}|p{0.35\textwidth}c p{0.35\textwidth}}
\hline
&$\beta$ discrete				&&$\beta>0$ continuous\\
\hline
Static random variables&Gaussian and Wishart ensembles&$\to$&${\beta}$-Hermite and $\beta$-Laguerre ensembles\\
&$\downarrow$&&\\
Stochastic processes&Dyson's Brownian motions, Wishart and Laguerre processes	&$\to$&Interacting Brownian motions and Bessel processes\\
\hline
\end{tabular}}
\caption{Dynamical and continuous-temperature extensions of the Gaussian and Wishart ensembles. The bottom right corner has no known matrix-valued formulation.}\label{TableOfExtensions}
\end{figure}

\section{Dunkl processes}

An alternative way to formulate this extension is achieved by considering a broad family of stochastic processes that includes the interacting Brownian motions and Bessel processes as particular cases. These are called \emph{Dunkl processes} due to the fact that they are defined using the differential-difference operators known as \emph{Dunkl operators} \cite{dunkl89}. Dunkl defined these operators for the study of symmetric polynomials of multiple variables, and they have been used in physics in the context of the Calogero-Moser systems \cite{calogero71, moser75, bakerdunklforrester}. The definition of Dunkl operators depends on the choice of a finite set of vectors called \emph{root system}, which generates a reflection group $W$, and the root system is invariant under the action of the elements of $W$. Within the context of the Calogero-Moser systems, Forrester has used Dunkl operators to prove the integrability of these systems \cite[Secs.~11.4.2-11.5.5]{forrester10}.

The definition of Dunkl processes is rooted in the work of R{\"o}sler, who considered and found a Green function solution of the Dunkl heat equation \cite{rosler98}. This is the generalization of the heat equation in which the spatial partial derivatives are replaced by Dunkl operators. Subsequently, R{\"o}sler and Voit \cite{roslervoit98} considered the Dunkl heat equation as a Markov semigroup, and gave the first formal definition of the Dunkl processes. Intuitively speaking, Dunkl processes are defined as a generalization of multidimensional Brownian motion as follows. It is well known \cite{mahnke09} that the transition probability density of a Brownian motion obeys the heat equation. Then, if spatial partial derivatives are replaced by the Dunkl operators $\{T_i\}_{i=1}^N$, then one can define the Dunkl generalization of the heat equation \cite{rosler98} as follows,
\[\frac{\partial}{\partial t}f(t,\bx)=\frac{1}{2}\sum_{i=1}^N T_i^2 f(t,\bx),\]
where $f(t,\bx)$ is a continuous function and $\bx$ is a vector in $N$-dimensional space. Then, Dunkl processes are defined as the stochastic processes whose transition probability densities obey the Dunkl heat equation \cite{roslervoit98}. It must be noted, however, that because the Dunkl operators are differential \emph{difference} operators, there are difference terms in the Dunkl heat equation that represent jumps in the trajectory of Dunkl processes. This means that, in general, Dunkl processes are discontinuous.
 
However, it is possible to take the continuous part of Dunkl processes, defining what are called the \emph{radial Dunkl processes}, introduced by Gallardo and Yor \cite{gallardoyor05}. Radial Dunkl processes, then, are diffusion processes with a drift that is determined by the root system considered. In this work, two particular root systems will be considered, the root system of type $A$ and the root system of type $B$. These root systems generate the symmetric group $S_N$ of permutations of the components of vectors of $N$ dimensions, and the group composed of all permutations and sign changes of the components of vectors, respectively. The reason for the choice of these two root systems is that the interacting Brownian motions and Bessel processes are realized as the radial Dunkl processes of type $A$ and the radial Dunkl processes of type $B$ respectively. This is a fact first pointed out by Demni \cite{demni08A}. Therefore, Dunkl processes can be viewed as a large family of processes which can be reduced to multivariate stochastic processes known in physics and random matrix theory by choosing a particular root system and taking their continuous part. They provide a natural formulation of the interacting Brownian motions and Bessel processes, and they have the advantage of not being bound to any particular set of values of $\beta$, but they have the disadvantage of not having a matrix-valued representation in general. Without being aware of their relationship to Dunkl processes and Dunkl operator theory, Baker and Forrester \cite{bakerforrester97} introduced a series of functions, called \emph{generalized hypergeometric functions} which make part of the transition probability density of the radial Dunkl processes. These functions are the basis for some of the results in Chapter~\ref{ParticularCases} of this work.

Perhaps the greatest merit of using Dunkl processes to study the interacting Brownian motions and Bessel processes is that Dunkl operator theory provides a powerful tool to analyze their properties, called the \emph{intertwining operator} (introduced by Dunkl in \cite{dunkl91}). This operator, denoted by $V_\beta$, is a functional that is defined by the following relation between partial derivatives in space and Dunkl operators
\[T_i V_\beta [f(\bx)]=V_\beta\Big[\frac{\partial}{\partial x_i}f(\bx)\Big],\]
where the function $f(\bx)$ is assumed to be analytical and bounded for finite $\bx$. As a consequence, the intertwining operator maps the heat equation into the Dunkl heat equation, which in turn means that $V_\beta$ provides a great part of the information about the behavior of Dunkl processes. However, the general explicit form of $V_\beta$ is still unknown in spite of recent development in the topic \cite{maslouhiyoussfi09}.

\section{Main results}

The objective of this work is to derive previously unknown expressions for $V_\beta$ and make use of these expressions as novel tools for the study of Dunkl processes in general, and the interacting Brownian motions and Bessel processes in particular. The results are the following:
\begin{itemize}
\item
Previously unknown explicit expressions for the effect of the intertwining operator on particular functions such as linear functions at finite temperature and the exponential function in the freezing ($\beta\to\infty$) limit are calculated.
\item
The convergence to the steady state of Dunkl processes on an arbitrary root system for arbitrary initial conditions is considered, and a finite lower bound for the time required for relaxation to occur is given for finite non-zero temperatures. Note that the stochastic processes considered are diffusion processes, and therefore do not have a steady state in the strict sense. Therefore, in the context of this work, the phrase \emph{steady state} refers to the steady state achieved by these processes after being transformed by a suitable time scaling. The detailed definition of the steady state will be given in Chapter~\ref{general_steady}.
\item
The behavior of Dunkl processes in the freezing limit is calculated, and the freezing positions of these processes are shown to be given by the peak sets of reflection groups \cite{dunkl89B}. In the particular case of the interacting Brownian motions and Bessel processes, the freezing positions \cite{andrauskatorimiyashita12, andrauskatorimiyashita13} are shown to be given by the Fekete points \cite{fekete1923, deift00}.
\item
The steady-state regime of the interacting Brownian motions and Bessel processes is shown to coincide with the eigenvalue density of the $\beta$-Hermite and $\beta$-Laguerre ensembles of random matrices \cite{andrauskatorimiyashita13}.
\item
Using previous knowledge about the interacting Brownian motions and Bessel processes, expressions for the effect of $V_\beta$ on symmetric polynomials are derived.  
\item
The fact that any Dunkl process can be mapped to a Calogero-Moser system evolving in imaginary time by using a variable substitution in both space and time followed by a similarity transformation is proved \cite{andrauskatorimiyashita13}. This mapping provides an indirect derivation of some of the main results of this thesis.
\end{itemize}

This thesis is arranged as follows: in Chapter~\ref{preliminaries}, the basic notations and mathematical objects required for the derivation of the main result are introduced. In particular, Dunkl operators, Dunkl processes and their relationship with the interacting Brownian motions and Bessel processes are presented in concrete mathematical terms. In Chapter~\ref{CalogeroMoserCorrespondence}, the general correspondence between Dunkl processes and Calogero-Moser systems is proved and used to give an intuitive derivation of the relaxation to the steady state and the freezing limit of Dunkl processes. In Chapter~\ref{general_steady}, the precise definition of the steady state of Dunkl processes is given, and the fact that Dunkl processes that have been properly scaled in space relax to their steady state in finite time is proved for arbitrary initial distributions. In Chapter~\ref{general_freezing}, the freezing limit of Dunkl processes is calculated. In Chapter~\ref{ParticularCases}, the steady-state and freezing regime results for the particular cases of the interacting Brownian motions and Bessel processes are addressed. Several numerical results are presented as evidence of the validity of the results from Chapters~\ref{general_steady} and \ref{general_freezing}, and the explicit form of the intertwining operator for symmetrical polynomials is derived and studied. The derivation of the freezing limit and the steady-state regime at low temperature for the interacting Brownian motions and Bessel processes is given in the final part of the Chapter. Finally, the main results of this work and some future prospects are discussed in Chapter~\ref{conclusions}.

% !TEX encoding = UTF-8 Unicode
\chapter{Dunkl operator theory and Dunkl processes}\label{preliminaries}

The multivariate stochastic processes considered here are defined through the use of Dunkl operators, which in turn depend on several mathematical objects. In this chapter, the definition of those objects and Dunkl processes themselves will be reviewed, and some of their general properties will be listed. The contents of this chapter are based on \cite{dunklxu,rosler08}.

\section{Root systems}

Consider two column vectors $\bx=(x_1,\ldots,x_N)^T$ and $\balpha=(\alpha_1,\ldots,\alpha_N)^T\in\RR^N$, and their dot product $\bx\cdot\balpha=\bx^T\balpha=\sum_{i=1}^Nx_i \alpha_i$. The reflection operator through the hyperplane defined by $\balpha$ is given by
\begin{equation}
\sigma_{\balpha}\bx=\bx-2\frac{\balpha\cdot\bx}{\balpha\cdot\balpha}\balpha.
\end{equation}
In matrix notation, assuming that $\balpha$ is a column vector, denoting its transpose by $\balpha^T$ and denoting the identity matrix by $I$, one may write
\begin{equation}
\sigma_{\balpha}=I-2\frac{\balpha\balpha^T}{\balpha^T\balpha}.
\end{equation}
For simplicity, the norm of a vector will be denoted by $\sqrt{\bx\cdot\bx}=x$ in the case where the notation does not cause confusion. A property of $\sigma_{\balpha}$ is that, for $\Theta\in O(N)$ (the group of orthogonal matrices of size $N$), it satisfies the equation
\begin{equation}\label{EquationSigmaOrthogonalTransformation}
\sigma_{\Theta\balpha}\bx=\Theta\sigma_{\balpha}\Theta^T\bx=\Theta\sigma_{\balpha}\Theta^{-1}\bx.
\end{equation}

A \emph{root system}, denoted by $R$, is defined as a set of vectors (called \emph{roots}) that is closed under reflections along its elements. That is, $R$ satisfies the relation
\begin{equation}
\sigma_{\balpha} R=\{\sigma_{\balpha}\bxi : \bxi\in R\}=R
\end{equation}
for all $\balpha\in R$. In particular, $\sigma_{\balpha}\balpha=-\balpha\in R$, meaning that every root has its negative in $R$. For this reason, $R$ can be divided into the \emph{positive} and \emph{negative subsystems} as follows: choose an arbitrary vector, say $\bmm$, such that for any root $\balpha$, $\bmm\neq c\balpha$ with $c\in\RR$ and $\balpha\cdot\bmm\neq 0$; then, construct the positive subsystem as
$R_+=\{\balpha\in R: \bmm\cdot\balpha>0\}$ and the negative subsystem $R_-$ in the same manner. A root system is called \emph{reduced} if, for every $\balpha\in R$, $r\balpha\in R$ implies that $r=\pm 1$. It will be assumed that every root system considered here is reduced. 

The reflections along the roots of $R$ with composition as group operation generate a Weyl group, denoted by $W$, of orthogonal operators. By definition, applying the elements of $W$ to any root $\balpha\in R$ produces a subset of $R$. This set is denoted by
\begin{equation}
W\balpha=\{\rho\balpha : \rho\in W\}.
\end{equation}
A \emph{multiplicity function} is a function $k: R\to \CCC$ that assigns a unique complex parameter to all the roots that belong to the subset $W\balpha$ for some root $\balpha$. That is, if for $\bxi, \bzeta, \balpha\in R$ the equation $\sigma_{\balpha}\bzeta=\bxi$ holds, then $k(\bzeta)=k(\bxi)$. In general, it will be assumed that the multiplicity function is real and positive.

In general, the action of an orthogonal operator $\Theta\in O(N)$ on a function $f(\bx)$ is given by
\begin{equation}
\Theta f(\bx)=f(\Theta^T\bx),\text{ and } \Theta^{-1} f(\bx)\stackrel{\text{notation}}=\Theta^T f(\bx)=f(\Theta\bx).
\end{equation}
In particular, the action of a succession of reflections $\sigma_{\balpha_1}\ldots\sigma_{\balpha_n}$ on $f(\bx)$, with $\balpha_j\in \RR^N$ for $1\leq j\leq n$ is given by
\begin{equation}
\sigma_{\balpha_1}\ldots\sigma_{\balpha_n}f(\bx)=f(\sigma_{\balpha_n}\ldots\sigma_{\balpha_1}\bx),
\end{equation}
because $\sigma_{\balpha}$ is represented by a symmetric matrix. A function is called $W$-invariant if it satisfies the equation
\begin{equation}\label{EquationWActionOnFunctions}
\rho f(\bx)=f(\rho^T\bx)=f(\bx)
\end{equation}
for all $\rho\in W$.

\section{Dunkl operators}

The $i$th \emph{Dunkl operator} $T_i$, $i=1,\ldots,N$, is written as
\begin{equation}
T_i f(\bx) = \frac{\partial}{\partial x_i}f(\bx)+\sum_{\balpha\in R_+}k(\balpha)\frac{(1-\sigma_{\balpha})f(\bx)}{\balpha\cdot\bx}\alpha_i.\label{DunklOperatorR}
\end{equation}
More generally, using the gradient vector $\bnabla=(\partial/\partial x_1,\ldots,\partial/\partial x_N)^T$, the Dunkl operator in the direction $\bxi\in\RR^N$ is defined as
\begin{equation}
T_{\bxi} f(\bx) = \sum_{i=1}^N \xi_i T_i f(\bx) = \bxi\cdot\bnabla f(\bx)+\sum_{\balpha\in R_+}k(\balpha)\frac{(1-\sigma_{\balpha})f(\bx)}{\balpha\cdot\bx}\bxi\cdot\balpha.
\end{equation}
Note that, if $f(\bx)$ is a polynomial of degree $n$, then $(1-\sigma_{\balpha})f(\bx)/\balpha\cdot\bx$ is a polynomial of degree $n-1$. To see this, consider without loss of generality a particular monomial of $f(\bx)$, and assume that the coordinate system of the Euclidean space is set up so that $\balpha=\be_1$, where $\be_i$ is the $i$th canonical base vector. Then the expression becomes
\begin{equation}
\frac{(1-\sigma_{\be_1})c\prod_{i=1}^N x_i^{p_i}}{x1}=\frac{x_1^{p_1}-(-x_1)^{p_1}}{x_1}c\prod_{i=2}^N x_i^{p_i},
\end{equation}
and this is 0 for $p_1$ even and $2x_1^{p_1-1}c\prod_{i=2}^N x_i^{p_i}$ for $p_i$ odd. In both cases, this ratio yields a monomial, and if the same strategy is followed for every monomial, the ratio on the r.h.s. of \eqref{DunklOperatorR} becomes a polynomial of degree $n-1$. This means that $T_i$ is a homogeneous operator of degree $-1$, like a partial derivative. In addition, for a fixed multiplicity function the operators $\{T_i\}_{i=1}^N$ commute. Also, when at least one of the functions $f,g$ is $W$-invariant, the Dunkl operators obey the product rule,
\begin{equation}\label{DunklProductRule}
T_i[f(\bx)g(\bx)]=g(\bx)T_if(\bx)+f(\bx)T_ig(\bx).
\end{equation}

The Dunkl Laplacian is a generalization of the Laplacian in which every partial derivative is replaced by a Dunkl operator,
\begin{equation}\label{ReplacementDerivativeDunklOperator}
\Delta:=\sum_{i=1}^N \frac{\partial^2}{\partial x_i^2}\quad\to\quad \sum_{i=1}^N T_i^2.
\end{equation}
Dunkl \cite{dunklxu} proved that the explicit form of the Dunkl Laplacian is
\begin{equation}\label{EquationDunklLaplacian}
\sum_{i=1}^N T_i^2 f(\bx)=\Delta f(\bx)+2\sum_{\balpha\in R_+}k(\balpha)\Big[\frac{\balpha\cdot\bnabla f(\bx)}{\balpha\cdot\bx}-\frac{\alpha^2}{2}\frac{(1-\sigma_{\balpha})f(\bx)}{(\balpha\cdot\bx)^2}\Big].
\end{equation}
Consider an arbitrary orthonormal base of $\RR^N$, $\{\btheta_i\}_{i=1}^N$. Denote the $j$th component of $\btheta_i$ by $\theta_{ji}$, and the matrix formed by the basis vectors as columns by $\Theta$. Then, one has
\begin{multline}\label{EquationDunklLaplacianBaseIndependence}
\sum_{i=1}^N T_{\btheta_i}^2 f(\bx)=\sum_{1\leq i,j,l\leq N} \theta_{ji}\theta_{li} T_{j}T_{l} f(\bx)=\sum_{1\leq j,l\leq N} [\Theta^T\Theta]_{jl} T_{j}T_{l} f(\bx)\\
= \sum_{j=1}^NT_{j}^2 f(\bx).
\end{multline}
Therefore, the Dunkl Laplacian is independent of the orthonormal basis chosen to calculate it.

\section{Dunkl processes}

Dunkl processes are defined \cite{roslervoit98} as the Markov processes which obey the Dunkl heat equation 
\begin{equation}
\frac{\partial}{\partial t}-\frac{1}{2}\sum_{i=1}^NT_i^2=0
\end{equation}
as their backward Kolmogorov equation (BKE). Denoting the transition density of a Dunkl process going from the position $\bx$ to the position $\by$ in a time $t$ by $p(t,\by|\bx)$, the BKE is given explicitly by
\begin{multline}
\frac{\partial}{\partial t}p(t,\by|\bx)=\frac{1}{2}\sum_{i=1}^NT_{i}^2p(t,\by|\bx)\\
=\frac{1}{2}\Delta p(t,\by|\bx)+\sum_{\balpha\in R_+}k(\balpha)\Big[\frac{\balpha\cdot\bnabla}{\balpha\cdot\bx}p(t,\by|\bx)-\frac{\alpha^2}{2}\frac{1-\sigma_{\balpha}}{(\balpha\cdot\bx)^2}p(t,\by|\bx)\Big].\label{DunklBackward}
\end{multline}
All the operands act on the variable $\bx$ in this equation. Some of the general properties of these processes can be read off from each of the terms in this equation: the first term is a simple diffusion term, while the second term is a drift term which drives the process in the directions given by the roots of $R$. The third term is a difference term which generates a probability flow from the point $\bx$ to the points $\{\sigma_{\balpha}\bx\}_{\balpha\in R+}$. This means that Dunkl processes are discontinuous, and that they \emph{jump} to any one of the reflected positions generated by the root system. The corresponding forward Kolmogorov equation (FKE) is given by
\begin{multline}
\frac{\partial}{\partial t}p(t,\by|\bx)=\frac{1}{2}\Delta^{(y)} p(t,\by|\bx)-\sum_{\balpha\in R_+}k(\balpha)\frac{\balpha\cdot\bnabla^{(y)} p(t,\by|\bx)}{\balpha\cdot\by}\\
+\sum_{\balpha\in R_+}k(\balpha)\frac{\alpha^2}{2}\frac{p(t,\by|\bx)+p(t,\sigma_{\balpha}\by|\bx)}{(\balpha\cdot\by)^2}.\label{DunklForward}
\end{multline}

Though there is lack of physical intuition or interpretation for the jumps in Dunkl processes, there are some well-known facts about them \cite{gallardoyor08}. The most important is the fact that at every infinitesimal time increment, the process jumps at most once, and if the positions of the process before and after the jump are $\bx$ and $\by$, respectively, one can always find a root $\balpha$ such that $\by=\sigma_{\balpha}\bx$. This means that the jumps of the Dunkl processes can be eliminated by considering the \emph{``radial''} part of the trajectory of the process \cite{gallardoyor05}. 

\begin{figure}[!ht]
\centering
\includegraphics[width=\textwidth]{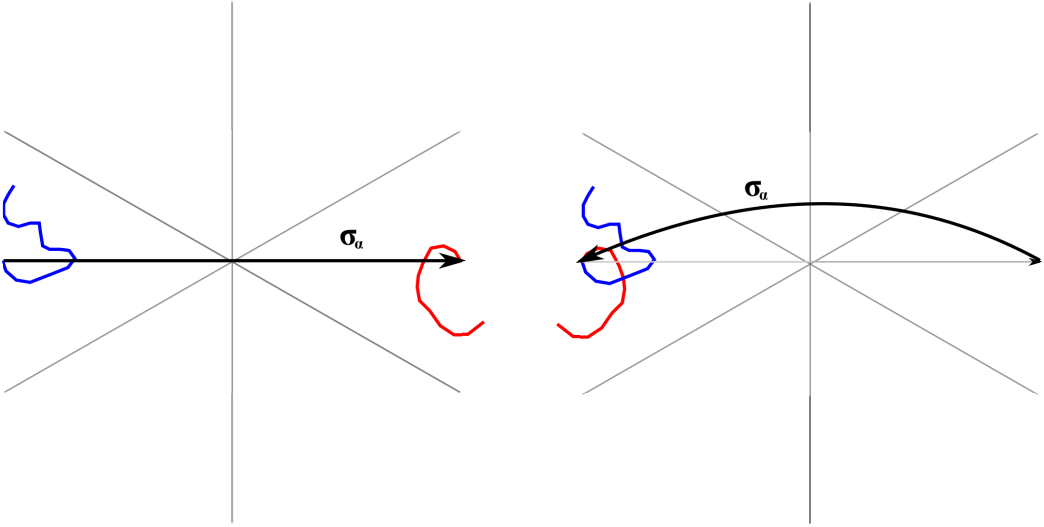}
\caption{Example path of a Dunkl process on the root system $A_2$. On the left part of the figure, the process completes the blue path before making a jump, and continues moving along the red path in a different region of space. On the right part of the figure, the radial or continuous part of the process is obtained by applying the same reflection that provoked the jump ($\sigma_{\balpha}$ in this case), bringing the process back to the end of the blue curve and forcing it to describe a continuous curve. The wedge on the left side of both sides of the figure is the Weyl chamber $C$ in this case.}\label{DunklPath}
\end{figure}
%MAKE A FIGURE OF THE JUMPS!!

Intuitively, radial Dunkl processes are obtained by projecting the path of a Dunkl process onto a subset of $\RR^N$ called the Weyl chamber, defined by
\begin{equation}
C=\{\bx\in\RR^N:\bx\cdot\balpha>0\ \forall \balpha\in R_+\}.
\end{equation}
The procedure is as follows: the Dunkl process is started from a point within $C$, and its path is followed until a jump occurs. After the jump, the process is projected back to the Weyl chamber by reflecting the path using the correct operator $\sigma_{\balpha}$ (see Fig.~\ref{DunklPath}). Repeating this procedure for every jump gives a trajectory that is continuous and contained in $C$. Consequently, the transition probability density of a radial Dunkl process is $W$-invariant. That is, the transition probability density of a radial Dunkl process, denoted by $P(t,\by|\bx)$, is related to the transition probability density of a non-radial Dunkl process by the equation
\begin{equation}\label{RegularToRadial}
P(t,\by|\bx)=\sum_{\rho\in W}p(t,\by|\rho\bx).
\end{equation}
The reason for this relation is that the jumps generate a probability flow from a point $\bx$ to the set of points $W\bx$, and the sum on the r.h.s. balances the flow of probability into and out of $C$. One way to see this is that $p(t,\by|\bx)$ is defined to be normalized over $\RR^N$, and by construction, $P(t,\by|\bx)$ is normalized over $C$.

The BKE of a radial Dunkl process is obtained from Eqs.~\eqref{DunklBackward} and \eqref{RegularToRadial}. Using the operator $\sum_{\rho\in W}\rho$ on Eq.~\eqref{DunklBackward}, and making it act on $\bx$, one obtains
\begin{multline}
\frac{\partial}{\partial t}P(t,\by|\bx)=\frac{\partial}{\partial t}\sum_{\rho\in W}p(t,\by|\rho\bx)=\frac{1}{2}\sum_{i=1}^N\sum_{\rho\in W}\rho[T_{i}^2p(t,\by|\bx)]\\
=\frac{1}{2}\sum_{i=1}^NT_{i}^2\sum_{\rho\in W} p(t,\by|\rho\bx)=\frac{1}{2}\Delta P(t,\by|\bx)+\sum_{\balpha\in R_+}k(\balpha)\frac{\balpha\cdot\bnabla}{\balpha\cdot\bx}P(t,\by|\bx).
\end{multline}
The third equality follows from Eq.~\eqref{EquationDunklLaplacianBaseIndependence} and the last equality follows from the $W$-invariance of $P(t,\by|\bx)$.

It is a known fact \cite{demni08A} that the interacting Brownian motions and the interacting Bessel processes are realized as the radial Dunkl processes of type $A$ and $B$ respectively, provided that the multiplicity function $k(\balpha)$ is chosen appropriately. Consider first the interacting Brownian motions. The root system of type $A$ is given by 
\begin{equation}
A=\{\balpha_{ij}=\be_i-\be_j:1\leq i\neq j\leq N\}.
\end{equation}
The positive subsystem is chosen as
\begin{equation}
A_{+}=\{\balpha_{ij}=\be_i-\be_j:1\leq j<i\leq N\}\label{positivesubsystemtypeA}
\end{equation}
where $\be_i$ denotes the $i$th unit base vector. Let us use the notation $\sigma_{\alpha_{ij}}=\sigma_{ij}$ for this particular root system. Note that the effect of $\sigma_{ij}$ on an arbitrary vector $\bx$ is that of exchanging its $i$th and $j$th components. To see this, compute the $l$th component of $\sigma_{ij}\bx$:
\begin{equation}
(\sigma_{ij}\bx)_l=x_l-(x_i-x_j)(\delta_{il}-\delta_{jl}).
\end{equation}
It is easy to see that $x_l$ remains unchanged for $l\neq i,j$, that $(\sigma_{ij}\bx)_i=x_j$ and that $(\sigma_{ij}\bx)_j=x_i$. Therefore, the group generated by the reflections along the elements of $A$ with composition as the group operation is the symmetric group $S_N$.

In view of this property of $A$, it follows that any root can be obtained from at most two reflections of any other root, a fact that is proved as follows. Consider an arbitrary root $\balpha_{ij}$ and apply to it the reflection $\sigma_{mj}$, with $m$ arbitrary. This reflection exchanges the $j$th and the $m$th components of $\balpha_{ij}$, giving $\balpha_{im}$. One more reflection using $\sigma_{il}$, with $l$ arbitrary gives $\balpha_{lm}$, as desired. Since $k(\balpha)$ is invariant under any of these reflections, one obtains
\[k(\balpha_{ij})=k(\sigma_{mj}\sigma_{il}\balpha_{ij})=k(\balpha_{lm})\]
in general, and therefore it can be concluded that $k(\balpha)$ is independent of its argument, so it is a single parameter.

Therefore, the Dunkl operators of type $A$ are given by
\begin{equation}
{}_{A}T_i f(\bx)= \frac{\partial}{\partial x_i}f(\bx)+k\sum_{\substack{j=1\\j\neq i}}^N\frac{f(\bx)-f(\sigma_{ij}\bx)}{x_i-x_j},
\end{equation}
and the corresponding BKE of the radial Dunkl process of type $A$ is given by
\begin{equation}\label{KBEInteractingBrownianMotions}
\frac{\partial}{\partial t}P_A(t,\by|\bx)=\frac{1}{2}\Delta P_A(t,\by|\bx) +k\sum_{1\leq i\neq j\leq N}\frac{1}{x_i-x_j}\frac{\partial}{\partial x_i}P_A(t,\by|\bx).
\end{equation}
Comparing this equation with the BKE of the interacting Brownian motions \cite{katoritanemura07},
\begin{equation}
\frac{\partial}{\partial t}p_\text{IBM}(t,\by|\bx)=\frac{1}{2}\Delta p_\text{IBM}(t,\by|\bx) +\frac{\beta}{2}\sum_{1\leq i\neq j\leq N}\frac{1}{x_i-x_j}\frac{\partial}{\partial x_i}p_\text{IBM}(t,\by|\bx),
\end{equation}
it follows that the radial Dunkl processes of type $A$ are equivalent (in distribution) to the interacting Brownian motions, provided one sets $k=\beta/2$. Henceforth, the transition probability density of the interacting Brownian motions will be denoted by $P_A(t,\by|\bx)$ while assuming that $k=\beta/2$.

Let us consider now the interacting Bessel processes. The root system of type $B$ is given by
\begin{equation}
B=\{\pm \be_i : 1\leq i\leq N\}\cup\{\pm(\be_i-\be_j),\pm(\be_i+\be_j):1\leq j<i\leq N\}.
\end{equation}
The positive subsystem is chosen as
\begin{equation}
B_{+}=\{ \be_i : 1\leq i\leq N\}\cup\{\be_i-\be_j,\be_i+\be_j:1\leq j<i\leq N\}.
\end{equation}
The reflections defined by the roots of $B$ are given by
\begin{eqnarray}
\sigma_{\pm\be_i}\bx&=&(x_1,\ldots,x_{i-1},-x_i,x_{i+1},\ldots,x_N),\nonumber\\
\sigma_{\pm(\be_i-\be_j)}\bx&=&(x_1,\ldots,x_{i-1},x_j,x_{i+1},\ldots,x_{j-1},x_i,x_{j+1},\ldots,x_N),\nonumber\\
\sigma_{\pm(\be_i+\be_j)}\bx&=&(x_1,\ldots,x_{i-1},-x_j,x_{i+1},\ldots,x_{j-1},-x_i,x_{j+1},\ldots,x_N).\IEEEeqnarraynumspace
\end{eqnarray}
The first reflection changes the sign of the $i$th component, the second reflection exchanges the $i$th and $j$th components and the third reflection exchanges the $i$th and $j$th components and changes their signs. Let us set the notations
\begin{IEEEeqnarray}{rCl}
\sigma_{\pm\be_i}&=&\hat{\sigma_i},\nonumber\\
\sigma_{\pm(\be_i-\be_j)}\bx&=&\sigma_{ij},\nonumber\\
\sigma_{\pm(\be_i+\be_j)}\bx&=&\hat{\sigma_i}\hat{\sigma_j}\sigma_{ij}.
\end{IEEEeqnarray}
Therefore, the reflection group $W_B$ contains all the permutations and sign changes that can be applied to a vector in $\RR^N$.

There are two multiplicities associated with $B$ because the roots $\{\pm\be_i\}_{i=1}^N$ and the roots $\{\pm\be_i\pm\be_j\}_{1\leq i\neq j\leq N}$ belong to different orbits. Let us set $k(\be_i)=k_0$ and $k(\pm\be_i\pm\be_j)=k_1$. Then, the Dunkl operator of type $B$ is given by
\begin{eqnarray}
{}_{B}T_i f(\bx)&=& \frac{\partial}{\partial x_i}f(\bx)+k_0\frac{f(\bx)-f(\hat{\sigma_i}\bx)}{x_i}\nonumber\\
&&k_1\sum_{\substack{j=1\\j\neq i}}^N\Bigg[\frac{f(\bx)-f(\sigma_{ij}\bx)}{x_i-x_j}+\frac{f(\bx)-f(\sigma_{ij}\hat{\sigma_i}\hat{\sigma_j}\bx)}{x_i+x_j}\Bigg],
\end{eqnarray}
and the BKE of the radial Dunkl process of type $B$ reads
\begin{multline}\label{KBEInteractingBesselProcesses}
\frac{\partial}{\partial t}P_B(t,\by|\bx)=\frac{1}{2}\Delta P_B(t,\by|\bx) +k_0\sum_{i=1}^N\frac{1}{x_i}\frac{\partial}{\partial x_i}P_B(t,\by|\bx)\\
+k_1\sum_{\substack{j=1\\j\neq i}}^N\Big[\frac{1}{x_i-x_j}+\frac{1}{x_i+x_j}\Big]\frac{\partial}{\partial x_i}P_B(t,\by|\bx).
\end{multline}
Comparing this equation with the BKE of the interacting Bessel processes \cite{katoritanemura11}
\begin{multline}
\frac{\partial}{\partial t}p_\text{IBP}(t,\by|\bx)=\frac{1}{2}\Delta p_\text{IBP}(t,\by|\bx)+\frac{\beta}{2}\Big[\sum_{i=1}^N\frac{2\nu+1}{2x_i}\frac{\partial}{\partial x_i}p_\text{IBP}(t,\by|\bx)\\
+\sum_{1\leq i\neq j\leq N}\Big(\frac{1}{x_i-x_j}+\frac{1}{x_i+x_j}\Big)\frac{\partial}{\partial x_i}p_\text{IBP}(t,\by|\bx)\Big],
\end{multline}
it follows that for $k_0=\beta(\nu+1/2)/2$ and $k_1=\beta/2$ the radial Dunkl processes of type $B$ and the interacting Bessel processes are equivalent in distribution. Unless otherwise noted, the transition probability density of the interacting Bessel processes will be denoted $P_B(t,\by|\bx)$ with the multiplicities chosen as indicated here.

Note that for both the interacting Brownian motions and Bessel processes, the multiplicities $k(\balpha)$ are proportional to $\beta/2$. Because the freezing limit consists of taking the limit $\beta\to\infty$, it is necessary to redefine the multiplicities so that they may be proportional to the inverse temperature. This is accomplished in two steps. First, choose one particular root $\balpha_0$ and set
\begin{equation}
k(\balpha_0)=\frac{\beta}{2}>0.
\end{equation}
Second, define a new multiplicity function $\kappa(\balpha)$, using the equation
\begin{equation}
\kappa(\balpha):=\frac{k(\balpha)}{k(\balpha_0)}>0.
\end{equation}
With this, one obtains the equation
\begin{equation}
k(\balpha)=\frac{\beta}{2}\kappa(\balpha).
\end{equation}
A quantity that appears repeatedly in calculations that involve Dunkl operators is the sum of the multiplicities over $R_+$,
\begin{equation}
\sum_{\balpha\in R_+}k(\balpha).
\end{equation}
The parameter $\gamma$ is given by
\begin{equation}\label{DefinitionParameterGamma}
\gamma:=\sum_{\balpha\in R_+}\kappa(\balpha),
\end{equation}
which gives
\begin{equation}
\sum_{\balpha\in R_+}k(\balpha)=\frac{\beta}{2}\gamma.
\end{equation}

\section{Dunkl's intertwining operator}

Dunkl operators are related to partial derivatives by the \emph{intertwining operator} $V_\beta$, which is a linear operator that conserves the degree of homogeneous polynomials. It is defined by the relation
\begin{equation}
T_iV_\beta f(\bx)=V_\beta\Big[\frac{\partial}{\partial x_i}f(\bx)\Big],\label{EquationVDefinition}
\end{equation}
and it is normalized by the relation $V_\beta 1 = 1$. This operator was introduced by Dunkl in \cite{dunkl91} (see also \cite{dunklxu}), and it is a powerful tool which allows one to treat Dunkl operators in almost the same way as partial derivatives. However, the explicit general form of the  intertwining operator is unknown. Its form is known, e.g., in the one-dimensional case \cite{dunkl91} and for the root system of type $A$ in three dimensions ($A_2$) \cite{dunkl95}. While some progress has been achieved in recent years \cite{maslouhiyoussfi09}, the general explicit effect of $V_\beta$ on arbitrary functions remains an open question.

The most important properties of $V_\beta$ are listed as follows. $V_\beta$ commutes with the action of $\rho\in W$,
\begin{equation}
 V_\beta=\rho^TV_\beta \rho.
\end{equation}
This follows from the fact that the operator on the r.h.s. satisfies Eq.~\eqref{EquationVDefinition}. In addition, a theorem by R\"{o}sler \cite{rosler99} gives bounds for functions deformed by the intertwining operator. The space of functions considered for this property is denoted by $A_r$, and it is defined using several mathematical objects. Denote by $\mathcal{P}_n^N$ the set of homogeneous polynomials on $\bx\in\RR^N$ of degree $n$. Define by $K_r=\{\bx\in\RR^N: |\bx|\leq r\}$ the $N$-dimensional ball of radius $r$, and denote by $||g||_{\infty,K_r}$ the maximum value of $|g(\bx)|$ within $K_r$. Then, $A_r$ is defined as the set of all functions $g:K_r\to\mathbb{C}$ such that
\begin{equation}\label{EquationDefinitionAr1}
g(\bx)=\sum_{n=0}^\infty g_n(\bx),\textrm{ with }g_n(\bx)\in\mathcal{P}_n^N
\end{equation}
and
\begin{equation}\label{EquationDefinitionAr2}
\sum_{n=0}^\infty ||g_n||_{\infty,K_r}<\infty.
\end{equation}
The theorem is stated as follows.

\begin{proposition}\label{PositivityVBeta}
(Thm. 1.2 and Cor. 5.3 in \cite{rosler99}) for $\beta, \kappa\geq 0$ and for every $\bx\in\RR^N$, there is a unique probability measure $\mu_{\bx}$ on the Borel $\sigma$-algebra of $\RR^N$ such that
\begin{equation}
V_\beta g(\bx) = \int_{\RR^N}g(\bxi)\ud\mu_{\bx}(\bxi)
\end{equation}
for all functions $g(\bx)\in A_{|\bx|}$. This measure satisfies
\begin{equation}
\mu_{r\bx}(B)=\mu_{\bx}(r^{-1}B)\textrm{ and }\mu_{\rho\bx}(B)=\mu_{\bx}(\rho^{-1}B)
\end{equation}
for all $r>0$ and $\rho\in W$, and its support is given by
\begin{equation}
\supp (\mu_{\bx})=\co(W\bx).
\end{equation}
Here, $\co(W\bx)$ denotes the convex hull of the set $W\bx=\{\bz:{}^\exists \rho\in W, \bz=\rho \bx\}$.
\end{proposition}
The fact that $V_\beta g(\bx)$ is bounded, as remarked in \cite[p.~166]{dunklxu}, is a consequence of this theorem.
\begin{equation}
|V_\beta g(\bx)|\leq\int_{\RR^N}|g(\bxi)|\ud\mu_{\bx}(\bxi)\leq\sup_{\bxi\in\co(W\bx)}|g(\bxi)|,\label{EquationVBoundedness}
\end{equation}
for all $g(\bx)\in A_r$. This bound does not depend on $\beta>0$, which means that the limit $\lim_{\beta\to\infty}V_\beta g(\bx)=V_\infty g(\bx)$ is well-defined whenever $g(\bx)$ is bounded for finite $\bx$.

In the context of Dunkl processes, $V_\beta$ is of great importance because it deforms the BKE of a multidimensional free Brownian motion (the heat equation) into the BKE of a Dunkl process (the Dunkl heat equation) as follows. The transition probability density of a Brownian motion in $N$ dimensions is given by the heat kernel
\begin{equation}
p_\text{BM}(t,\by|\bx)=\frac{\rme^{-(\by-\bx)^2/2t}}{(2\pi t)^{N/2}},
\end{equation}
which in turn obeys the heat equation,
\begin{equation}
\Big(\frac{\partial}{\partial t} - \frac{1}{2}\Delta_x\Big)p_\text{BM}(t,\by|\bx)=0.
\end{equation}
Applying $V_\beta$ from the l.h.s. gives
\begin{equation}
V_\beta\Big(\frac{\partial}{\partial t} - \frac{1}{2}\Delta_x\Big)p_\text{BM}(t,\by|\bx)=\Big(\frac{\partial}{\partial t} - \frac{1}{2}\sum_{i=1}^NT_i^2\Big)V_\beta p_\text{BM}(t,\by|\bx)=0.
\end{equation}
This means that $V_\beta$ contains information on the distinctive features of Dunkl processes. Therefore, to understand the nature of Dunkl processes it is necessary to understand the behavior of $V_\beta$. Of particular importance is the effect of $V_\beta$ on the exponential function, also called the Dunkl kernel.

\section{The Dunkl kernel}

Define the Dunkl kernel by
\begin{equation}
E_\beta(\bx,\by):=V_\beta \rme^{\bx\cdot\by}.
\end{equation}
The Dunkl kernel is the analog of the exponential function for Dunkl operators in the following sense. By Eq.~\eqref{EquationVDefinition}, 
\begin{equation}
T_iE_\beta(\bx,\by)=T_iV_\beta \rme^{\bx\cdot\by}=V_\beta\Big(\frac{\partial}{\partial x_i} \rme^{\bx\cdot\by}\Big)=V_\beta(y_i\rme^{\bx\cdot\by})=y_iE_\beta(\bx,\by).
\end{equation}
Therefore, the action of a Dunkl operator on the Dunkl kernel is the same as the action of a partial derivative on the exponential of $\bx\cdot\by$. Note that $E_\beta(i\bx,\by)$ is bounded as
\begin{equation}\label{DunklKernelInTheTransform}
|V_\beta \rme^{\rmi \bx\cdot\by}|\leq \int_{\RR^N}|\rme^{\rmi \bxi\cdot\by}|\ud\mu_{\bx}(\bxi)=1,
\end{equation}
where $\mu_{\bx}(\bxi)$ is the measure specified in Prop.~\ref{PositivityVBeta}. Other useful properties of the Dunkl kernel are listed as follows; for $c\in\CCC$ and $\rho\in W$,
\begin{IEEEeqnarray}{rCl}
E_\beta(\bx,\by)&=&E_\beta(\by,\bx),\\
E_\beta(c\bx,\by)&=&E_\beta(\bx,c\by),\\
E_\beta(\rho\bx,\rho\by)&=&E_\beta(\bx,\by),\\
E_\beta(\bx,\by)^\dagger&=&E_\beta(\bx^\dagger,\by^\dagger),
\end{IEEEeqnarray}
where ${}^\dagger$ indicates complex conjugation. Finally, the following relation due to Dunkl (see \cite{rosler08}) is listed here for use in later chapters,
\begin{equation}\label{GaussianIntegralDunklKernel}
\frac{1}{c_\beta}\int_{\RR^N}V_\beta[\rme^{\bx\cdot\by}]V_\beta[\rme^{\bx\cdot\bz}]\rme^{-x^2/2}w_\beta(\bx)\ud\bx=\rme^{(y^2+z^2)/2}V_\beta\rme^{\by\cdot\bz}.
\end{equation}
This expression is in no way trivial, and it is also of great use in later chapters, particularly Chapter~\ref{general_steady}. The proof of Eq.~\eqref{GaussianIntegralDunklKernel} is given in Appendix~\ref{TheUsefulIntegral}. Because the Dunkl kernel with one imaginary argument is bounded, it is useful for the definition of the Dunkl generalization of the Fourier transform. 

\section{The Dunkl transform}

The Dunkl transform is a generalization of the Fourier transform given by the equation
\begin{equation}\label{DefinitionDunklTransform}
\hat{f}(\bxi):=\frac{1}{c_\beta}\int_{\RR^N}f(\bx)V_\beta\rme^{-\rmi \bxi\cdot\bx}w_\beta(\bx)\ud\bx,
\end{equation}
where the weight function $w_\beta$ and the normalization constant $c_\beta$ are given by
\begin{equation}
w_\beta(\bx):=\prod_{\balpha\in R_+}|\balpha\cdot\bx|^{\beta\kappa(\balpha)}
\end{equation}
and
\begin{equation}
c_\beta:=\int_{\RR^N}\rme^{-x^2/2}w_\beta(\bx)\ud\bx,
\end{equation}
respectively. This integral is known for many different cases, and it is known as a Selberg integral \cite{mehta04}. The inverse Dunkl transform is given almost everywhere by 
\begin{equation}\label{InverseDunklTransform}
f(\bx)\stackrel{\text{a.e.}}{=}[\hat{f}]\check{\phantom{1}}(\bx):=\frac{1}{c_\beta}\int_{\RR^N}\hat{f}(\bxi)V_\beta\rme^{\rmi \bxi\cdot\bx}w_\beta(\bxi)\ud\bxi
\end{equation}
and this equality holds for all points if $f(\bx)$ is continuous. The Dunkl transform is very similar to the Fourier transform in many respects, and it is particularly useful for solving the Dunkl heat equation. 

\section{Representations of the transition probability density}

A short derivation of the transition probability density of a Dunkl process will be given here. A more detailed and rigorous derivation is given in \cite{rosler99} and \cite{rosler08}. Consider the Fourier representation of the heat kernel,
\begin{equation}
p_\text{BM}(t,\by|\bx)=\frac{1}{(2\pi)^N}\int_{\RR^N}\rme^{-t\xi^2/2}\rme^{\rmi \by\cdot\bxi}\rme^{-\rmi \bx\cdot\bxi}\ud\bxi.
\end{equation}
This representation is, essentially, the inverse Fourier transform of the function $\rme^{-t\xi^2/2}\rme^{-\rmi \bx\cdot\bxi}$. In analogy with this formula, consider the function
\begin{equation}\label{EquationGreenFunctionGamma}
\Gamma(t,\bx,\by)=\frac{1}{c_\beta^2}\int_{\RR^N}\rme^{-t\xi^2/2}V_\beta [\rme^{\rmi \by\cdot\bxi}]V_\beta [\rme^{-\rmi \bx\cdot\bxi}] w_\beta(\bxi)\ud\bxi.
\end{equation}
By construction, this function solves the Dunkl heat equation,
\begin{multline}
\frac{\partial}{\partial t}\Gamma(t,\bx,\by)=\frac{-1}{c_\beta^2}\int_{\RR^N}\frac{\xi^2}{2}\rme^{-t\xi^2/2}V_\beta [\rme^{\rmi \by\cdot\bxi}]V_\beta [\rme^{-\rmi \bx\cdot\bxi}] w_\beta(\bxi)\ud\bxi\\
=\frac{1}{c_\beta^2}\int_{\RR^N}\rme^{-t\xi^2/2}V_\beta [\rme^{\rmi \by\cdot\bxi}]V_\beta \Big[\frac{1}{2}\Delta\rme^{-\rmi \bx\cdot\bxi}\Big] w_\beta(\bxi)\ud\bxi\\
=\frac{1}{2}\sum_{i=1}^NT_i^2\frac{1}{c_\beta^2}\int_{\RR^N}\rme^{-t\xi^2/2}V_\beta [\rme^{\rmi \by\cdot\bxi}]V_\beta [\rme^{-\rmi \bx\cdot\bxi}] w_\beta(\bxi)\ud\bxi\\
=\frac{1}{2}\sum_{i=1}^NT_i^2\Gamma(t,\bx,\by).
\end{multline}
Using Eq.~\eqref{GaussianIntegralDunklKernel} one obtains
\begin{equation}
\Gamma(t,\bx,\by)=\frac{\rme^{-(y^2+x^2)/2t}}{c_\beta t^{(\beta\gamma+N)/2}}V_\beta\rme^{\bx\cdot\by/t}.
\end{equation}

In general, $\Gamma(t,\bx,\by)$ is not normalized when it is integrated with respect to $\by$. However, $w_\beta(\by)\Gamma(t,\bx,\by)$ is normalized,
\begin{multline}
\int_{\RR^N}w_\beta(\by)\Gamma(t,\bx,\by)\ud\by=\frac{\rme^{-x^2/2t}}{c_\beta t^{(\beta\gamma+N)/2}}\int_{\RR^N}\rme^{-y^2/2t}V_\beta\rme^{\bx\cdot\by/t}w_\beta(\by)\ud\by\\
=\rme^{-x^2/2t}\rme^{x^2/2t}=1.
\end{multline}
Equation~\eqref{GaussianIntegralDunklKernel} was used to obtain the first equality in the second line. Thus, one obtains
\begin{equation}\label{TransitionDensityExplicit}
p(t,\by|\bx)=w_\beta(\by)\Gamma(t,\bx,\by)=w_\beta\left(\frac{\by}{\sqrt{t}}\right)\frac{\rme^{-(y^2+x^2)/2t}}{c_\beta t^{N/2}}V_\beta \exp \left(\frac{\bx\cdot\by}{t}\right).
\end{equation}
Equivalently, from Eq.~\eqref{EquationGreenFunctionGamma} one has
\begin{equation}\label{TransitionDensityTransform}
p(t,\by|\bx)=\frac{w_\beta(\by)}{c_\beta^2}\int_{\RR^N}\rme^{-t\xi^2/2}[V_\beta \rme^{\rmi \by\cdot\bxi}][V_\beta \rme^{-\rmi \bx\cdot\bxi}]w_\beta(\bxi)\ud\bxi.
\end{equation}
Therefore, the transition probability density of a radial Dunkl process is given by
\begin{equation}\label{TransitionDensityRadialDunkl}
P(t,\by|\bx)=w_\beta\left(\frac{\by}{\sqrt{t}}\right)\frac{\rme^{-(y^2+x^2)/2t}}{c_\beta t^{N/2}}\sum_{\rho\in W}V_\beta \exp \left(\frac{\rho\bx\cdot\by}{t}\right).
\end{equation}

The mathematical objects required to write down the transition probability densities of the interacting Brownian motions and Bessel processes are summarized in the following table. Care must be taken with the domain of definition of $P(t,\by|\bx)$, as this density is normalized to one if it is integrated over the Weyl chamber $C$, but it is normalized to $|W|$ when integrated over $\RR^N$.
\begin{table}[!ht]
\begin{center}
{\renewcommand{\arraystretch}{1.5}
\renewcommand{\tabcolsep}{0.2cm}
\begin{tabular}{ll}
\hline
Inter. Brownian Motions				&Interacting Bessel processes\\
\hline
$\displaystyle w_A(\bx)=\!\!\!\!\prod_{1\leq i<j\leq N}\!\!\!\!|x_j-x_i|^\beta$	&$\displaystyle w_B(\bx)=\prod_{i=1}^N|x_i|^{\beta(\nu+1/2)}\!\!\!\!\prod_{1\leq i<j\leq N}\!\!\!\!|x_j^2-x_i^2|^{\beta}$\\
$\gamma_A=N(N-1)/2$	&$\gamma_B=N(N+\nu-1/2)$\\
$\displaystyle c_A=\prod_{j=1}^N\frac{\sqrt{2\pi}\Gamma(1+j\frac{\beta}{2})}{\Gamma(1+\frac{\beta}{2})}$		&$\displaystyle c_B= 2^{\frac{\beta\gamma_B+N}{2}}\prod_{j=1}^N\frac{\Gamma(1+j\frac{\beta}{2})\Gamma[\frac{\beta}{2}(\nu+j-\frac{1}{2})+\frac{1}{2}]}{\Gamma(\frac{\beta}{2}+1)}$\\
$C_A=\{\bx: x_1<\ldots<x_N\}$&$C_B=\{\bx: 0<x_1<\ldots<x_N\}$\\
\hline
\end{tabular}}
\end{center}
\caption{Weight function $w_\beta(\bx)$, the sum of multiplicities $\gamma$, normalization constant $c_\beta$ and Weyl chamber $C$ for the interacting Brownian motions and Bessel processes. Here, $\bx\in\RR^N$, and $c_\beta$ is given by the Selberg integrals \cite[p.~321]{mehta04}.}
\label{RadialDunklProcessesQuantities}
\end{table}

% !TEX encoding = UTF-8 Unicode
\chapter{Correspondence with the Calogero-Moser systems}\label{CalogeroMoserCorrespondence}

The purpose of this chapter is to prove that there exists a correspondence between Dunkl processes and the Calogero-Moser (CM) systems on a given root system, and use that correspondence to obtain information concerning Dunkl processes in both the steady state and the freezing regime.

\section{Proof of the correspondence}

Under an arbitrary root system $R$, the CM systems on a line with a harmonic background potential and an inverse-square repulsion potential are given by the Hamiltonian (see, e.g., \cite{khastgir00})
\begin{equation}
\mathcal{H}_{\text{CM}}^R=-\frac{1}{2}\Delta^{(x)}+\frac{\beta}{2}\sum_{\balpha\in R_+}\frac{\alpha^2}{2}\frac{\kappa(\balpha)[\beta\kappa(\balpha)/2-\sigma_\alpha]}{(\balpha\cdot\bx)^2}+\frac{\omega^2}{2}\sum_{i=1}^Nx_i^2,\label{generalcm}
\end{equation}
where all the particles have unit mass, and $ \hbar=1$.

Dunkl operators have been used as a tool to prove the integrability of the CM systems \cite{forrester10}. It has been shown under several root systems \cite{rosler98} that after applying a similarity transformation (using the ground state eigenfunction), the Hamiltonian of the CM system is expressed as a Dunkl Laplacian plus a restoring term of the form $\bx\cdot\bnabla$. One can then find the polynomial eigenfunctions for the transformed Hamiltonian as stated in \cite{bakerdunklforrester} and shown in \cite{bakerforrester97}. The objective is to transform the FKE of a Dunkl process into the Schr{\"o}dinger equation of the CM systems. 

The diffusion-scaling transformation is defined as follows. In view of the transformation of a simple Brownian motion into a one-dimensional quantum harmonic oscillator in imaginary time proposed in \cite{katoritanemura07}, consider the substitution given by
\begin{equation}
(t,\by)\to(\tau,\bzeta)=\Bigg(\frac{\log t}{2\omega},\frac{\by}{\sqrt{2\omega t}}\Bigg).\label{substitutionr}
\end{equation}
Note that the spatial variable $\by$ is scaled by a factor of $\sqrt{t}$. Denote the density of the Dunkl process at a time $t$ for a given initial distribution by $f(t,\by)$. The diffusion-scaling transformation consists of performing the variable substitution 
\begin{equation}
f[t(\tau,\bzeta,\omega),\by(\tau,\bzeta,\omega)]=f(\tau,\bzeta), 
\end{equation}
and applying the similarity transformation
\begin{equation}
f(\tau,\bzeta)=\exp[-W(\tau,\bzeta)]U(\tau,\bzeta)\label{transformationr}
\end{equation}
with $W(\tau,\bzeta)$ given by
\begin{equation}
W(\tau,\bzeta)=\frac{1}{2}\omega\sum_{i=1}^N\zeta_i^2-\frac{\beta}{2}\sum_{\balpha\in R_+}\kappa(\balpha)\log |\balpha\cdot\bzeta|+\omega N\tau.\label{Wr}
\end{equation}
Because the scaling is isotropic, it is independent of the root system $R$.

\begin{proposition}\label{correspondencer}
The diffusion-scaling transformation given by \eqref{substitutionr}, \eqref{transformationr}, and \eqref{Wr} transforms the Dunkl process on the root system $R$ into the CM system with harmonic confinement on the same root system evolving in imaginary time.
\end{proposition}
\begin{proof}
Let us transform the KFE \eqref{DunklForward}. The derivatives in time and space in terms of the new variables are given by
\begin{eqnarray}
\frac{\partial}{\partial t}&=&\frac{1}{2\omega t}\frac{\partial}{\partial \tau}-\frac{1}{2t}\bzeta\cdot\bnabla^{(\zeta)},\nonumber\\
\frac{\partial}{\partial y_i}&=&\frac{1}{\sqrt{2\omega t}}\frac{\partial}{\partial \zeta_i}.\label{dertransformation1}
\end{eqnarray}
The differential operators that result from inserting the above in \eqref{DunklForward} are transformed by \eqref{transformationr} as follows:
\begin{eqnarray}
\rme^{W}\frac{\partial}{\partial \tau}\rme^{-W}&=&\frac{\partial}{\partial \tau}-\omega N,\nonumber\\
\rme^{W}\frac{\partial}{\partial \zeta_i}\rme^{-W}&=&\frac{\partial}{\partial \zeta_i}-\omega\zeta_i+\frac{\beta}{2}\sum_{\balpha\in R_+}\frac{\kappa(\balpha)}{\balpha\cdot\bzeta}\alpha_i,\nonumber\\
\rme^{W}\Delta^{(\zeta)}\rme^{-W}&=&\Delta^{(\zeta)}+2\Bigg(\frac{\beta}{2}\sum_{\balpha\in R_+}\frac{\kappa(\balpha)}{\balpha\cdot\bzeta}\balpha-\omega\bzeta\Bigg)\cdot\bnabla^{(\zeta)}+\omega^2\zeta^2\nonumber\\
&&-(\beta\gamma+N)\omega+\frac{\beta^2}{4}\sum_{\balpha\in R_+}\sum_{\bxi\in R_+}\frac{\kappa(\balpha)\kappa(\bxi)}{(\balpha\cdot\bzeta)(\bxi\cdot\bzeta)}\balpha\cdot\bxi\nonumber\\
&&-\frac{\beta}{2}\sum_{\balpha\in R_+}\frac{\kappa(\balpha)}{(\balpha\cdot\bzeta)^2}\alpha^2.\label{differentialtransformationsr}
\end{eqnarray}
Therefore, inserting \eqref{dertransformation1} and \eqref{differentialtransformationsr} successively in \eqref{DunklForward} yields
\begin{eqnarray}
\frac{\partial}{\partial \tau}U(\tau,\bzeta)&=&\frac{1}{2}\Delta^{(\zeta)}U(\tau,\bzeta)+\frac{\omega}{2}[\beta\gamma+N-\omega\zeta^2]U(\tau,\bzeta)\nonumber\\
&&+\frac{\beta}{2}\sum_{\balpha\in R_+}\frac{\alpha^2}{2}\frac{\kappa(\balpha)}{(\balpha\cdot\bzeta)^2}U(\tau,\sigma_{\balpha}\bzeta)\nonumber\\
&&\ -\frac{\beta^2}{4}\sum_{\balpha\in R_+}\sum_{\bxi\in R_+}\frac{\balpha\cdot\bxi}{2}\frac{\kappa(\balpha)\kappa(\bxi)}{(\balpha\cdot\bzeta)(\bxi\cdot\bzeta)}U(\tau,\bzeta).\label{withdoublesum}
\end{eqnarray}
The double sum in the bottom term of the equation above can be simplified because all the terms where $\balpha\neq\bxi$ cancel each other (see Lemma~4.4.6 of \cite{dunklxu}). By denoting the ground-state energy by $E_{\text{0}}^{R}=\omega(\beta\gamma+N)/2$ and using $\HH{CM}^{R}$ with $\bzeta$ instead of $\bx$, we finally obtain
\begin{equation}
-\frac{\partial}{\partial \tau}U(\tau,\bzeta)=[\HH{CM}^{R}-E_{\text{0}}^{R}]U(\tau,\bzeta),\label{done}
\end{equation}
as desired.
\end{proof}

{\it Remark:} this proof involves only straightforward calculations, with the notable exception of the step required to simplify the double sum in \eqref{withdoublesum}. This is perhaps the most important part of the proof, and it is not trivial. The simplest case is when $R$ is the root system of type $A$ (see, e.g., \cite{forrester10}, Proposition~11.3.1). Note also that Proposition~\ref{correspondencer} only requires that $\omega>0$. If $\omega=0$, there is no need to use the diffusion scaling \eqref{substitutionr}, and one may simply apply a similarity transformation on the Dunkl process to obtain the unconfined CM system on the same root system. 

\section{The steady state}

Having established Prop.~\ref{correspondencer}, the time evolution of the function $U(\tau,\bzeta)$ given by Eq.~\eqref{done} implies that after a long time, only the ground-state eigenfunction survives. This means that $f(\tau,\bzeta)$, as defined in Eq.~\eqref{transformationr}, must converge to a non-trivial form as $\tau\to\infty$; as detailed in Chapter~\ref{general_steady}, this is the sense in which the steady state is defined in the present context. More precisely, suppose that the eigenfunctions of $\HH{CM}^{R}$ are denoted by $\{\psi_\eta(\bzeta)\}_\eta$, and that their corresponding eigenvalues are denoted by $\{E_\eta^R\}_\eta$, where $\eta$ is a discrete multi-index (see \cite{khastgir00}). Furthermore, assume that the set of eigenfunctions $\{\psi_\eta(\bzeta)\}_\eta$ is a complete basis of the Hilbert space of this system. Therefore, the following expression holds in general:
\begin{equation}
U(\tau,\bzeta)=\sum_\eta \rme^{-\tau[E_\eta^R-E_0^R]}\psi_\eta(\bzeta).
\end{equation}
Because $E_\eta^R\geq E_0^R$, after a sufficiently long time $\tau$ the function $U(\tau,\bzeta)$ reduces to the ground-state wavefunction $\psi_0(\bzeta)$, given by
\begin{equation}
\psi_0(\bzeta)=a_0\rme^{-\omega \zeta^2/2}\prod_{\balpha\in R_+}|\balpha\cdot\bzeta|^{\beta\kappa(\balpha)/2},
\end{equation}
where $a_0$ is a normalization constant. In Chapter~\ref{general_steady}, the scaled distribution
\begin{equation}
f(t,\sqrt{\beta t}\bv)(\beta t)^{N/2}
\end{equation}
will be considered. This is equivalent to the diffusion-scaling transformation with $\omega = \beta$, $\tau=(\log t)/2\beta$ and $\bzeta=\bv/\sqrt{2}$. It follows that
\begin{multline}
f(t,\sqrt{\beta t}\bv)(\beta t)^{N/2}=\beta^{N/2}U\Big(\frac{\log t}{2\beta},\frac{\bv}{\sqrt{2}}\Big)\\
\times\exp\Big[-\frac{\beta v^2}{4}+\frac{\beta}{2}\sum_{\balpha\in R_+}\kappa(\balpha)\log |\balpha\cdot\bv/\sqrt{2}|\Big]\\
\stackrel{t\to\infty}{\longrightarrow}\beta^{N/2}\exp\Big[-\frac{\beta v^2}{4}+\frac{\beta}{2}\sum_{\balpha\in R_+}\kappa(\balpha)\log |\balpha\cdot\bv/\sqrt{2}|\Big]\psi_0(\bv/\sqrt{2})\\
=\frac{a_0\beta^{N/2}}{2^{\beta\gamma/2}}\exp\Big[-\beta\Big(\frac{v^2}{2}-\sum_{\balpha\in R_+}\kappa(\balpha)\log |\balpha\cdot\bv|\Big)\Big].\label{SteadyStateDistributionPreliminary}
\end{multline}
This means that the final distribution of a Dunkl process converges to a steady-state form if the final position $\by$ is rescaled as $\sqrt{\beta t}\bv$. Henceforth, the phrase \emph{steady state of a Dunkl process} (or interacting Brownian motion or Bessel process) will refer to the steady state of the \emph{scaled process}. The convergence to the steady state in finite time will be proved rigorously in Chapter~\ref{general_steady} without the help of the quantum mechanics of the Calogero-Moser systems, and using Dunkl operator theory.

\section{The freezing regime, peak sets and Fekete points}

It is known that the Calogero-Moser systems form spin chains at very low temperatures \cite{polychronakos92, polychronakos93, frahm93, yamamototsuchiya96}. These spin chains consist of a series of particles that are fixed in space and that interact through exchange operators that are defined using the reflection operators $\sigma_{\balpha}$. The main idea is to consider the Hamiltonian $\mathcal{H}_{\text{CM}}^R$ in the limit where $\beta\to\infty$. As in the previous section, setting $\omega=\beta$ gives the following leading-order terms in $\beta$:
\begin{equation}\label{CMLeadingPotential}
V_{\text{CM}}(\bx)=\frac{1}{4}\sum_{\balpha\in R_+}\frac{\alpha^2}{2}\frac{\kappa^2(\balpha)}{(\balpha\cdot\bx)^2}+\frac{1}{2}\sum_{i=1}^Nx_i^2.
\end{equation}
As $\beta\to\infty$, the kinetic energy term of the Hamiltonian becomes negligible, and the particles of the Calogero-Moser system freeze at the minima of the potential~\eqref{CMLeadingPotential}, which is clearly positive and convex. Consider now the argument of the exponential in the last line of Eq.~\eqref{SteadyStateDistributionPreliminary},

\begin{equation}
F_{\text{DP}}(\bv)=\frac{v^2}{2}-\sum_{\balpha\in R_+}\kappa(\balpha)\log |\balpha\cdot\bv|.
\end{equation}
A straightforward calculation yields
\begin{multline}
|\bnabla^{(v)}F_{\text{DP}}|^2=v^2-2\gamma+\sum_{\balpha\in R_+}\frac{\alpha^2\kappa^2(\balpha)}{(\balpha\cdot\bv)^2}=4V_{\text{CM}}(\bv/\sqrt{2})-2\gamma.
\end{multline}
This expression is obtained through the use of Lemma~4.4.6 in \cite{dunklxu}. Taking the gradient gives
\begin{equation}
(\bnabla^{(v)}F_{\text{DP}}\cdot\bnabla^{(v)})\bnabla^{(v)}F_{\text{DP}}=2\bnabla^{(v)}[V_{\text{CM}}(\bv/\sqrt{2})].
\end{equation}
From this relation it is deduced that, if $F_{\text{DP}}(\bv)$ achieves a minimum at $\bv=\bz$, then $V_{\text{CM}}(\bv)$ achieves a minimum at $\bv=\bz/\sqrt{2}$ (the fact that $F_{\text{DP}}(\bv)$ is a convex function will be proved in Chapter~\ref{general_freezing}.) In view of Eq.~\eqref{SteadyStateDistributionPreliminary}, it is expected that as $\beta\to\infty$ the steady state distribution of a Dunkl process converges to a series of delta functions located at the minima of $F_{\text{DP}}(\bv)$. In particular, it is known that the Calogero-Moser system of type $A$ (resp. type $B$) freezes to the roots of the Hermite polynomials \cite{frahm93} (resp. Laguerre polynomials \cite{yamamototsuchiya96}). Consequently, the interacting Brownian motions (resp. interacting Bessel processes) must freeze at these points as well.

The location at which Dunkl processes freeze is determined, then, by the solutions of the equation
\begin{equation}\label{EquationSolvedByThePeakSet}
\bv=\sum_{\balpha\in R_+}\frac{\kappa(\balpha)}{\balpha\cdot\bv}\balpha.
\end{equation}
The set of vectors $\{\bos_i\}_i$ which satisfy this equation is known as the peak set of the root system $R$ \cite{dunkl89B}. Originally, the peak set was defined as the set of vectors of unit norm that maximizes the function
\begin{equation}
\prod_{\balpha\in R_+}|\balpha\cdot\bos|^{\kappa(\balpha)}.
\end{equation}
In the case of Eq.~\eqref{EquationSolvedByThePeakSet}, the norm of the vectors of the peak set is $\sqrt{\gamma}$. Intuitively, the peak set represents the set of directions in the $N$-dimensional space where the Dunkl process is most likely to be found. More concretely, for the root systems of types $A$ and $B$ the peak sets are known as Fekete points \cite[p.~132]{deift00}. These represent the set of points on the real line where a system of $N$ charged particles must be located in order to minimize its potential energy. The particles interact with each other through a logarithmic potential and with a background potential, $(N-1)Q(x)$. Concretely, the Fekete points maximize the function
\begin{equation}
\prod_{1\leq i<j\leq N}(x_j-x_i)^2\rme^{-(N-1)\sum_{i=1}^N Q(x_i)},
\end{equation}
where $(N-1)Q(x)=x^2$ for the interacting Brownian motions. In the case of the interacting Bessel processes, $(N-1)Q(y)=y-(\nu+\frac{1}{2})\ln y$ with $y_i=x_i^2$ for every $i$. This electrostatic analogy is well-known, and it gives a physically meaningful interpretation to the peak set of $R$ (see \cite[p.~366-369]{szego} and \cite{fekete1923}).

Note that the derivations in this chapter are not rigorous, but they provide information on the behavior of Dunkl processes in the steady state and in the freezing limit. In particular, taking the freezing limit of the final expression in Eq.~\eqref{SteadyStateDistributionPreliminary} gives rise to inconsistencies in the time scale of the process, because the time $\tau=(\log t)/2\beta$ should not reach infinity if $\beta$ is infinitely large. However, the results are correct, and they will be proved for Dunkl processes in general in Chapters~\ref{general_steady} and \ref{general_freezing}.

\chapter{Steady-state in an arbitrary root system}\label{general_steady}

The objective of this chapter is to derive one of the main results of this work, which is the convergence of Dunkl processes in finite time to a steady state. For this purpose, the precise definition of the steady state of Dunkl processes is given as the steady state of their corresponding \emph{time-scaled} processes. Next, the precise statement of the result is given in Theorem~\ref{TheoremSteadyState}, and the proof of the theorem is given in the final two sections.

\section{Definition of the steady state}

Dunkl processes in general, and interacting Brownian motions and Bessel processes in particular, are diffusion processes without restoring forces. This means that their evolution is such that their probability density spreads out in space without stopping. Consequently, these processes do not have a steady state in the strict sense. However, as mentioned in the previous chapter, if the distribution function of the Dunkl process is scaled suitably, then it has a limit form (in time) given by Eq.~\eqref{SteadyStateDistributionPreliminary}, and the underlying scaled process has a steady state.

The concrete form of these statements is as follows. The SDE of a radial Dunkl process is given by \cite{demni09A}
\begin{equation}\label{RadialDunklSDE}
\ud \bX_t=\ud \bB_t+\frac{\beta}{2}\sum_{\balpha\in R_+}\frac{\kappa(\balpha)\balpha}{\balpha\cdot\bX_t}\ud t.
\end{equation}
This is a semi-martingale in $N$ dimensions with a drift term that forces the process away from the origin in the directions $\balpha\in R$. Clearly, this process does not achieve a steady state, as it diffuses without bounds. Its corresponding time-scaled process is defined by $\bY_t=\bX_t/\sqrt{\beta t}$. Using It{\^o}'s formula \cite{protter05}, the SDE for $\bY_t$ is 
\begin{equation}
\ud \bY_t=\frac{\ud \bB_t}{\sqrt{\beta t}}+\frac{1}{2t}\sum_{\balpha\in R_+}\frac{\kappa(\balpha)\balpha}{\balpha\cdot\bY_t}\ud t-\frac{\bY_t}{2t}\ud t.
\end{equation}
The last term is a restoring force term, which restrains the diffusion of the process. If, furthermore, the time variable is redefined as $\tau=(\log t)/\beta$, then $\ud t=\beta t\ud \tau$ and
\begin{equation}
\ud \bY_\tau=\ud \bB_\tau+\frac{\beta}{2}\Big[\sum_{\balpha\in R_+}\frac{\kappa(\balpha)\balpha}{\balpha\cdot\bY_\tau}-\bY_\tau\Big]\ud \tau.
\end{equation}
This is a process with a harmonic restoring force, which is very similar to an Ornstein-Uhlenbeck processes. It is known that these generalized Ornstein-Uhlenbeck processes are stationary \cite[Sec.~10]{roslervoit98}, so it is clear that a properly scaled radial Dunkl process achieves a steady state. Furthermore, because the jumps of non-radial Dunkl processes preserve the distance of the process to the origin just before and after the jump (because $\sigma_{\balpha}$ is an isometry), then non-radial Dunkl processes scaled by a factor $\sqrt{\beta t}$ are stationary as well. 

From the point of view of the process distribution $f(t,\by)$, the procedure is as follows. According to Eq.~\eqref{SteadyStateDistributionPreliminary}, if the probability distribution $f(t,\by)$ of a Dunkl process is scaled as
\begin{equation}
f_t(\bv):=f(t,\sqrt{\beta t}\bv)(\beta t)^{N/2},
\end{equation}
then the steady-state distribution is proportional to
\begin{equation}
f_{\text{ss}}(\bv):=\exp\Big[-\beta\Big(\frac{v^2}{2}-\sum_{\balpha\in R_+}\kappa(\balpha)\log |\balpha\cdot\bv|\Big)\Big].
\end{equation}
Let us verify this claim. Because the Dunkl process density $f(t,\by)$ obeys the FKE \eqref{DunklForward}, then the scaled distribution obeys the following FKE:
\begin{multline}\label{ScaledDunklFKE}
2t\frac{\partial}{\partial t}f_t(\bv)=\frac{1}{\beta}\Delta f_t(\bv)-\sum_{\balpha\in R_+}\kappa(\balpha)\Big[\frac{\balpha\cdot\bnabla f_t(\bv)}{\balpha\cdot\bv}-\frac{\alpha^2}{2}\frac{(1+\sigma_{\balpha})f_t(\bv)}{(\balpha\cdot\bv)^2}\Big]\\
+\bv\cdot\bnabla f_t(\bv)+Nf_t(\bv).
\end{multline}
This is because Dunkl operators behave like partial derivatives under uniform scalings, i.e., if an arbitrary function $g(\by)$ is scaled as $\tilde{g}(\bv)=g(\sqrt{\beta t}\bv)$, then
\begin{equation}
T_i^{(y)}g(\by)=\frac{1}{\sqrt{\beta t}}T_i^{(v)}\tilde{g}(\bv),
\end{equation}
and also because the time derivative of $f_t(\bv)$ is given by
\begin{multline}
\frac{\partial}{\partial t}f_t(\bv)=\frac{\partial}{\partial t}\Big[f(t,\sqrt{\beta t}\bv)(\beta t)^{N/2}\Big]\\
=\frac{N}{2t}f_t(\bv)+(\beta t)^{N/2}\frac{\partial}{\partial t}f(t,\by)\Big|_{\by=\sqrt{\beta t}\bv}+\frac{1}{2t}\bv\cdot\bnabla f_t(\bv),
\end{multline}
meaning that
\begin{equation}
\frac{\partial}{\partial t}f_t(\bv)=\frac{N}{2t}f_t(\bv)+\frac{1}{2t}\bv\cdot\bnabla f_t(\bv)+\frac{1}{2\beta t}\sum_{i=1}^NT_i^2f_t(\bv).
\end{equation}
Inserting the explicit form of the Dunkl Laplacian into the r.h.s. and moving the factor $2t$ to the l.h.s. gives Eq.~\eqref{ScaledDunklFKE}.

Let us show that $f_\text{ss}(\bv)$ is the steady-state solution of Eq.~\eqref{ScaledDunklFKE}, which amounts to showing that the r.h.s. of Eq.~\eqref{ScaledDunklFKE} vanishes when $f_t(\bv)$ is replaced by $f_\text{ss}(\bv)$. A series of straightforward calculations yield the following expressions:
\begin{IEEEeqnarray}{rCl}
\bv\cdot\bnabla f_\text{ss}(\bv)&=&-\beta(v^2-\gamma)f_\text{ss}(\bv),\\
\balpha\cdot\bnabla f_\text{ss}(\bv)&=&-\beta\Big[\balpha\cdot\bv-\sum_{\bzeta\in R_+}\frac{\kappa(\bzeta)\bzeta\cdot\balpha}{\bzeta\cdot\bv}\Big]f_\text{ss}(\bv),\\
\frac{\partial^2 f_\text{ss}}{\partial v_j^2}&=&\beta^2\Big[v_j-\sum_{\balpha\in R_+}\frac{\kappa(\balpha)\alpha_j }{\balpha\cdot\bv}\Big]^2f_\text{ss}(\bv)\nonumber\\
&&\qquad-\beta\Big[1+\sum_{\balpha\in R_+}\frac{\kappa(\balpha)\alpha_j^2 }{(\balpha\cdot\bv)^2}\Big]f_\text{ss}(\bv).
\end{IEEEeqnarray}
Then, the r.h.s. of Eq.~\eqref{ScaledDunklFKE} for $f_\text{ss}(\bv)$ is equal to
\begin{multline}
\beta\Big[v^2-2\gamma+\sum_{\balpha,\bzeta\in R_+}\frac{\kappa(\balpha)\kappa(\bzeta)\balpha\cdot\bzeta}{(\balpha\cdot\bv)(\bzeta\cdot\bv)}\Big]f_\text{ss}(\bv)-\Big[N+\sum_{\balpha\in R_+}\frac{\kappa(\balpha)\alpha^2 }{(\balpha\cdot\bv)^2}\Big]f_\text{ss}(\bv)\\
-\sum_{\balpha\in R_+}\kappa(\balpha)\Big[-\beta+\beta\sum_{\bzeta\in R_+}\frac{\kappa(\bzeta)\bzeta\cdot\balpha}{(\balpha\cdot\bv)(\bzeta\cdot\bv)}-\frac{\alpha^2}{(\balpha\cdot\bv)^2}\Big]f_\text{ss}(\bv)\\
-\beta(v^2-\gamma)f_\text{ss}(\bv) +Nf_\text{ss}(\bv).
\end{multline}
A close inspection of this expression reveals that, indeed, all terms cancel, meaning that the steady-state distribution of the process is obtained from normalizing $f_\text{ss}(\bv)$.

\section{Setting}

With the initial condition $\mu(\bx)$, the probability distribution of a Dunkl process using the Dunkl transform representation~\eqref{TransitionDensityTransform} is given by
\begin{equation}\label{ProbabilityDensityTransform}
f(t,\by)\ud\by=\frac{w_\beta(\by)}{c_\beta^2}\int_{\RR^N}\rme^{-t\xi^2/2}[V_\beta \rme^{\rmi \by\cdot\bxi}]\Big[\int_{\RR^N}[V_\beta \rme^{-\rmi \bx\cdot\bxi}]\mu(\bx)\ud\bx\Big] w_\beta(\bxi)\ud\bxi\ud\by.
\end{equation}
The main advantage of this expression is that all the integrals are well behaved due to the fact that $V_\beta \rme^{\rmi \bx\cdot\by}$ is bounded (see Eq.~\eqref{DunklKernelInTheTransform}). Recall the function
\begin{equation}\label{PotentialR}
F_R(\bv,\kappa)=\frac{v^2}{2}-\sum_{\balpha\in R_+} \kappa(\balpha)\log |\balpha\cdot\bv|,
\end{equation}
so that
\begin{equation}
f_\text{ss}(\bv)=\rme^{-\beta F_R(\bv,\kappa)}=\rme^{-\beta v^2/2}w_\beta(\bv)
\end{equation}
and
\begin{equation}\label{PartitionR}
z_\beta=\int_{\RR^N}\rme^{-\beta F_R(\bzeta,\kappa)}\ud\bzeta=\int_{\RR^N}f_\text{ss}(\bzeta)\ud\bzeta.
\end{equation}
With these expressions, we define the steady-state distribution by
\begin{equation}
f_R(\bv):=\frac{1}{z_\beta}f_\text{ss}(\bv)=\frac{1}{z_\beta}\rme^{-\beta F_R(\bv,\kappa)}.
\end{equation}
Also, assume that $\mu(\bx)$ is a probability distribution with finite second-order moments, i.e.,
\begin{equation}
\Big|\int_{\RR^N}x_i x_j \mu(\bx)\ud\bx\Big|<\infty.
\end{equation}
The mean and variance of $\mu(\bx)$ are defined by
\begin{IEEEeqnarray}{rCl}
\bar{\bx}_\mu&:=&\int_{\RR^N}\bx\mu(\bx)\ud\bx,\label{MuFirstMoments}\\
s_\mu^2&:=&\int_{\RR^N}|\bx-\bar{\bx}_\mu|^2\mu(\bx)\ud\bx\label{MuSecondMoment}
\end{IEEEeqnarray}
respectively. In this chapter, the main goal is to prove the following.
\begin{theorem}\label{TheoremSteadyState}
Dunkl processes relax to the scaled steady state distribution
\begin{equation}
f(t,\sqrt{\beta t}\bv)(\beta t)^{N/2}\ud\bv=f_R(\bv,\beta)\ud\bv[1+O(\eta\sqrt{\beta})+O(\epsilon^2)]
\end{equation}
whenever
\begin{equation}\label{GeneralRelaxationTime}
t\gg (s_\mu^2+\bar{x}_\mu^2)\max[1,\beta]
\end{equation}
with positive numbers $\eta$ and $\epsilon$ such that 
\begin{equation}\label{EquationFinalConditionsEpsilonEta}
\epsilon^2t\gg1\quad\text{and}\quad\eta^2 t \min[1,\beta]\gg 1.
\end{equation}
\end{theorem}

In Sec.~\ref{ProofTheoremSteadyState} we prove Thm.~\ref{TheoremSteadyState}, in which the evaluation of the integrals in Eq.~\eqref{ProbabilityDensityTransform} is important. The derivations given in Sec.~\ref{IntegralOverU} are used for the calculations in the proof.

\section{Setup for the proof of Theorem~\ref{TheoremSteadyState}}\label{IntegralOverU}

With the substitutions $\by=\sqrt{\beta t}\bv$, $\bxi=\bzeta/\sqrt{t}$ and $\bx=\bu\sqrt{t}$, Eq.~\eqref{ProbabilityDensityTransform} reads
\begin{multline}\label{ScaledProbabilityDensityTransform}
f(t,\sqrt{\beta t}\bv)(\beta t)^{N/2}\ud\bv=\frac{w_\beta(\bv)}{z_\beta c_\beta}\int_{\RR^N}\rme^{-\zeta^2/2}[V_\beta \rme^{\rmi \sqrt{\beta}\bv\cdot\bzeta}]\\
\times\Big[\int_{\RR^N}[V_\beta \rme^{-\rmi \bu\cdot\bzeta}]t^{N/2}\mu(\sqrt{t}\bu)\ud\bu\Big] w_\beta(\bzeta)\ud\bzeta\ud\bv.
\end{multline}
The first quantity that must be considered is the integral over $\bu$, which is given the following notation,
\begin{equation}
\mathcal{I}(t,\bzeta):=\int_{\RR^N}[V_\beta \rme^{-\rmi \bu\cdot\bzeta}]t^{N/2}\mu(\sqrt{t}\bu)\ud\bu.
\end{equation}
By the mean value theorem, there exist vectors $\bu_\text{r}$ and $\bu_\text{i}$ such that 
\begin{equation}\label{FirstApproximationI1}
\mathcal{I}(t,\bzeta)=V_\beta[\cos(\bu_\text{r}\cdot\bzeta)-\rmi\sin(\bu_\text{i}\cdot\bzeta)].
\end{equation}
In general, $\bu_\text{r}$ and $\bu_\text{i}$ are functions of $\bzeta$. From the relations
\begin{IEEEeqnarray}{rCl}
&\int_{\RR^N}\bu t^{N/2}\mu(\sqrt{t}\bu)\ud\bu=\bar{\bx}_\mu/\sqrt{t},&\\
&\int_{\RR^N}(\bu-\bar{\bx}_\mu/\sqrt{t})^2t^{N/2}\mu(\sqrt{t}\bu)\ud\bu=s_\mu^2/t,&
\end{IEEEeqnarray}
it follows that $\mathcal{I}(t,\bzeta)$ satisfies the following equations when $\bzeta=\bzero$,
\begin{IEEEeqnarray}{rCl}
T_j\mathcal{I}(t,\bzero)&=&-\rmi\frac{\bar{x}_{\mu,j}}{\sqrt{t}},\\
\sum_{j=1}^NT_j^2\mathcal{I}(t,\bzero)&=&-\frac{1}{t}(s_\mu^2+\bar{x}_\mu^2).
\end{IEEEeqnarray}
Inserting Eq.~\eqref{FirstApproximationI1} in these two equations gives
\begin{IEEEeqnarray}{rCl}
T_jV_\beta[\cos(\bu_\text{r}\cdot\bzeta)-\rmi\sin(\bu_\text{i}\cdot\bzeta)]|_{\bzeta=\bzero}&=&-\rmi (\bu_\text{i})_j|_{\bzeta=\bzero},\\
\sum_{j=1}^NT_j^2V_\beta[\cos(\bu_\text{r}\cdot\bzeta)-\rmi\sin(\bu_\text{i}\cdot\bzeta)]|_{\bzeta=\bzero}&=&-u_\text{r}^2|_{\bzeta=\bzero}.
\end{IEEEeqnarray}
Therefore, one has
\begin{IEEEeqnarray}{rCl}
\bu_\text{I}&:=&\bu_\text{i}|_{\bzeta=\bzero}=\bar{\bx}_\mu/\sqrt{t},\label{EquationForUI}\\
u_\text{R}^2&:=&u_\text{r}^2|_{\bzeta=\bzero}=(s_\mu^2+\bar{x}_\mu^2)/t.\label{EquationForUR}
\end{IEEEeqnarray}
Thus, for $\zeta^2\ll t/(s_\mu^2+\bar{x}_\mu^2)$, the integral $\mathcal{I}(t,\bzeta)$ is approximated by
\begin{equation}\label{EquationApproximationForI}
V_\beta[\cos(\bu_\text{R}\cdot\bzeta)-\rmi\sin(\bu_\text{I}\cdot\bzeta)]+O(\zeta^2/t).
\end{equation}
The term of $O[\zeta^2/t]$ is included because this approximation is obtained from the behavior of $\mathcal{I}(t,\bzeta)$ near $\bzeta=\bzero$, and only the action of first-order Dunkl operators is reproduced exactly (only the norm of $\bu_\text{R}$ is specified).

For the estimation of the relaxation time $t_\epsilon$, it will be necessary to consider the linear approximation of the function $V_\beta \sinh(\bx\cdot\by)$. The effect of the intertwining operator on linear functions is given in the following lemma. Here, an orthonormal basis $\{\bphi_i\}_{i=1}^N$ of $\RR^N$ is defined such that the first $d_R$ vectors belong to the linear envelope of $R$ and the last $N-d_R$ vectors are orthogonal to it.
\begin{lemma}\label{LemmaFirstOrderV}
For linear polynomials, $V_\beta$ is given by
\begin{equation}
V_\beta\bx\cdot\by=\frac{1}{1+\beta\gamma/d_R}\Big[\bx\cdot\by+\frac{\beta\gamma}{d_R}\sum_{i=d_R+1}^{N}(\bx\cdot\bphi_i)(\by\cdot\bphi_i)\Big].\label{FirstOrderVNotFullRank}
\end{equation}
\end{lemma}

\begin{proof}
Let $V_\beta$ be represented in the linear case by a matrix, $[M_\beta]_{ij}=m_{ij}$, such that
\begin{equation}
V_\beta\bx\cdot\by=\sum_{1\leq i,j\leq N}x_im_{ij}y_j.
\end{equation}
By the defining property of $V_\beta$, Eq.~\eqref{EquationVDefinition}, the relationship
\begin{equation}
y_i=\sum_{j=1}^N m_{ij}y_j+\frac{\beta}{2}\sum_{\balpha\in R_+}\alpha_i\frac{\kappa(\balpha)}{\balpha\cdot\bx}(1-\sigma_{\balpha})\sum_{1\leq i,j\leq N}x_lm_{lj}y_j
\end{equation}
holds. Rewritten in terms of vector and matrices, this equation is equivalent to
\begin{equation}
(1-\sigma_{\balpha})\bx^TM_\beta \by=\Big(2\frac{\balpha\cdot\bx}{\alpha^2}\balpha\Big)^T M_\beta \by=2\frac{\balpha\cdot\bx}{\alpha^2}\balpha^TM_\beta\by,
\end{equation}
which yields
\begin{equation}
\by=M_\beta\by+\beta\sum_{\balpha\in R_+}\kappa(\balpha)\frac{\balpha\balpha^T}{\alpha^2}M_\beta\by=\Bigg[I+\beta\sum_{\balpha\in R_+}\kappa(\balpha)\frac{\balpha\balpha^T}{\alpha^2}\Bigg]M_\beta \by.\label{FirstOrderV}
\end{equation}
The problem is reduced, then to calculating 
\begin{equation}\label{AlmostTheMWeWant}
M_\beta=\Bigg[I+\beta\sum_{\balpha\in R_+}\kappa(\balpha)\frac{\balpha\balpha^T}{\alpha^2}\Bigg]^{-1}.
\end{equation}

To calculate $M_\beta$, we need to calculate the sum over $\balpha$, a task that can be accomplished easily once the sum is rewritten as a sum over a group. If $\by$ is orthogonal to $\spn (R)$, then the sum over $\balpha$ is equal to zero and the result is trivial; therefore, we assume that $\by\in\spn(R)$. First, we separate terms with different multiplicities: let $n_R$ denote the number of different multiplicities assigned to the root system R, and denote by $\{\bxi_i\}_{i=1}^{n_R}$ a collection of roots such that $\kappa(\bxi_i)\neq\kappa(\bxi_j)$ for $i\neq j$. Then, the following holds:
\begin{equation}
\sum_{\balpha\in R_+}\kappa(\balpha)\frac{\balpha\balpha^T}{\alpha^2}=\sum_{i=1}^{n_R}\frac{\kappa(\bxi_i)}{|\bxi_i|^2}\sum_{\balpha\in R_+\cap W\bxi_i}\balpha\balpha^T.
\end{equation}
Here, $W\bxi_i=\{\rho\bxi_i: \rho\in W\}$ denotes the orbit of $\bxi_i$ on $W$. To obtain this equality it suffices to note that the terms in the sum with equal multiplicities must belong to the same orbit in $W$. Because $W$ is a reflection group, all of its elements are isometries, meaning that the roots of $R$ in the orbit of $\bxi_i$ must have the same norm. Then, the squared norm of $\alpha$ can be taken out of the second sum as the term $|\bxi_i|^2$. Because the vectors $\balpha$ in the second sum are elements of $W\bxi_i$, we can rewrite the sum as
\begin{equation}
\sum_{\balpha\in R_+}\kappa(\balpha)\frac{\balpha\balpha^T}{\alpha^2}=\sum_{i=1}^{n_R}\frac{\kappa(\bxi_i)}{|\bxi_i|^2}\frac{|R_+\cap W\bxi_i|}{|W|}\sum_{\rho\in W}(\rho\bxi_i)(\rho\bxi_i)^T,\label{ProjectorSumUnsolved}
\end{equation}
where the coefficient $|R_+\cap W\bxi_i|/|W|$ is included to account for double counting in the sum over $\rho$. Denote by $[\rho]_{ij}$ the $ij$th component of a faithful and reduced representation of the reflection group $W$. Then, the $jl$th component of the matrix representation of the sum over $\rho$ is
\begin{equation}
\sum_{\rho\in W}[\rho\bxi_i]_j[\rho\bxi_i]_l=\sum_{\rho\in W}\sum_{n,n^\prime=1}^{d_R}[\rho]_{jn}[\bxi_i]_n[\rho]_{ln^\prime}[\bxi_i]_{n^\prime}=\sum_{n,n^\prime=1}^{d_R}[\bxi_i]_{n^\prime}[\bxi_i]_n\sum_{\rho\in W}[\rho]_{jn}[\rho]_{ln^\prime}.
\end{equation}
By the great orthogonality theorem \cite{inuitanabeonodera96},
the sum over $\rho$ is given by
\begin{equation}
\sum_{\rho\in W}[\rho]_{jn}[\rho]_{ln^\prime}=\frac{|W|}{d_R}\delta_{jl}\delta_{nn^\prime},
\end{equation}
which in turn yields
\begin{equation}
\sum_{\rho\in W}[\rho\bxi_i]_j[\rho\bxi_i]_l=\sum_{n,n^\prime=1}^{d_R}[\bxi_i]_{n^\prime}[\bxi_i]_n\frac{|W|}{d_R}\delta_{jl}\delta_{nn^\prime}=\frac{|W||\bxi_i|^2}{d_R}\delta_{jl}.
\end{equation}
Inserting in \eqref{ProjectorSumUnsolved} gives
\begin{equation}
\sum_{\balpha\in R_+}\kappa(\balpha)\frac{\balpha\balpha^T}{\alpha^2}=\sum_{i=1}^{n_R}\kappa(\bxi_i)\frac{|R_+\cap W\bxi_i|}{d_R}I=\frac{I}{d_R}\sum_{\balpha\in R_+}\kappa(\balpha)=\frac{\gamma}{d_R}I.
\end{equation}
This result implies that, for $\by\in\spn(R)$,
\begin{equation}
M_\beta=I/(1+\beta\gamma/d_R).
\end{equation}

The complete form of $M_\beta$ is obtained by noticing that all vectors can be decomposed into the component that belongs to $\spn(R)$ and its orthogonal component:
\begin{equation}
\by=\Big[\by-\sum_{i=d_R+1}^{N}(\bphi_i\cdot\by)\bphi_i\Big]+\sum_{i=d_R+1}^{N}(\bphi_i\cdot\by)\bphi_i.
\end{equation}
Then, one has
\begin{IEEEeqnarray}{rCl}
&&\Big[I+\beta\sum_{\balpha\in R_+}\kappa(\balpha)\frac{\balpha \balpha^T}{\alpha^2}\Big]\Big[\Big(\by-\sum_{i=d_R+1}^{N}(\bphi_i\cdot\by)\bphi_i\Big)+\sum_{i=d_R+1}^{N}(\bphi_i\cdot\by)\bphi_i\Big]\nonumber\\
&&=\Big[\Big(1+\frac{\beta\gamma}{d_R}\Big)I-\frac{\beta\gamma}{d_R}\sum_{i=d_R+1}^{N}\bphi_i\bphi_i^T\Big]\by=M_\beta^{-1}\by,
\end{IEEEeqnarray}
and $M_\beta$ is given by 
\begin{equation}
M_\beta=\frac{1}{1+\beta\gamma/d_R}\Big[I+\frac{\beta\gamma}{d_R}\sum_{i=d_R+1}^{N}\bphi_i\bphi_i^T\Big]
\end{equation}
because
\begin{equation}
\Big[\Big(1+\frac{\beta\gamma}{d_R}\Big)I-\frac{\beta\gamma}{d_R}\sum_{i=1+d_R}^{N}\bphi_i\bphi_i^T\Big]\frac{1}{1+\beta\gamma/d_R}\Big[I+\frac{\beta\gamma}{d_R}\sum_{i=d_R+1}^{N}\bphi_i\bphi_i^T\Big]=I,
\end{equation}
which completes the proof.
\end{proof}
Denoting by $\bx_\perp$ the component of $\bx$ that is orthogonal to $\spn(R)$ and by $\bx_\parallel$ the component of $\bx$ which belongs to $\spn(R)$, one may write
\begin{equation}\label{LinearApproximationVSin}
V_\beta \bx\cdot\by = \frac{\bx_{\parallel}\cdot\by_\parallel}{1+\beta\gamma/d_R}+\bx_{\perp}\cdot\by_\perp.
\end{equation}
Unless otherwise noted, the subscripts $\perp$ and $\parallel$ in this equation will carry the meaning described here for the rest of the text. This last equation implies that the intertwining operator has no effect on linear functions of vectors that are orthogonal to the root system.

\section{Proof of Theorem~\ref{TheoremSteadyState}}\label{ProofTheoremSteadyState}

The first objective is to calculate the approximate value of the integral
\begin{equation}
\mathcal{J}(t,\bv):=\frac{1}{c_\beta}\int_{\RR^N}\rme^{-\zeta^2/2} \mathcal{I}(t,\bzeta) [V_\beta \rme^{\rmi \sqrt{\beta}\bv\cdot\bzeta}] w_\beta(\bzeta)\ud\bzeta,
\end{equation}
for large values of $t$. This expression is the inverse Dunkl transform of $\rme^{-\zeta^2/2} \mathcal{I}(t,\bzeta)$ evaluated at $\sqrt{\beta}\bv$ (see Eq.~\eqref{InverseDunklTransform}). In order to use the approximated form of $ \mathcal{I}(t,\bzeta) $, the positive variable $\epsilon$ is chosen with the assumption that $\sqrt{t}\epsilon\gg 1$. Then, $\mathcal{J}(t,\bv)$ is divided into two integrals. The first one, denoted $\mathcal{J}_1$, is taken over $\zeta<\sqrt{t}\epsilon$, while the second one, denoted $\mathcal{J}_2$ is taken over  $\zeta\geq\sqrt{t}\epsilon$. For $\mathcal{J}_2$, one has the following behavior:
\begin{multline}\label{EquationBoundForOuterIntegral}
|\mathcal{J}_2|\leq\frac{1}{c_\beta}\int_{\zeta\geq\sqrt{t}\epsilon}\rme^{-\zeta^2/2}  w_\beta(\bzeta)\ud\bzeta=\frac{C_\mathcal{J}}{c_\beta}\int_{\sqrt{t}\epsilon}^\infty\rme^{-\zeta^2/2} \zeta^{\beta\gamma+N-1}\ud\zeta\\
=\frac{C_\mathcal{J}2^{(\beta\gamma+N-2)/2)}}{c_\beta}\int_{t\epsilon^2/2}^\infty\rme^{-z} z^{(\beta\gamma+N-2)/2}\ud z=O(\rme^{-t\epsilon^2/2}(t\epsilon^2)^{(\beta\gamma+N-2)/2}).
\end{multline}
Here, $C_\mathcal{J}$ denotes the angular part of the integral, the substitution $z=\zeta^2/2$ was carried out in the second line, and the final step is obtained from integration by parts. Note that the first inequality follows from the fact that $\mathcal{I}(t,\bzeta)$ is bounded,
\begin{equation}
|\mathcal{I}(t,\bzeta)|\leq\int_{\RR^N}|V_\beta \rme^{-\rmi \bu\cdot\bzeta}|t^{N/2}\mu(\sqrt{t}\bu)\ud\bu\leq\int_{\RR^N}\mu(\bx)\ud\bx=1.
\end{equation}

Integrals of this form appear repeatedly in this derivation. Thus, it is convenient to have the general expression
\begin{equation}\label{EquationIntegralOutside}
\frac{1}{c_\beta}\int_{\zeta\geq\sqrt{t}\epsilon}\rme^{-\zeta^2/2} g(\bzeta) w_\beta(\bzeta)\ud\bzeta=O(\rme^{-t\epsilon^2/2}(t\epsilon^2)^{(\beta\gamma+N+r-2)/2})
\end{equation}
provided that $g(\bzeta)\sim \zeta^r$ for large $\zeta$, which is derived using the procedure that leads to Eq.~\eqref{EquationBoundForOuterIntegral}. 

Using Eq.~\eqref{EquationApproximationForI}, one may separate the integral over the region $\zeta<\epsilon\sqrt{t}$, $\mathcal{J}_1$, into the part that corresponds to the cosine ($\mathcal{J}_{\cos}$), the part that corresponds to the sine ($\mathcal{J}_{\sin}$), and the part that corresponds to the lower-order terms ($\mathcal{J}_\text{o}$). After extending the domain of integration to $\RR^N$ at the expense of the error term from Eq.~\eqref{EquationIntegralOutside} with $r=0$, the first integral reads
\begin{multline}
\mathcal{J}_{\cos}=\frac{1}{2c_\beta}\int_{\RR^N}\rme^{-\zeta^2/2}V_\beta(\rme^{\rmi \bu_\text{R}\cdot\bzeta}+\rme^{-\rmi\bu_\text{R}\cdot\bzeta})V_\beta\rme^{\rmi\sqrt{\beta}\bv\cdot\bzeta}w_\beta(\bzeta)\ud\bzeta\\
+O(\rme^{-t\epsilon^2/2}(t\epsilon^2)^{(\beta\gamma+N-2)/2}).
\end{multline}
Using Eq.~\eqref{GaussianIntegralDunklKernel}, this integral can be evaluated immediately, yielding
\begin{equation}
\mathcal{J}_{\cos}=\rme^{-(u_\text{R}^2+\beta v^2)/2}V_\beta\cosh(\sqrt{\beta} \bu_\text{R}\cdot\bv)
+O(\rme^{-t\epsilon^2/2}(t\epsilon^2)^{(\beta\gamma+N-2)/2}).
\end{equation}
The integral $\mathcal{J}_{\sin}$ can be evaluated in a similar manner as
\begin{equation}
\mathcal{J}_{\sin}=\rme^{-(u_\text{I}^2+\beta v^2)/2}V_\beta\sinh(\sqrt{\beta} \bu_\text{I}\cdot\bv)
+O(\rme^{-t\epsilon^2/2}(t\epsilon^2)^{(\beta\gamma+N-2)/2}).
\end{equation}
The final integral $\mathcal{J}_\text{o}$ is also calculated using Eq.~\eqref{EquationIntegralOutside}:
\begin{multline}
\mathcal{J}_\text{o}=\frac{O[\epsilon^2]}{c_\beta}\int_{\zeta<\sqrt{t}\epsilon}\rme^{-\zeta^2/2}V_\beta\rme^{\rmi\sqrt{\beta}\bv\cdot\bzeta}w_\beta(\bzeta)\ud\bzeta\\
=O[\epsilon^2]\Big[\frac{1}{c_\beta}\int_{\RR^N}\rme^{-\zeta^2/2}V_\beta\rme^{\rmi\sqrt{\beta}\bv\cdot\bzeta}w_\beta(\bzeta)\ud\bzeta+O(\rme^{-t\epsilon^2/2}(t\epsilon^2)^{(\beta\gamma+N-2)/2})\Big]\\
=O[\epsilon^2][\rme^{-\beta v^2/2}+O(\rme^{-t\epsilon^2/2}(t\epsilon^2)^{(\beta\gamma+N-2)/2})].
\end{multline}

For a very large value of $t\epsilon^2$, the term $O(\rme^{-t\epsilon^2/2}(t\epsilon^2)^{(\beta\gamma+N-2)/2})$ can be neglected. For this regime, the result is
\begin{multline}\label{EquationApproximationForJ}
\mathcal{J}(t,\bv)=\rme^{-\beta v^2/2}\{\rme^{-u_\text{R}^2/2}V_\beta\cosh(\sqrt{\beta} \bu_\text{R}\cdot\bv)\\
+\rme^{-u_\text{I}^2/2}V_\beta\sinh(\sqrt{\beta} \bu_\text{I}\cdot\bv)+O[\epsilon^2]\}.
\end{multline}
Let us recall from Eq.~\eqref{EquationApproximationForI} that this result holds only when the conditions
\begin{equation}\label{EquationFirstIntegralConditionForT}
t\epsilon^{2}\gg 1\quad\text{and}\quad t u_\text{R}^2\epsilon^{2}=\epsilon^{2}(s_\mu^2+\bar{x}_\mu^2)\ll1
\end{equation}
are satisfied. From Eqs.~\eqref{EquationForUI}, \eqref{EquationForUR}, and \eqref{EquationApproximationForJ}, it is clear that
\begin{equation}
\mathcal{J}(t,\bv)\stackrel{t\to\infty}{\longrightarrow}\rme^{-\beta v^2/2},
\end{equation}
as expected. 

The second objective of this derivation is, then, to find a lower bound on the time required to reach the steady state. Because $t$ is required to be large in the previous result, additional approximations can be made. The expression $V_\beta\cosh(\bu_\text{R}\cdot\bv)$ for $u_\text{R}v\ll 1$ is approximated by noticing that the function
\begin{equation}\label{ApproximationVCosine}
g(\bv,\bu_\text{R})=\exp\Big[\frac{v^2u_\text{R}^2}{2(N+\beta\gamma)}\Big]
\end{equation}
behaves roughly like $V_\beta\cosh(\bu_\text{R}\cdot\bv)$ in the following sense. Consider the expression
\begin{equation}
\sum_{j=1}^N T_j^2 V_\beta\cosh(\bu_\text{R}\cdot\bv)|_{\bzeta=\bzero}=u_\text{R}^2,
\end{equation}
where the Dunkl operators act on $\bv$. A straightforward calculation gives
\begin{equation}
\Delta_v g(\bv,\bu_\text{R})=\frac{u_\text{R}^2}{N+\beta\gamma}\Big[\frac{v^2u_\text{R}^2}{N+\beta\gamma}+N\Big]g(\bzeta,\bu_\text{R})
\end{equation}
and
\begin{equation}
\beta\sum_{\balpha\in R_+}\kappa(\balpha)\frac{\balpha\cdot\bnabla_v g(\bv,\bu_\text{R})}{\balpha\cdot\bv}=\frac{\beta\gamma u_\text{R}^2}{N+\beta\gamma}g(\bv,\bu_\text{R}),
\end{equation}
which yields
\begin{equation}
\sum_{j=1}^N T_j^2 g(\bv,\bu_\text{R})|_{\bv=\bzero}=u_\text{R}^2.
\end{equation}
It follows that $g(\bv,\bu_\text{R})$ reproduces the isotropic part of the second-order term of $V_\beta\cosh(\bu_\text{R}\cdot\bv)$ at $\bv=\bzero$. Consequently, the combination of $g(\bzeta,\bu_\text{R})$ and Eq.~\eqref{LinearApproximationVSin} gives an approximation of the hyperbolic sine and cosine terms that is only accurate up to first order, but which is still better than a simple first-order approximation:
\begin{multline}\label{EquationDoubleApproximationForJ}
\rme^{\beta v^2/2}\mathcal{J}(t,\bv)-O(\epsilon^2)=\rme^{-u_\text{R}^2/2}V_\beta\cosh( \sqrt{\beta}\bu_\text{R}\cdot\bv)
+\rme^{-u_\text{I}^2/2}V_\beta\sinh( \sqrt{\beta}\bu_\text{I}\cdot\bv)\\
=\exp\Big[-\frac{u_\text{R}^2}{2}\Big(1-\frac{\beta v^2}{N+\beta\gamma}\Big)\Big] + \rme^{-u_\text{I}^2/2}\Big[\frac{\sqrt{\beta}\bu_{\text{I}\parallel}\cdot\bv_\parallel}{1+\beta\gamma/d_R}+\sqrt{\beta}\bu_{\text{I}\perp}\cdot\bv_\perp  \Big] +O(\beta v^2/t).
\end{multline}

The choice of the function $g(\bv,\bu_\text{R})$ is motivated by the fact that if the approximation is carried out up to first order, the information concerning the second moments of $\mu(\bx)$ is lost. Approximating the function $V_\beta\cosh(\bu_\text{R}\cdot\bv)$ using an isotropic function such as $g(\bv,\bu_\text{R})$ neglects the anisotropies that correspond to the vector $\bu_\text{R}$ as well as those that correspond to the intertwining operator, but it also allows one to have an approximate form of the hyperbolic functions which carries additional information about the initial distribution.

To estimate the relaxation time, consider a test function $h(\bv)$ with a power-law asymptotic behavior $h(\bv)\sim v^r$ when $v$ is very large. Define the expectation of $h$ over the steady-state distribution by
\begin{equation}
\langle h \rangle:=\int_{\RR^N}h(\bv)f_R(\beta,\bv)\ud\bv.
\end{equation}
Then, the expectation of $h$ before the steady-state is achieved is given by
\begin{equation}\label{EquationAverageBeforeRelaxation}
\langle h \rangle_t:=\frac{1}{z_\beta}\int_{\RR^N}h(\bv)w_\beta(\bv)\mathcal{J}(t,\bv)\ud\bv.
\end{equation}
Equation~\eqref{EquationDoubleApproximationForJ} can only be used if the inequality
\begin{equation}
\frac{(s_\mu^2+\bar{x}_\mu^2)\beta v^2}{t}\ll 1
\end{equation}
is satisfied. Therefore, the integral \eqref{EquationAverageBeforeRelaxation} can be divided into the region where $v<\eta\sqrt{t}$ and the region where $v\geq\eta\sqrt{t}$, where it is assumed that $\eta\sqrt{t}\gg 1$.

Denote the integral over the former region by $\mathcal{K}_1$ and the latter by $\mathcal{K}_2$. After making the substitution $\bv^\prime=\sqrt{\beta}\bv$, the integral over the outer region is given by 
\begin{multline}
|\mathcal{K}_2|\leq\frac{1}{c_\beta}\int_{v^\prime\geq\eta\sqrt{\beta t}}|h(\bv^\prime/\sqrt{\beta})|w_\beta(\bv^\prime)\rme^{- (v^\prime)^2/2}|\{\rme^{-u_\text{R}^2/2}V_\beta\cosh(\bu_\text{R}\cdot\bv^\prime)\\
+\rme^{-u_\text{I}^2/2}V_\beta\sinh(\bu_\text{I}\cdot\bv^\prime)+O[\epsilon^2]\}|\ud\bv^\prime\\
\leq \frac{C_\mathcal{K}}{c_\beta}
\int_{\eta\sqrt{\beta t}}^\infty
(v^\prime)^{\beta\gamma+N-1}
(v^\prime/\sqrt{\beta})^r
\{\rme^{- (v^\prime-u_\text{R})^2/2}/2
+\rme^{-(v^\prime-u_\text{I})^2/2}/2\\
+\rme^{-(v^\prime)^2/2}O[\epsilon^2]\}\ud\bv^\prime.
\end{multline}
For the final expression, the asymptotic form of $h(\bv)$ was used, as well as the bounds and asymptotic forms for $v^\prime\gg 1$
\begin{IEEEeqnarray}{rCl}
|V_\beta\cosh(\bu_\text{R}\cdot\bv^\prime)|\leq \cosh(u_\text{R}v^\prime)\sim \exp(u_\text{R}v^\prime)/2,\\
|V_\beta\sinh(\bu_\text{I}\cdot\bv^\prime)|\leq \sinh(u_\text{I}v^\prime)\sim \exp(u_\text{I}v^\prime)/2.
\end{IEEEeqnarray}
The coefficient $C_\mathcal{K}$ represents the value of the angular integral. By Eq.~\eqref{EquationIntegralOutside}, one obtains
\begin{equation}\label{EquationNonSteadyStateExpectationOuterIntegral}
\mathcal{K}_2= O[\rme^{-\eta^2\beta t/2}(\eta^2 \beta t)^{(\beta\gamma+N+r-2)/2}\beta^{-r/2}(1+\epsilon^2)].
\end{equation}
Therefore, one must assume that $\eta\sqrt{\beta t}\gg 1$ in order to neglect the terms with exponential decay.

The inner integral is divided into the lower-order ($\mathcal{K}_o$), sine ($\mathcal{K}_{\sinh}$) and cosine ($\mathcal{K}_{\cosh}$) terms. Unless otherwise noted, exponentially decreasing correction terms (e.g., the correction in Eq.~\eqref{EquationNonSteadyStateExpectationOuterIntegral}) will be implicit in each of the expressions that follow. The lower-order integral is immediate,
\begin{multline}
\mathcal{K}_o=\frac{O(\epsilon^2)+O(\beta\eta^2)}{c_\beta}\int_{v^\prime<\eta\sqrt{\beta t}}h(\bv^\prime/\sqrt{\beta})w_\beta(\bv^\prime)\rme^{- (v^\prime)^2/2}\ud\bv^\prime\\
=\langle h \rangle [O(\epsilon^2)+O(\beta\eta^2)].
\end{multline}
The sine term becomes, assuming that $\eta\ll1/\sqrt{\beta \bar{x}_\mu^2}$, 
\begin{multline}\label{EquationRelaxationSinh}
\mathcal{K}_{\sinh} = \frac{1}{c_\beta}\int_{v^\prime<\eta\sqrt{\beta t}}h(\bv^\prime/\sqrt{\beta})w_\beta(\bv^\prime)\rme^{- (v^\prime)^2/2}\\
\times\rme^{-u_\text{I}^2/2}\Big[\frac{\bu_{\text{I}\parallel}\cdot\bv^\prime_\parallel}{1+\beta\gamma/d_R}+\bu_{\text{I}\perp}\cdot\bv^\prime_\perp\Big]\ud\bv^\prime.
\end{multline}
Within this region of integration, the bound
\begin{equation}\label{AttentionSlowRelaxation}
\Big|\frac{\bu_{\text{I}\parallel}\cdot\bv^\prime_\parallel}{1+\beta\gamma/d_R}+\bu_{\text{I}\perp}\cdot\bv^\prime_\perp\Big|\leq \frac{\eta\sqrt{\beta}}{\bar{x}_\mu}\Big[\frac{\bar{x}_{\mu\parallel}^2}{1+\beta\gamma/d_R}+\bar{x}_{\mu\perp}^2\Big]
\end{equation}
is satisfied. This means that
\begin{equation}
\mathcal{K}_{\sinh}=\langle h\rangle O(\eta\sqrt{\beta}),
\end{equation}
where the correction term comes from Eq.~\eqref{AttentionSlowRelaxation}, and $\langle h \rangle$ comes from Eq.~\eqref{EquationRelaxationSinh} after taking $u_\text{I}\to 0$ due to the inequality \eqref{EquationFirstIntegralConditionForT}. The cosine term becomes, after using Eq.~\eqref{ApproximationVCosine},
\begin{multline}
\mathcal{K}_{\cosh}=\frac{1}{c_\beta}\int_{v^\prime<\eta\sqrt{\beta t}}h(\bv^\prime/\sqrt{\beta})w_\beta(\bv^\prime)\exp\Big[-\frac{(v^\prime)^2}{2} \Big(1-\frac{s_\mu^2+\bar{x}_\mu^2}{t(\beta\gamma+N)}\Big)\Big]\ud\bv^\prime\\
\times\rme^{-(s_\mu^2+\bar{x}_\mu^2)/(2t)}=\langle h \rangle.
\end{multline}
For the last equality, Eq.~\eqref{EquationFirstIntegralConditionForT} has been used. Finally, adding all the terms and discarding higher orders of $\epsilon$ and $\eta$ yields
\begin{equation}
\langle h\rangle_t=\langle h \rangle [1+O(\eta\sqrt{\beta})+O(\epsilon^2)],
\end{equation}
provided that all of the following assumptions are satisfied,
\begin{multline}\label{EquationSecondIntegralConditionForT}
t\gg s_\mu^2 +\bar{x}_\mu^2,\quad \epsilon^2t\gg1,\quad \epsilon^{-2}\gg s_\mu^2+\bar{x}_\mu^2,\quad \\
\eta^2 \beta t\gg 1,\quad \eta^2 t\gg 1,\quad \eta^{-2}\gg \beta(s_\mu^2+\bar{x}_\mu^2).
\end{multline}
Therefore, the distribution must be relaxed to the steady state for $t\gg(s_\mu^2+\bar{x}_\mu^2)\max[1,\beta]$. The statement follows. \qquad\ \qquad\ \qquad\ \qquad \qquad\ \qquad\ \qquad $\square$

While the time bound in Theorem~\ref{TheoremSteadyState} guarantees the relaxation of the system to the steady state, it is not a true estimation of the relaxation time. The time bound estimated here implies that as $\beta$ grows, a longer time is required for the system to relax in general. This is not true at least in the freezing limit, where the system achieves the steady state instantaneously. This property of Dunkl processes in the freezing limit will be proved in Chapter~\ref{general_freezing}. If one assumes a smooth change of behavior from a Dunkl process as $\beta$ changes, then it seems reasonable to believe that Dunkl processes reach the steady state in a small time for large values of $\beta$. This means that there exists an estimation of the relaxation time that can be achieved by making some assumptions on the initial distribution, e.g. compact support, exponential decay at large $\bx$, etc. The result presented here, however, only requires $\mu(\bx)$ to have finite second moments, and in consequence gives a time bound that seems to exceed the actual relaxation time of the system.

It is also worth noting that  Eq.~\eqref{AttentionSlowRelaxation} is responsible for the highest-order correction. Because $1+\beta\gamma/d_R>1$, the largest correction is caused by the component of $\bar{\bx}_\mu$ that is orthogonal to the root system. This means that if the mean of $\mu(\bx)$ has a large component in the space that is orthogonal to $\spn(R)$, the Dunkl process will take a long time to reach the steady state. This fact will be illustrated using numerical simulations of the interacting Brownian motions in Chapter~\ref{ParticularCases}. In addition, if $\bar{\bx}_{\mu\perp}=\bzero$, the correction due to Eq.~\eqref{AttentionSlowRelaxation} decreases in magnitude as $\beta$ grows. Therefore, Dunkl processes on a root system of full rank converge more rapidly to the steady state at large $\beta$ than on a root system with $d_R<N$ in general.

% !TEX encoding = UTF-8 Unicode
\chapter{Freezing regime in an arbitrary root system}\label{general_freezing}

The discussion from the previous chapter was focused on how Dunkl processes behave after a long time, in which case their distribution functions do not depend on the initial distribution. Here, the focus is shifted to the situation in which the inverse temperature tends to infinity (the freezing limit). In this case, Dunkl processes attain their steady state instantly, and their distribution is given by a sum of delta functions localized in the peak set of $R$ in general \cite{dunkl89B}. In some particular cases, such as the root systems $A$ and $B$, these sets of points are identified as the solution of certain log-Fekete problems as mentioned in Chapter~\ref{CalogeroMoserCorrespondence} \cite{deift00}. 

\section{Setting}

In this case, it is convenient to use Eq.~\eqref{TransitionDensityExplicit} to express the probability distribution of the Dunkl process with the initial distribution $\mu(\bx)$,
\begin{equation}
f(t,\by)\ud\by=w_\beta\left(\frac{\by}{\sqrt{t}}\right)\frac{\rme^{-y^2/2t}}{c_\beta t^{N/2}}\int_{\RR^N}\rme^{-x^2/2t}V_\beta \exp \left(\frac{\bx\cdot\by}{t}\right)\mu(\bx)\ud\bx\ud\by.
\end{equation}
The objective in this chapter is to calculate how the scaled distribution
\begin{multline}\label{ScaledProcessDistributionFreezing}
f(t,\sqrt{\beta t}\bv)(\beta t)^{N/2}\ud\bv=\frac{w_\beta(\bv)\rme^{-\beta v^2/2}}{z_\beta}\int_{\RR^N}\rme^{-x^2/2t}V_\beta \rme^{\bx\cdot\bv\sqrt{\beta/t}}\mu(\bx)\ud\bx\ud\bv\\
=f_R(\bv)\int_{\RR^N}\rme^{-x^2/2t}V_\beta \rme^{\bx\cdot\bv\sqrt{\beta/t}}\mu(\bx)\ud\bx\ud\bv
\end{multline}
behaves as $\beta\to\infty$. As mentioned in Chap.~\ref{CalogeroMoserCorrespondence}, the \emph{peak set} of the root system $R$ is defined as the set of vectors $\{\bos_i\}_{i=1}^{|W|}$ where the function $F_R(\bv,\kappa)$ attains its minima \cite{dunkl89B}. Assume that the mean and variance of the distribution $\mu(\bx)$ are given by $\bar{\bx}_\mu$ and $s_\mu^2$ (Eqs.~\eqref{MuFirstMoments} and \eqref{MuSecondMoment} respectively). The freezing limit of a Dunkl process is given by the following.
\begin{theorem}\label{TheoremFreezingLimit}
In the limit where $\beta\to\infty$, the scaled probability distribution of a Dunkl process for $t>0$ is given by
\begin{equation}\label{GeneralFreezingLimit}
\lim_{\beta\to\infty}f(t,\sqrt{\beta t}\bv)(\beta t)^{N/2}\ud\bv=\frac{1}{|W|}\sum_{i=1}^{|W|}\delta^{(N)}(\bv-\bos_i)\ud\bv.
\end{equation}
\end{theorem}

Two remarks are required to give adequate meaning to the statement in Theorem~\ref{TheoremFreezingLimit}. The first remark is that the freezing limit depicted in Eq.~\eqref{GeneralFreezingLimit} is derived from the scaled distribution $f(t,\sqrt{\beta t}\bv)(\beta t)^{N/2}$ and not from the steady-state distribution $f_R(\bv)$. In fact, due to Eq.~\eqref{ScaledProcessDistributionFreezing} the freezing limit of $f(t,\sqrt{\beta t}\bv)(\beta t)^{N/2}$ is obtained from the freezing limit of $f_R(\bv)$ and of the integral over $\mu(\bx)$
\begin{equation}
\int_{\RR^N}\rme^{-x^2/2t}V_\beta \rme^{\bx\cdot\bv\sqrt{\beta/t}}\mu(\bx)\ud\bx.
\end{equation}
It will be shown in Lemma~\ref{SteadyStateDistributionFreezingLimit} that the freezing limit of $f_R(\bv)$ is the sum of delta functions on the r.h.s.\ of Eq.~\eqref{GeneralFreezingLimit}, while this integral tends to one when $\bv=\bos_i$ as $\beta\to\infty$.

The second remark is that, because $f(t,\sqrt{\beta t}\bv)(\beta t)^{N/2}$ and $f_R(\bv)$ are equal in the freezing limit for $t>0$, it follows that scaled Dunkl processes relax to the steady state instantaneously when $\beta\to\infty$. This statement is given a physical meaning as follows. Recalling the SDE of radial Dunkl processes, Eq.~\eqref{RadialDunklSDE}, the SDE of the scaled process $\bY_t=\bX_t/\sqrt{\beta}$ can be easily found to be
\begin{equation}
\ud \bY_t=\frac{\ud \bB_t}{\sqrt{\beta}}+\frac{1}{2}\sum_{\balpha\in R_+}\frac{\kappa(\balpha)\balpha}{\balpha\cdot\bY_t}\ud t.
\end{equation}
From the Smoluchowski-Kramers approximation \cite{nelson67,freidlin04}, this process can be interpreted as a particle of negligible mass inside a viscous fluid in $N$ dimensions. The particle interacts with an external force $-\bnabla \Phi$ with potential
\begin{equation}
\Phi(\bv)=-\log\Big[\prod_{\balpha\in R_+}|\balpha\cdot\bv|^{\kappa(\balpha)/2}\Big],
\end{equation}
and with the component particles of the background fluid through thermal (probabilistic) collisions. The second interaction is represented here by the term $\ud \bB_t/\sqrt{\beta}$. From this point of view, the freezing limit corresponds to the situation in which the thermal vibration of the particles of the background fluid is negligible. Because $\bY_t=\bX_t/\sqrt{\beta}$, the initial condition $\bX_0=\bx$ is translated to $\bY_0=\bx/\sqrt{\beta}$, and it follows that in the freezing limit $\bY_0=\bzero$ for any $\bx$ finite. This means that when $\beta\to\infty$ the process $\bY_t$ always starts from the origin, and because in this regime the motion of $\bY_t$ is deterministic, its trajectory is independent of the initial condition on $\bX_t$. To see this in Eq.~\eqref{GeneralFreezingLimit}, it suffices to make $\bu=\sqrt{t}\bv$ to obtain
\begin{equation}
\lim_{\beta\to\infty}f(t,\sqrt{\beta}\bu)(\beta)^{N/2}\ud\bu=\frac{1}{|W|}\sum_{i=1}^{|W|}\delta^{(N)}(\bu-\sqrt{t}\bos_i)\ud\bu.
\end{equation}
Therefore, the path of the process $\bY_t$ in the freezing limit is reduced to a deterministic set of curves that are independent of the initial distribution. This is because in the scale of $\bY_t$, the initial condition $\mu(\bx)$ is reduced to a delta function at the origin as $\beta\to\infty.$ The mathematical basis for this assertion is given by the procedure in Section~\ref{ProofOfTheoremFreezingLimit} using the variable substitutions $\bu=\bv/\sqrt{t}$ and $\bx=\sqrt{\beta}\bxi$, and it is omitted for brevity. After this long consideration, it can be concluded that Dunkl processes in the freezing limit relax instantaneously because at times $t>0$ an initial distribution $\mu(\bx)$ is localized close to the origin when viewed from the scale of $\bv=\by/\sqrt{\beta t}$. Consequently, the drift terms of the process are significantly stronger than the thermal fluctuations, which in turn allows the system to achieve the steady state faster as $\beta$ grows to infinity.

\section{The Dunkl kernel in the freezing regime}

Before tackling the proof of Thm.~\eqref{TheoremFreezingLimit}, it is necessary to investigate the behavior of the integral in Eq.~\eqref{ScaledProcessDistributionFreezing}, which in turn involves the behavior of the Dunkl kernel as $\beta$ tends to infinity. The objective of this section is to derive the form of the Dunkl kernel in this regime. Define the limit
\begin{equation}
V_\infty f(\bx) := \lim_{\beta\to\infty} V_\beta f(\bx)
\end{equation}
for functions $f(\bx)\in A_{|\bx|}$, where the set $A_{|\bx|}$ is defined in Eqs.~\eqref{EquationDefinitionAr1} and \eqref{EquationDefinitionAr2}.

\begin{lemma}\label{PropositionVInfinityWInvariance}
The function $V_\infty f(\bx)$ is $W$-invariant.
\end{lemma}
\begin{proof}
Consider Eq.~\eqref{EquationVDefinition} divided by $\beta$. Arranging terms, one obtains
\begin{equation}
\frac{1}{\beta}\Big[V_\beta\frac{\partial}{\partial x_i}f(\bx)-\frac{\partial}{\partial x_i}V_\beta f(\bx)\Big]=\frac{1}{2}\sum_{\balpha\in R_+}\alpha_i\kappa(\balpha)\frac{(1-\sigma_{\balpha})V_\beta f(\bx)}{\balpha\cdot\bx}
\end{equation}
for $i=1,\ldots,N$. Due to the boundedness of $V_\beta$ (Eq.~\eqref{EquationVBoundedness}), which is independent of $\beta$, the term on the l.h.s.\ tends to 0 when $\beta\to\infty$. Therefore, in the freezing limit the term on the r.h.s.\ must vanish for any multiplicity function $\kappa(\balpha)$ and any $f(\bx)\in A_{|\bx|}$. This means that $V_\infty f(\bx)=V_\infty f(\sigma_{\balpha}\bx)$ for any $\balpha$, which implies that $V_\infty f(\bx)$ must be $W$-invariant.
\end{proof}

As a consequence of Lemma~\ref{PropositionVInfinityWInvariance}, it is necessary to find first- and second-order combinations of Dunkl operators that preserve the $W$-invariance. These operators will be the main tools in the derivation of the Dunkl kernel in the freezing regime. For the first-order case, one has the following.
\begin{lemma}\label{PropositionFirstOrderWInvariantOperator}
Full-rank root systems do not have first-order operators which preserve the $W$-invariance.
For root systems that are not of full rank, any operator of the form $\sum_{i=1}^N \xi_i T_i$ with $\bxi=(\xi_1,\ldots,\xi_N)^T$ orthogonal to $\spn(R)$ preserves the $W$-invariance.
\end{lemma}
\begin{proof}
Suppose that the function $f(\bx)$ is $W$-invariant. The objective is to find out whether
\begin{equation}
T_{\bxi} f(\bx)=\bxi\cdot\bnabla f(\bx)+\frac{\beta}{2}\sum_{\balpha\in R_+}(\bxi\cdot\balpha) \kappa(\balpha)\frac{(1-\sigma_{\balpha})f(\bx)}{\balpha\cdot\bx}
\end{equation}
is $W$-invariant. The sum on the r.h.s.\ vanishes because $f(\bx)=f(\sigma_{\balpha}\bx)$ for $\balpha\in R$, and thus
\begin{equation}\label{EquationProofFirstOrderWInvarianceSetup}
T_{\bxi} f(\bx)=\bxi\cdot\bnabla f(\bx).
\end{equation}
From Eq.~\eqref{EquationWActionOnFunctions}, for $\rho\in W$,
\begin{equation}
\rho\Big[\frac{\partial}{\partial x_i}f(\bx)\Big]=\frac{\partial f(\rho^T\bx)}{\partial (\rho^T\bx)_i}.
\end{equation}
Then, with $\brho_i$ denoting the $i$th column of the matrix representation of $\rho$, one has the equation
\begin{equation}\label{EquationPartialDerivativeOrthogonalChangeBasis}
\rho\Big[\frac{\partial}{\partial x_i}f(\bx)\Big]=\brho_i\cdot\bnabla [\rho f(\bx)]=\brho_i\cdot\bnabla f(\bx).
\end{equation}
For the r.h.s.\ of Eq.~\eqref{EquationProofFirstOrderWInvarianceSetup} to be $W$-invariant, it is required that 
\begin{equation}
\bxi\cdot\bnabla f(\bx) = \rho[\bxi\cdot\bnabla f(\bx)] = \sum_{j=1}^N\xi_j \rho\Big[\frac{\partial}{\partial x_i}f(\bx)\Big] = [\rho\bxi]\cdot\bnabla f(\bx)
\end{equation}
for any $\rho\in W$. This means that $\bxi$ preserves the $W$-invariance of $f$ only when $\rho\bxi=\bxi$. That is, $\sigma_{\balpha}\bxi=\bxi$ (${}^\forall \balpha\in R$), which means that $\bxi$ must be orthogonal to all the roots in $R$. The result follows.
\end{proof}

Out of all the second-order combinations of Dunkl operators, the Dunkl Laplacian is the only one which has the following property.
\begin{lemma}\label{PropositionSecondOrderWInvariantOperator}
The Dunkl Laplacian preserves the $W$-invariance for any root system.
\end{lemma}

\begin{proof}
For $f(\bx)$ $W$-invariant and an arbitrary real symmetric matrix $[A]_{ij}=a_{ij}$, the following expression holds,
\begin{multline}\label{WInvariantSecondDegreeOperator}
\sum_{1\leq i,j\leq N}a_{ij}T_iT_j f(\bx)=\sum_{1\leq i,j\leq N}a_{ij}\Big[\frac{\partial^2}{\partial x_i\partial x_j}\\
+\frac{\beta}{2}\sum_{\balpha\in R_+}\alpha_i\kappa(\balpha)\frac{1-{\sigma_{\balpha}}}{\balpha\cdot\bx}\frac{\partial}{\partial x_j}\Big]f(\bx).
\end{multline}
The other terms vanish due to the $W$-invariance of $f$. The objective is to find conditions on the matrix $A$ which make the l.h.s.\ of Eq.~\eqref{WInvariantSecondDegreeOperator} $W$-invariant. Each term in the r.h.s.\ must preserve the $W$-invariance of $f$. Therefore, one may examine each term separately. 

Consider the first term in the expression. Suppose that $\rho\in W$. Using Eq.~\eqref{EquationPartialDerivativeOrthogonalChangeBasis} and the $W$-invariance of $f$, one obtains:
\begin{multline}
\sum_{1\leq i,j\leq N}a_{ij}\frac{\partial^2}{\partial x_i\partial x_j} f(\bx)=\rho\Big[\sum_{1\leq i,j\leq N}a_{ij}\frac{\partial^2}{\partial x_i\partial x_j} f(\bx)\Big]\\
=\sum_{1\leq i,j\leq N}a_{ij}\brho_i\cdot\bnabla [\brho_j\cdot\bnabla f(\bx)]\\
=\sum_{1\leq i,j\leq N}[\rho^T A \rho]_{ij}\frac{\partial}{\partial x_i}\frac{\partial}{\partial x_j}f(\bx).
\end{multline}
Then, for any $\rho\in W$,
\begin{equation}\label{ConditionForSecondDegreeWInvariance}
A=\rho^T A\rho.
\end{equation}

Consider now the second term. In general, 
\begin{multline}
\rho\Big[\sum_{\balpha\in R_+}\alpha_i\kappa(\balpha)\frac{1-{\sigma_{\balpha}}}{\balpha\cdot\bx}\frac{\partial f(\bx)}{\partial x_j}\Big]=\sum_{\balpha\in R_+}\alpha_i\kappa(\balpha)\frac{1-{\sigma_{\rho \balpha}}}{(\rho\balpha)\cdot\bx}[\brho_j\cdot\bnabla f(\bx)]\\
=\sum_{\balpha^\prime\in R_+}[\brho_i\cdot\balpha]\kappa(\balpha^\prime)\frac{1-{\sigma_{\balpha^\prime}}}{\balpha^\prime\cdot\bx}[\brho_j\cdot\bnabla f(\bx)].
\end{multline}
Here, Eq.~\eqref{EquationPartialDerivativeOrthogonalChangeBasis} has been used, and in the last line the roots $\balpha^\prime=\rho\balpha$ have been defined. The second term transformed by $\rho$, including the sum over $i$ and $j$, becomes
\begin{multline}
\frac{\beta}{2}\sum_{1\leq i,j\leq N}a_{ij}\rho\Big[\sum_{\balpha\in R_+}\alpha_i\kappa(\balpha)\frac{1-{\sigma_{\balpha}}}{\balpha\cdot\bx}\frac{\partial f(\bx)}{\partial x_j}\Big]\\
=\frac{\beta}{2}\sum_{\balpha\in R_+}\kappa(\balpha)\frac{1-{\sigma_{\balpha}}}{\balpha\cdot\bx}[\rho A\rho^T\balpha]\cdot\bnabla f(\bx).
\end{multline}
Thus, if the second term is $W$-invariant, the following condition must hold:
\begin{multline}
\frac{\beta}{2}\sum_{\balpha\in R_+}\kappa(\balpha)\frac{1-{\sigma_{\balpha}}}{\balpha\cdot\bx}[\rho A\rho^T\balpha]\cdot\bnabla f(\bx)=\frac{\beta}{2}\sum_{\balpha\in R_+}\kappa(\balpha)\frac{1-{\sigma_{\balpha}}}{\balpha\cdot\bx}[A\balpha]\cdot\bnabla f(\bx).
\end{multline}

This means that for both terms in Eq.~\eqref{WInvariantSecondDegreeOperator} to be $W$-invariant, the condition given by Eq.~\eqref{ConditionForSecondDegreeWInvariance} must hold. For root systems with $d_R=N$ this condition is only satisfied by $A=cI$ with $c\in\RR$. When $d_R<N$, one requires $A=(cI_{d_R})\otimes A^\prime$. That is, $A$ must behave like an identity matrix in the subspace $\spn(R)$ and like an arbitrary symmetric matrix in the subspace that is orthogonal to $\spn(R)$. Setting $A=I$ in Eq.~\eqref{WInvariantSecondDegreeOperator} yields the Dunkl Laplacian, which satisfies Eq.~\eqref{ConditionForSecondDegreeWInvariance} for all root systems.
\end{proof}

From the previous two lemmas, it is clear that the behavior of the Dunkl kernel depends on the rank of the root system. Consider a root system of rank less than $N$. Then, one may consider the basis used for Lemma~\ref{LemmaFirstOrderV}, $\{\bphi_i\}_{i=1}^N$. In this case, the Dunkl kernel in the freezing limit is given by the following simple form.
\begin{lemma}\label{FreezingLimitDunklKernelNonFullRank}
For root systems with $d_R<N$, the freezing limit of the Dunkl kernel is given by
\begin{equation}\label{EquationFreezingLimitDunklKernelNonFullRank}
V_\infty\rme^{\bx\cdot\by}=\lim_{\beta\to\infty}V_\beta\rme^{\bx\cdot\by}=\exp\Big[\sum_{d_R<i\leq N}(\bx\cdot\bphi_i)(\bphi_i\cdot\by)\Big].
\end{equation}
\end{lemma}

\begin{proof}
For this derivation, denote $V_\infty\rme^{\bx\cdot\by}$ by $g(\bx,\by)$. By Lemma~\ref{PropositionVInfinityWInvariance}, the function $g(\bx,\by)$ must be $W$-invariant. At the same time, by definition,
\begin{equation}\label{EquationDunklKernelFirstOrder}
T_{\bxi} V_\beta \rme^{\bx\cdot\by}=\bxi\cdot\by V_\beta \rme^{\bx\cdot\by}
\end{equation}
for all $\beta>0$ and $\bxi\in\RR^N$. However, by Lemma~\ref{PropositionFirstOrderWInvariantOperator}, the operator $T_{\bxi}$ does not preserve $W$-invariance unless $\bxi$ is orthogonal to $\spn(R)$. Therefore, Eq.~\eqref{EquationDunklKernelFirstOrder} only holds in the limit $\beta\to\infty$ when $\bxi$ is a linear combination of $\{\bphi_i\}_{d_R<i\leq N}$, otherwise it must be zero because $W$-invariant and non-$W$-invariant quantities cannot be identically equal.

For $d_R<j\leq N$, one has
\begin{equation}
T_{\bphi_j}=\bphi_j\cdot\bnabla+\frac{\beta}{2}\sum_{\balpha\in R_+}[\bphi_j\cdot\balpha] \kappa(\balpha) \frac{1-\sigma_{\balpha}}{\balpha\cdot\bx}=\bphi_j\cdot\bnabla,
\end{equation}
because $\bphi_j\cdot\balpha=0\ \forall \balpha\in R,$ $d_R<j\leq N$. Then, when $\beta\to\infty$, 
\begin{equation}\label{DunklKernelFirstOrderDifferentialEquation}
\bphi_j\cdot\bnabla g(\bx,\by)=[\bphi_j\cdot\by] g(\bx,\by)
\end{equation}
for every $d_R<j\leq N$. Consider the orthogonal map $Z$ such that $[Z]_{ij}=\zeta_{ij}=[\bphi_j]_i$ and define $\bu,\bv\in\RR^N$ such that
\begin{IEEEeqnarray}{rCl}
\bx=Z\bu\ \Longleftrightarrow\ \bu=Z^T\bx,\label{NiceOrthogonalChangeOfBasisX}\\
\by=Z\bv\ \Longleftrightarrow\ \bv=Z^T\by\label{NiceOrthogonalChangeOfBasisY}.
\end{IEEEeqnarray}
Define $g_Z(\bu,\bv)=g(\bx,\by)$. Then, Eq.~\eqref{DunklKernelFirstOrderDifferentialEquation} becomes
\begin{equation}
\frac{\partial}{\partial u_j} g_Z(\bu,\bv)=
\begin{cases}\label{EquationSimpleExponential}
v_j g_Z(\bu,\bv) & \text{if }d_R<j\leq N,\\
0 & \text{otherwise.}
\end{cases}
\end{equation}
Keeping in mind the condition $g(\bzero,\by)=1$, the solution is
\begin{equation}
g_Z(\bu,\bv)=\exp\Big[\sum_{d_R<i\leq N}u_iv_i\Big],
\end{equation}
and the proof is completed by transforming $\bu$ and $\bv$ back into $\bx$ and $\by$.
\end{proof}

The freezing limit of the Dunkl kernel for root systems of full rank takes a different form, and it is non-trivial only when its arguments are scaled by a factor of $\sqrt{\beta}$.
\begin{lemma}\label{FreezingLimitDunklKernelFullRank}
For root systems with $d_R=N$, $V_\infty\rme^{\bx\cdot\by}=1$. Furthermore,
\begin{equation}\label{EquationFreezingLimitKernelFullRank}
\lim_{\beta\to\infty}V_\beta \rme^{\sqrt{\beta}\bx\cdot\by}=\exp\Big[\frac{x^2y^2}{2\gamma}\Big].
\end{equation}
\end{lemma}

\begin{proof}
The first part of the statement follows from Lemma~\ref{PropositionFirstOrderWInvariantOperator} and the normalization of the intertwining operator ($V_\beta 1=1\ {}^\forall \beta$). Thus, in the expansion
\begin{equation}\label{DunklKernelExpansion}
V_\beta \rme^{\bx\cdot\by}=\sum_{n=0}^\infty V_\beta\frac{(\bx\cdot\by)^n}{n!},
\end{equation}
all terms vanish as $\beta\to\infty$ except for the zeroth order term which is equal to one.

To prove the second part of the statement, the decay with $\beta$ of each of the terms in this expansion is derived. By Lemma~\ref{LemmaFirstOrderV}, the first order term is
\begin{equation}\label{DunklKernelLinearTerm}
V_\beta\bx\cdot\by=\frac{\bx\cdot\by}{1+\beta\gamma/N}\stackrel{\beta\text{ large}}{\approx}\frac{N\bx\cdot\by}{\beta\gamma}\sim\frac{1}{\beta}.
\end{equation}
By Lemma~\ref{PropositionVInfinityWInvariance}, the freezing limit eliminates the non-$W$-invariant part of $V_\beta\exp(\bx\cdot\by)$ faster than its $W$-invariant part. Consequently, the slowest decay for each of the terms in Eq.~\eqref{DunklKernelExpansion} is obtained by using the Dunkl Laplacian, which relates higher-order terms with lower-order terms while conserving their $W$-invariance (Lemma~\ref{PropositionSecondOrderWInvariantOperator}).

In general, each term in the expansion~\eqref{DunklKernelExpansion} satisfies
\begin{multline}\label{EquationSetupFrozenDunklKernelFullRank}
\frac{y^2}{\beta}V_\beta\frac{(\bx\cdot\by)^{n-2}}{(n-2)!}=\Big[\frac{1}{\beta}\Delta+\sum_{\balpha\in R_+}\kappa(\balpha)\Big(\frac{\balpha\cdot\bnabla}{\balpha\cdot\bx}-\frac{\alpha^2}{2}\frac{1-\sigma_{\balpha}}{(\balpha\cdot\bx)^2}\Big)\Big]V_\beta\frac{(\bx\cdot\by)^{n}}{n!}
\end{multline}
for $n>1$. Here, the mathematical induction method is used. Assume that
\begin{equation}\label{DunklKernelBetaAsymptotics}
V_\beta\frac{(\bx\cdot\by)^{2m}}{(2m)!}\sim\frac{1}{\beta^m}\quad\text{and}\quad V_\beta\frac{(\bx\cdot\by)^{2m+1}}{(2m+1)!}\sim\frac{1}{\beta^{m+1}},
\end{equation}
and note that these assumptions hold for $m=0$. Because spatial partial derivatives and $\sigma_{\balpha}$ do not have an effect on $\beta$, one may write
\begin{multline}
\sum_{\balpha\in R_+}\kappa(\balpha)\Big(\frac{\balpha\cdot\bnabla}{\balpha\cdot\bx}-\frac{\alpha^2}{2}\frac{1-\sigma_{\balpha}}{(\balpha\cdot\bx)^2}\Big)V_\beta\frac{(\bx\cdot\by)^{n}}{n!}\\=\frac{1}{\beta}\Big[y^2V_\beta\frac{(\bx\cdot\by)^{n-2}}{(n-2)!}-\Delta V_\beta\frac{(\bx\cdot\by)^{n}}{n!}\Big]\\
\stackrel{\beta\text{ large}}{\approx}\frac{y^2}{\beta}V_\beta\frac{(\bx\cdot\by)^{n-2}}{(n-2)!}\sim
\begin{cases}
\frac{1}{\beta^{m+1}}&\ \text{for }n=2(m+1),\\
\frac{1}{\beta^{m+2}}&\ \text{for }n=2(m+1)+1.
\end{cases}
\end{multline}
By induction, Eq.~\eqref{DunklKernelBetaAsymptotics} holds for $m\geq 0$. Then, it follows that
\begin{equation}
V_\beta\frac{\beta^m(\bx\cdot\by)^{2m}}{(2m)!}
\end{equation}
converges to a non-zero value as $\beta\to\infty$ and that 
\begin{equation}
V_\beta\frac{\beta^{m+1/2}(\bx\cdot\by)^{2m+1}}{(2m+1)!}\sim\frac{1}{\sqrt{\beta}}\stackrel{\beta\to\infty}{\longrightarrow}0.
\end{equation}
In other words, the terms of odd orders vanish in the limit. 

Define the freezing limit of the scaled even terms of the expansion~\eqref{DunklKernelExpansion} by
\begin{equation}
L_m(\bx,\by):=\lim_{\beta\to\infty}V_\beta\frac{\beta^{m}(\bx\cdot\by)^{2m}}{(2m)!}.
\end{equation}
These functions, by Lemma~\ref{PropositionVInfinityWInvariance}, are $W$-invariant. Multiplying Eq.~\eqref{EquationSetupFrozenDunklKernelFullRank} by $\beta^{m}$ with $n=2m$ gives
\begin{multline}
y^2 V_\beta\frac{\beta^{m-1}(\bx\cdot\by)^{2(m-1)}}{(2(m-1))!}\\=\Big[\frac{1}{\beta}\Delta+\sum_{\balpha\in R_+}\kappa(\balpha)\Big(\frac{\balpha\cdot\bnabla}{\balpha\cdot\bx}-\frac{\alpha^2}{2}\frac{1-\sigma_{\balpha}}{(\balpha\cdot\bx)^2}\Big)\Big]V_\beta\frac{\beta^{m}(\bx\cdot\by)^{2m}}{(2m)!}.
\end{multline}
Taking the freezing limit of this equation yields
\begin{equation}\label{DifferentialEquationFrozenDunklKernelFullRankExpansion}
y^2L_{m-1}(\bx,\by)=\sum_{\balpha\in R_+}\kappa(\balpha)\frac{\balpha\cdot\bnabla L_{m}(\bx,\by)}{\balpha\cdot\bx}.
\end{equation}
This equation has the boundary condition
\begin{equation}\label{EquationBoundaryConditionFrozenDunklKernelFullRank}
L_m(\bzero,\by)=\delta_{0,m}.
\end{equation}
Let us assume the following solution,
\begin{equation}
L_m(\bx,\by)=\frac{1}{m!}\Big(\frac{x^2y^2}{2\gamma}\Big)^m.
\end{equation}
It satisfies the boundary condition \eqref{EquationBoundaryConditionFrozenDunklKernelFullRank}, and inserting it into Eq.~\eqref{DifferentialEquationFrozenDunklKernelFullRankExpansion} gives
\begin{multline}
\sum_{\balpha\in R_+}\kappa(\balpha)\frac{\balpha\cdot\bnabla L_{m}(\bx,\by)}{\balpha\cdot\bx}=L_{m-1}(\bx,\by)\frac{y^2}{\gamma}\sum_{\balpha\in R_+}\kappa(\balpha)=y^2L_{m-1}(\bx,\by)
\end{multline}
for all $m>0$. Thus, summing up over $m$ the Lemma is proved, i.e.,
\begin{equation}
\lim_{\beta\to\infty}V_\beta\rme^{\sqrt{\beta}\bx\cdot\by}=\sum_{m=0}^\infty L_m(\bx,\by)=\exp\Big(\frac{x^2y^2}{2\gamma}\Big).\qedhere
\end{equation}
\end{proof}

In addition to this lemma, it is possible to consider Eq.~\eqref{EquationFreezingLimitKernelFullRank} when $\beta$ is large but finite. By Lemma~\ref{PropositionVInfinityWInvariance}, the non-$W$-invariant part of $V_\beta \exp[\sqrt{\beta}\bx\cdot\by]$ can be neglected. The result of applying the Dunkl Laplacian to the scaled Dunkl kernel yields the equation
\begin{equation}\label{DifferentialEquationApproximationFrozenKernel}
y^2V_\beta\rme^{\sqrt{\beta}\bx\cdot\by}=\Big[\frac{1}{\beta}\Delta+\sum_{\balpha\in R_+}\kappa(\balpha)\Big(\frac{\balpha\cdot\bnabla}{\balpha\cdot\bx}\Big)\Big]V_\beta\rme^{\sqrt{\beta}\bx\cdot\by}
\end{equation}
after dividing by $\beta$. In view of Eq.~\eqref{EquationFreezingLimitKernelFullRank}, the form $\exp[ x^2y^2/2\eta]$ with $\eta>0$ is a reasonable approximation to the solution of this equation for large $\beta$. Inserting this form in Eq.~\eqref{DifferentialEquationApproximationFrozenKernel} gives the following equation for $\eta$,
\begin{equation}
\eta=\frac{1}{2\beta}[N+\beta\gamma+\sqrt{(N+\beta\gamma)^2+4\beta x^2y^2}]\stackrel{\beta\text{ large}}{\approx} \frac{\gamma}{2}+\sqrt{\frac{\gamma^2}{4}+\frac{x^2y^2}{\beta}}.
\end{equation}
From this relation, it follows that 
\begin{IEEEeqnarray}{rCl}
\lim_{\beta\to\infty}\eta&=&\gamma,\label{CorrectedGammaAtInfinity}\\
\eta&\stackrel{\beta,x\text{ large}}{\approx}&x y/\sqrt{\beta}.
\end{IEEEeqnarray}
Thus, for $\bx$ finite
\begin{equation}\label{ApproximationFrozenKernel}
V_\beta \exp[\sqrt{\beta}\bx\cdot\by]\stackrel{\beta\text{ large}}{\approx}\exp\Big[\frac{x^2y^2}{2(\gamma+\varepsilon_\beta)}\Big],
\end{equation}
where $\varepsilon_\beta$ is given by
\begin{equation}\label{EquationBehaviorEpsilonBeta}
\varepsilon_\beta=
\begin{cases}
x^2y^2/(\beta\gamma)&\text{when }x\ll\sqrt{\beta},\\
xy/\sqrt{\beta}&\text{when }x\gg\sqrt{\beta}.
\end{cases}
\end{equation}

The arguments of the Dunkl kernel that appears inside the integral in Eq.~\eqref{ScaledProcessDistributionFreezing} are scaled as $\bv(\beta/t)^{1/2}$. Thus, it is necessary to estimate an expression for this scaled Dunkl kernel to prove Thm.~\ref{TheoremFreezingLimit}. To estimate this limit one must be careful when replacing $\bx$ by $\sqrt{\beta}\bx$ in Eq.~\eqref{EquationFreezingLimitDunklKernelNonFullRank}. Equation~\eqref{EquationFreezingLimitKernelFullRank} results from such a scaling, and it contains $\beta$-independent terms. Therefore, it is expected that replacing $\bx$ by $\sqrt{\beta}\bx$ in Eq.~\eqref{EquationFreezingLimitDunklKernelNonFullRank} produces additional terms which correspond to $\bx_\parallel$ and $\by_\parallel$.
\begin{lemma}\label{FreezingLimitDunklKernel}
The Dunkl kernel in the freezing regime is given by the expression
\begin{equation}\label{DunklKernelGeneralFreezing}
V_\beta\rme^{\sqrt{\beta}\bx\cdot\by}\stackrel{\beta\text{ large}}{\approx}\exp\Big[\sqrt{\beta}\bx_\perp\cdot\by_\perp+\frac{x_\parallel^2y_\parallel^2}{2(\gamma+\varepsilon_\beta)}\Big].
\end{equation}
\end{lemma}

\begin{proof}
Consider the scaled Dunkl kernel
\begin{equation}
E(\bx,\by):=V_\beta\rme^{\sqrt{\beta}\bx\cdot\by}.
\end{equation} 
By definition, $E(\bx,\by)$ must satisfy the following two equations:
\begin{IEEEeqnarray}{rCl}
T_{\bxi}E(\bx,\by)&=&\sqrt{\beta}\bxi\cdot\by E(\bx,\by),\label{EquationDunklOperatorOnScaledKernel}\\
\sum_{i=1}^N T_i^2 E(\bx,\by)&=&\beta y^2 E(\bx,\by).\label{EquationDunklLaplacianOnScaledKernel}
\end{IEEEeqnarray}

Let us transform the vectors $\bx$ and $\by$ into $\bu$ and $\bv$ by using Eqs.~\eqref{NiceOrthogonalChangeOfBasisX} and \eqref{NiceOrthogonalChangeOfBasisY}. Define the transformed Dunkl kernel by
\begin{equation}
E_Z(\bu,\bv):=E(\bx,\by).
\end{equation}
The various operators that appear in the Dunkl operators are transformed as follows: the directional derivative becomes
\begin{IEEEeqnarray}{rCl}
\bxi\cdot\bnabla E(\bx,\by)&=&(Z^T\bxi)\cdot\bnabla_u E_Z(\bu,\bv),
\end{IEEEeqnarray}
and the reflection operator turns into
\begin{multline}
\sigma_{\balpha} E(\bx,\by)=E(\sigma_{\balpha} \bx,\by)=E_Z(Z^T\sigma_{\balpha} Z \bu,\bv)\\=E_Z(\sigma_{Z^T\balpha}\bu,\bv)=\sigma_{Z^T\balpha}E_Z(\bu,\bv).
\end{multline}
Using these transformations, the Dunkl operator $T_{\bxi}$ becomes
\begin{multline}
T_{\bxi}E(\bx,\by)=(Z^T\bxi)\cdot\bnabla_u E_Z(\bu,\bv)+\frac{\beta}{2}\sum_{\balpha\in R_+}[\bxi\cdot\balpha]\kappa(\balpha)\frac{1-\sigma_{Z^T \balpha}}{(Z^T\balpha)\cdot\bu}E_Z(\bu,\bv)\\
=(Z^T\bxi)\cdot\bnabla_u E_Z(\bu,\bv)+\frac{\beta}{2}\sum_{\balpha_Z\in R_{Z+}}[(Z^T\bxi)\cdot\balpha_Z]\kappa_Z(\balpha_Z)\frac{1-\sigma_{\balpha_Z}}{\balpha_Z\cdot\bu}E_Z(\bu,\bv).
\end{multline}
In the last line, the root system $R_{Z}$ is given by
\begin{equation}
R_Z=\{\balpha_Z=Z^T\balpha : \balpha\in R\},
\end{equation}
and its corresponding multiplicity function is defined so that $\kappa_Z(\balpha_Z)=\kappa(\balpha)$.

Because the last $N-d_R$ vectors of the basis $\{\bphi_i\}_{i=1}^N$ are orthogonal to $\spn(R)$, $[\balpha_Z]_j=0$ for $d_R<j\leq N$:
\begin{equation}\label{AlphaZPerpendicularComponents}
[\balpha_Z]_j=[Z^T\balpha]_j=\sum_{i=1}^N[Z]_{ij}\alpha_i=\sum_{i=1}^N[\bphi_j]_i\alpha_i=\bphi_j\cdot\balpha=0.
\end{equation}
Thus, if $\bxi=\bphi_i$, then $Z^T\bxi=\be_i$ and the Dunkl operator becomes
\begin{multline}
T_{\bphi_i}E(\bx,\by)=\\
T_{Z,i} E_Z(\bu,\bv)=
\Big[\frac{\partial}{\partial u_i}+\frac{\beta}{2}\sum_{\balpha_Z\in R_{Z+}}\alpha_{Z,i}\kappa_Z(\balpha_Z)\frac{1-\sigma_{\balpha_Z}}{\balpha_Z\cdot\bu}\Big]E_Z(\bu,\bv)\label{TransformedDunklOperator}
\end{multline}
when $1\leq i\leq d_R$, and the variables $u_{d_R+1},\ldots,u_N$ do not appear in the dot products $\balpha_Z\cdot\bu$.

When $d_R<i\leq N$, the Dunkl operators become partial derivatives,
\begin{equation}\label{TransformedDunklOperatorPerpendicular}
T_{Z,i} E_Z(\bu,\bv)=\frac{\partial}{\partial u_i}E_Z(\bu,\bv).
\end{equation}
Because of the property \eqref{AlphaZPerpendicularComponents}, the Dunkl operators $T_{Z,i}$ act only on the space $\spn(R)$ for $1\leq i\leq d_R$. Thus, one may use the method of separation of variables.

Define $F_Z(\bu_\parallel,\bv_\parallel)$ and $G_Z(\bu_\perp,\bv_\perp)$ with the vectors $\bu_\parallel=(u_1,\ldots,u_{d_R})^T$, $\bu_\perp=(u_{d_R+1},\ldots,u_{N})^T$ and similar expressions for $\bv_\parallel$ and $\bv_\perp$, such that
\begin{equation}
E_Z(\bu,\bv)=F_Z(\bu_\parallel,\bv_\parallel)G_Z(\bu_\perp,\bv_\perp).
\end{equation}
Then, for $d_R< i\leq N$, Eq.~\eqref{EquationDunklOperatorOnScaledKernel} reads,
\begin{equation}
\frac{\partial}{\partial u_i}G_Z(\bu_\perp,\bv_\perp)=\sqrt{\beta}\bphi_i\cdot(Z^T\bv_\perp)G_Z(\bu_\perp,\bv_\perp)=\sqrt{\beta}v_iG_Z(\bu_\perp,\bv_\perp).
\end{equation}
This is nothing but Eq.~\eqref{EquationSimpleExponential} scaled up by a factor of $\sqrt{\beta}$, 
\begin{equation}\label{SolutionGZ}
G_Z(\bu_\perp,\bv_\perp)=\exp\Big[\sqrt{\beta}\sum_{d_R<i\leq N}u_iv_i\Big]=\exp[\sqrt{\beta}\bu_\perp\cdot\bv_\perp].
\end{equation}
The corresponding solution for $F_Z(\bu_\parallel,\bv_\parallel)$ is obtained as in Lemma~\ref{FreezingLimitDunklKernelFullRank}. 

Because the Dunkl Laplacian is independent of the orthogonal basis (see Eq.~\eqref{EquationDunklLaplacianBaseIndependence}), one may write Eq.~\eqref{EquationDunklLaplacianOnScaledKernel} as
\begin{multline}
\beta v^2 E_Z(\bu,\bv)=\beta y^2 E(\bx,\by)=\sum_{i=1}^N T_{\bphi_i}^2 E(\bx,\by)=\sum_{i=1}^N T_{Z,i}^2 E_Z(\bu,\bv)\\
=\sum_{i=1}^{d_R} T_{Z,i}^2 E_Z(\bu,\bv)+\sum_{i=d_R+1}^N \frac{\partial^2}{\partial u_i^2} E_Z(\bu,\bv).
\end{multline}
Because $v^2=v_\parallel^2+v_\perp^2$,
\begin{multline}
\frac{1}{F_Z(\bu_\parallel,\bv_\parallel)}\sum_{i=1}^{d_R} T_{Z,i}^2 F_Z(\bu_\parallel,\bv_\parallel)-\beta v_\parallel^2=\\
\beta v_\perp^2 - \frac{1}{G_Z(\bu_\perp,\bv_\perp)}\sum_{i=d_R+1}^N \frac{\partial^2}{\partial u_i^2} G_Z(\bu_\perp,\bv_\perp)
=c,
\end{multline}
where $c$ is a constant. From Eq.~\eqref{SolutionGZ} it is found that $c=0$, so the equation for $F_Z(\bu_\parallel,\bv_\parallel)$ is
\begin{equation}
\sum_{i=1}^{d_R} T_{Z,i}^2 F_Z(\bu_\parallel,\bv_\parallel)=\beta v_\parallel^2 F_Z(\bu_\parallel,\bv_\parallel).
\end{equation}
By Eq.~\eqref{ApproximationFrozenKernel}, for large values of $\beta$, 
\begin{equation}
F_Z(\bu_\parallel,\bv_\parallel)\stackrel{\beta\text{ large}}{\approx}\exp\Big[\frac{u_\parallel^2v_\parallel^2}{2(\gamma+\varepsilon_\beta)}\Big],
\end{equation}
where $\lim_{\beta\to\infty}\varepsilon_\beta=0$. Reassembling $E(\bx,\by)$ yields the result.
\end{proof}

With the asymptotic form in Eq.~\eqref{DunklKernelGeneralFreezing} and using Eqs.~\eqref{PotentialR} and \eqref{PartitionR}, one may rewrite Eq.~\eqref{ScaledProcessDistributionFreezing} as follows,
\begin{multline}
f(t,\sqrt{\beta t}\bv)(\beta t)^{N/2}\ud\bv\\\stackrel{\beta\text{ large}}{\approx}\frac{\rme^{-\beta F_R(\bv,\kappa)}}{z_\beta}\int_{\RR^N}\rme^{-x^2/2t}\exp\Big[\sqrt{\frac{\beta}{t}}\bx_{\perp}\cdot\bv_{\perp}+\frac{x_\parallel^2v_\parallel^2}{2(\gamma+\varepsilon_\beta) t}\Big]\mu(\bx)\ud\bx\ud\bv.
\end{multline}

\section{Steady-state distribution as $\beta\to\infty$}

In this section, the objective will be to examine the behavior of the steady-state distribution
\begin{equation}
f_R(\bv,\beta)=\frac{\rme^{-\beta F_R(\bv,\kappa)}}{z_\beta}
\end{equation}
in the freezing regime. Intuitively, it is clear that as $\beta$ grows, this distribution will take large values near the minima of $F_R(\bv,\kappa)$ and small values everywhere else. Denote any vector where $F_R(\bv,\kappa)$ attains a minimum by $\bos$. Then, one has the following result.
\begin{lemma}\label{SteadyStateDistributionFreezingLimit}
In the freezing regime, the distribution $f_R(\bv,\beta)$ has the form
\begin{equation}
f_R(\bv,\beta)\approx \frac{\beta^{N/2}\det H}{\pi^{N/2}|W|}\sum_{\rho\in W}\rme^{-\beta (\bv-\rho\bos)^T H (\bv-\rho\bos)}\stackrel{\beta\to\infty}{\longrightarrow}\frac{1}{|W|}\sum_{\rho\in W}\delta^{(N)}(\bv-\rho\bos).
\end{equation}
\end{lemma}

\begin{proof}
When $\beta$ takes on sufficiently large values, one may use a saddle-point approximation to calculate $z_\beta$. This, of course, requires knowledge of the extrema of $F_R(\bv,\kappa)$ which occur at the solutions of
\begin{equation}\label{ExtremaPotentialR}
\frac{\partial}{\partial v_i}F_R(\bv,\kappa)=v_i-\sum_{\balpha\in R_+}\frac{\kappa(\balpha)}{\balpha\cdot\bv}\alpha_i=0,\ 1\leq i\leq N.
\end{equation}

Denote one solution vector of these equations by $\bos$,
\begin{equation}
\bos=\sum_{\balpha\in R_+}\frac{\kappa(\balpha)}{\balpha\cdot\bos}\balpha.
\end{equation}
Therefore, the vector $\bos$ belongs to the linear envelope of the root system. Note that $s^2=\gamma$ because of Eq.~\eqref{DefinitionParameterGamma}, 
\begin{equation}
s^2=\bos\cdot\bos=\sum_{\balpha\in R_+}\frac{\kappa(\balpha)}{\balpha\cdot\bos}\bos\cdot\balpha=\sum_{\balpha\in R_+}\kappa(\balpha)=\gamma.
\end{equation}

The elements of the Hessian matrix $H(\bv)$ of $F_R(\bv,\kappa)$ are given by
\begin{equation}
[H(\bv)]_{ij}=\frac{\partial^2}{\partial v_j\partial v_i}F_R(\bv,\kappa)=\delta_{ij}+\sum_{\balpha\in R_+}\frac{\kappa(\balpha)}{(\balpha\cdot\bv)^2}\alpha_i\alpha_j.
\end{equation}
$H(\bv)$ is a positive definite matrix for $\bv\cdot\balpha\neq 0$, because for $\bx\in\RR^N$,
\begin{equation}
\sum_{1\leq i,j\leq N}x_ix_j\frac{\partial^2}{\partial v_j\partial v_i}F_R(\bv,\kappa)=x^2+\sum_{\balpha\in R_+}\frac{\kappa(\balpha)}{(\balpha\cdot\bv)^2}(\balpha\cdot\bx)^2\geq 0.
\end{equation}
Therefore, all the extrema of $F_R(\bv,\kappa)$ are minima.

Taking $\rho\in W$, one has
\begin{equation}
\rho \bos=\sum_{\balpha\in R_+}\frac{\kappa(\balpha)}{\balpha\cdot\bos}\rho\balpha=\sum_{\balpha^\prime\in R_+}\frac{\kappa(\balpha^\prime)}{\rho^{-1}\balpha^\prime\cdot\bos}\balpha^\prime=\sum_{\balpha^\prime\in R_+}\frac{\kappa(\balpha^\prime)}{\balpha^\prime\cdot\rho\bos}\balpha^\prime.
\end{equation}
Here, the substitution $\balpha^\prime=\rho\balpha$ has been carried out. This means that $\rho\bos$ is also a solution of Eq.~\eqref{ExtremaPotentialR}, and consequently, its solutions are related with each other by an element of the reflection group $W$. Therefore, there are $|W|$ solutions of Eq.~\eqref{ExtremaPotentialR}. This set of solutions is called the peak set of the reflection group $W$. Because $F_R(\bv,\kappa)$ is $W$-invariant, all the minima have the same value.  

One can approximate $z_\beta$ for large values of $\beta$ as follows.
\begin{equation}
z_\beta=\int_{\RR^N}\rme^{-\beta F_R(\bv,\kappa)}\ud\bv\approx|W|\rme^{-\beta F_R(\bos,\kappa)}\int_{\RR^N}\exp[-\beta \bor^TH(\bos)\bor/2]\ud\bor,
\end{equation}
where $\bor=\bv-\bos$. Because the Hessian matrix is positive definite and symmetric, all of its eigenvalues at the minima are positive. Therefore, one can use an orthogonal coordinate transformation to solve this Gaussian integral. The result is
\begin{equation}
z_\beta\approx |W|\rme^{-\beta F_R(\bos,\kappa)} \prod_{i=1}^N\sqrt{\frac{2\pi}{\beta \lambda_i}},
\end{equation}
where the $\{\lambda_i\}_{i=1}^N$ are the eigenvalues of $H(\bos)$. Then, the steady-state distribution is approximated as a sum of Gaussians as shown below,
\begin{equation}\label{SteadyStateDistributionLargeBetaApproximation}
f_R(\bv,\beta)\approx\frac{\beta^{N/2}\sqrt{\det H}}{(2\pi)^{N/2}|W|}\sum_{\rho\in W}\rme^{-\beta (\bv-\rho\bos)^T H (\bv-\rho\bos)/2}.
\end{equation}
Note that the approximate distribution is normalized. 

Finally, as $\beta\to\infty$ each of the Gaussians tends to a delta function in the sense of distributions. Therefore, the steady-state distribution tends to a sum of delta functions centered at the peak set of $W$. 
\end{proof}

In order to prove Thm.~\ref{TheoremFreezingLimit}, it is necessary to take into account the value of the integral in Eq.~\eqref{ScaledProcessDistributionFreezing} for $\beta\to\infty$. This value can be estimated by investigating the variance of the steady-state distribution before taking the freezing limit. From Eq.~\eqref{SteadyStateDistributionLargeBetaApproximation},
\begin{multline}\label{DeviationFromS}
\int_{\RR^N}r^2f_R(\bor+\bos,\beta)\ud\bor\approx\frac{1}{|W|}\sqrt{\frac{\beta}{2\pi}}\sum_{j=1}^N\sqrt{\lambda_j}\int_{\RR}r_j^2\rme^{-\beta\lambda_jr_j^2/2}\ud r_j\\
=\frac{1}{\beta|W|}\sum_{j=1}^N\lambda_j^{-1},
\end{multline}
which means that the typical deviation from any of the elements of the peak set of $W$ is of the order of $\beta^{-1/2}$.

\section{Proof of Theorem~\ref{TheoremFreezingLimit}}\label{ProofOfTheoremFreezingLimit}

First, the large-$\beta$ asymptotics of the scaled distribution are considered. From Lemmas~\ref{FreezingLimitDunklKernel} and \ref{SteadyStateDistributionFreezingLimit}, one may write Eq.~\eqref{ScaledProcessDistributionFreezing} in the approximated form
\begin{multline}
f(t,\sqrt{\beta t}\bv)(\beta t)^{N/2}\ud\bv\approx \frac{\beta^{N/2}\sqrt{\det H}}{(2\pi)^{N/2}|W|}\sum_{\rho\in W}\rme^{-\beta (\bv-\rho\bos)^T H (\bv-\rho\bos)/2}\\
\times\int_{\RR^N}\rme^{-x^2/2t}\exp\Big[\sqrt{\frac{\beta}{t}}\bx_{\perp}\cdot\bv_{\perp}+\frac{x_\parallel^2v_\parallel^2}{2(\gamma+\varepsilon_\beta) t}\Big]\mu(\bx)\ud\bx\ud\bv.
\end{multline}
If the root system is of full rank, then $x_\perp=v_\perp=0$, $x_\parallel=x$, $v_\parallel=v$, and the integral over $\bx$ tends to
\begin{equation}\label{InitialDistributionIntegralFreezingLimitFullRank}
\int_{\RR^N}\exp\Big[-\frac{x^2}{2t}\Big(1-\frac{v^2}{\gamma+\varepsilon_\beta}\Big)\Big]\mu(\bx)\ud\bx\stackrel{\beta\to\infty}{\longrightarrow}1.
\end{equation}
The reason for this is that the approximating Gaussians of $f_R(\bv,\beta)$ have a variance of order $\beta^{-1}$, so $v^2=\gamma+O(\beta^{-1})$, as given by Eq.~\eqref{DeviationFromS}, while $\varepsilon_\beta\approx xv/\sqrt{\beta}$ at large values of $x$, which guarantees the convergence of the integral. 

On the other hand, if the root system is not of full rank, all vectors orthogonal to $\spn(R)$ are eigenvectors of $H$ with eigenvalue one. Consequently,
\begin{equation}
(\bv_\parallel+\bv_\perp-\rho\bos)^T H (\bv_\parallel+\bv_\perp-\rho\bos)=v_\perp^2+(\bv_\parallel-\rho\bos)^T H (\bv_\parallel-\rho\bos),
\end{equation}
and
\begin{multline}\label{EquationFinalFreezingComment}
f(t,\sqrt{\beta t}\bv)(\beta t)^{N/2}\ud\bv\approx \frac{\beta^{d_R/2}}{(2\pi)^{d_R/2}|W|}\prod_{i=1}^{d_R}\lambda_i^{1/2}\sum_{\rho\in W}\rme^{-\beta (\bv_\parallel-\rho\bos)^T H (\bv_\parallel-\rho\bos)/2}\\
\times\frac{\beta^{(N-d_R)/2}}{(2\pi)^{(N-d_R)/2}}\int_{\RR^N}\rme^{-\beta (\bv_\perp-\bx_\perp/\sqrt{\beta t})^2/2}\exp\Big[-\frac{x^2_\parallel}{2t}\Big(1-\frac{v_\parallel^2}{\gamma+\varepsilon_\beta}\Big)\Big]\mu(\bx)\ud\bx\ud\bv.
\end{multline}
The part of the integral that depends on $\bx_\parallel$ and $\bv_\parallel$ behaves like Eq.~\eqref{InitialDistributionIntegralFreezingLimitFullRank}. This means that the only part that needs to be calculated is the part that depends on $\bx_\perp$ and $\bv_\perp$.

Consider a test function $\phi(\bv_\perp)$ and the integral
\begin{equation}
\frac{\beta^{(N-d_R)/2}}{(2\pi)^{(N-d_R)/2}}\int_{\RR^N}\phi(\bv_\perp)\rme^{-\beta (\bv_\perp-\bx_\perp/\sqrt{\beta t})^2/2}\ud\bv_\perp.
\end{equation}
By choosing $X\gg 1$ such that $X/\sqrt{\beta}\ll 1$, the integral becomes (after setting $\by=\sqrt{\beta}\bv_\perp-\bx_\perp/\sqrt{t}$)
\begin{equation}
\frac{1}{(2\pi)^{(N-d_R)/2}}\Big[\int_{y<X}+\int_{y\geq X}\Big]\phi[{\textstyle \frac{1}{\sqrt \beta}(\by+\bx_\perp/\sqrt{t})}]\rme^{-y^2/2}\ud\by.
\end{equation}
The inner region integral becomes $\phi(\bzero)$ as $\beta\to\infty$. The outer region integral tends to zero in the freezing regime as long as $\phi(\by)$ does not grow exponentially when $|\by|$ is large,
\begin{multline}
\Big|\frac{1}{(2\pi)^{(N-d_R)/2}}\int_{y\geq X}\phi[{\textstyle \frac{1}{\sqrt \beta}(\by+\bx_\perp/\sqrt{t})}]\rme^{-y^2/2}\ud\by\Big|\\
\leq \sup_{y\geq X}|\phi[{\textstyle \frac{1}{\sqrt \beta}(\by+\bx_\perp/\sqrt{t})}]\rme^{-y^2/4}| \int_{y\geq X}\frac{\rme^{-y^2/4}}{(2\pi)^{(N-d_R)/2}}\ud\by\stackrel{\sqrt{\beta}\gg X\to\infty}{\longrightarrow}0.
\end{multline}
Thus, in the sense of distributions,
\begin{equation}\label{FreezingLimitGeneralPerpendicularPart}
\frac{\beta^{(N-d_R)/2}}{(2\pi)^{(N-d_R)/2}}\rme^{-\beta (\bv_\perp-\bx_\perp/\sqrt{\beta t})^2/2}\stackrel{\beta\to\infty}{\longrightarrow}\delta^{(N-d_R)}(\bv_\perp).
\end{equation}
Finally, one obtains 
\begin{multline}
f(t,\sqrt{\beta t}\bv)(\beta t)^{N/2}\ud\bv\stackrel{\beta\to\infty}{\longrightarrow}\frac{1}{|W|}\delta^{(N-d_R)}(\bv_\perp)\sum_{\rho\in W}\delta^{(d_R)}(\bv_\parallel-\rho\bos)\ud\bv
\end{multline}
in the sense of distributions. Because $(\rho\bos)_\parallel=(\rho\bos)$ and $(\rho\bos)_\perp=\bzero$, this completes the proof.\qquad\qquad\qquad\qquad\qquad\qquad\qquad\qquad\qquad\qquad\quad\quad\quad\ \ $\square$
% !TEX encoding = UTF-8 Unicode
\chapter{Limiting regimes for the interacting Brownian motions and Bessel processes}\label{ParticularCases}

Keeping in line with the main motivation of this work, it is of interest to investigate the results discussed in the previous chapter in the case of the root systems of type $A$ and $B$. The radial Dunkl processes in these two cases correspond to the Dyson model and to the Wishart and Laguerre processes, respectively.

\section{Statement of results and numerical evidence}
Denote by $\bh_N=(h_{N,1},\ldots,h_{N,N})$ the vector of zeroes (in ascending order) of the $N$th Hermite polynomial $H_N(x)$ defined by (see, e.g., \cite{szego})
\begin{equation}
H_N(x)=(-1)^N \rme^{x^2}\frac{\ud^N}{\ud x^N}(\rme^{-x^2}).\label{HermitePolynomial}
\end{equation}
Similarly, denote by $\bl^{(\alpha)}_{N}=(l^{(\alpha)}_{N,1},\ldots,l^{(\alpha)}_{N,N})$ the vector of zeroes of the associated Laguerre polynomial of parameter $\alpha$, $L_N^{(\alpha)}(x)$, defined by
\begin{equation}
\rme^{-x}x^{\alpha}L_N^{(\alpha)}(x)=\frac{1}{N!}\Big(\frac{\ud}{\ud x}\Big)^N(\rme^{-x}x^{N+\alpha})
\end{equation}
in ascending order. In addition, recall the function $F_R$, Eq.~\eqref{PotentialR}, for $R=A$ and $B$,
\begin{IEEEeqnarray}{rCl}
F_A(\bv)&=&\frac{v^2}{2}-\sum_{1\leq i<j\leq N}\log|v_j-v_i|,\label{EquationLogSteadyDistributionTypeA}\\
F_B(\bv,\nu)&=&\frac{v^2}{2}-\frac{2\nu+1}{2}\sum_{i=1}^N\log |v_i|-\sum_{1\leq i<j\leq N}\log|v_j^2-v_i^2|,\label{EquationLogSteadyDistributionTypeB}
\end{IEEEeqnarray}
and an additional function for the interacting Bessel processes,
\begin{IEEEeqnarray}{rCl}
\tilde{F}_B(\bv)&=&\frac{v^2}{2}-\frac{1}{2}\sum_{i=1}^N\log v_i^2-\frac{N}{2},\label{PotentialFTildeTypeB}
\end{IEEEeqnarray}
which will be used for the limit $\nu\to\infty$. Define also the constants $K_A$ and $K_B$ by
\begin{IEEEeqnarray}{rCl}
K_A&=&\frac{N}{4}(N-1)(1+\log 2)-\frac{1}{2}\sum_{i=1}^Ni\log i,\label{ConstantLargeBetaKA}\\
K_B&=&\frac{N}{2}(N+\nu-1/2)-\frac{1}{2}\sum_{i=1}^N i\log i\nonumber\\
&&-\frac{1}{2}\sum_{i=1}^N(\nu+i-1/2) \log (\nu+i-1/2).\label{ConstantLargeBetaKB}
\end{IEEEeqnarray}
For the following statements, it is assumed that the processes start from the arbitrary initial distributions $\mu_A(\bx)$ and $\mu_B(\bx)$, which have finite second moments.

\begin{proposition}\label{LimitingRegimesA}
For large values of $\beta$, the interacting Brownian motions follow the steady-state distribution
\begin{equation}
f_A(t,\sqrt{\beta t}\bv)(\beta t)^{N/2}\ud\bv=N!\Big(\frac{\beta}{2\pi}\Big)^{N/2}\rme^{-\beta [F_A(\bv)-K_A]}\ud\bv,\label{EquationRelaxationA}
\end{equation}
and in the freezing limit $\beta\to\infty$,
\begin{equation}
f_A(t,\sqrt{\beta t}\bv)(\beta t)^{N/2}\ud\bv\stackrel{\beta\to\infty}{\longrightarrow}\delta^{(N)}(\bv-\bh_N)\ud\bv.\label{EquationFreezingA}
\end{equation}
\end{proposition}

\begin{proposition}\label{LimitingRegimesB}
For large values of $\beta$, the interacting Bessel processes follow the steady-state distribution
\begin{equation}
f_B(t,\sqrt{\beta t}\bv)(\beta t)^{N/2}\ud\bv=N!(2\beta)^{N/2}\rme^{-\beta [F_B(\bv,\nu)-K_B]}\ud\bv,\label{EquationRelaxationBBeta}
\end{equation}
and in the freezing limit $\beta\to\infty$,
\begin{equation}
f_B(t,\sqrt{\beta t}\bv)(\beta t)^{N/2}\ud\bv\stackrel{\beta\to\infty}{\longrightarrow}\delta^{(N)}(\bv-\bor_{\nu-1/2,N})\ud\bv,\label{EquationFreezingB}
\end{equation}
with $(\bor_{\nu-1/2,N})^2=\bl_{\nu-1/2,N}$. In addition, as $\nu\to\infty$,
\begin{equation}
f_B(t,\sqrt{\beta \nu t}\bv)(\beta \nu t)^{N/2}\ud\bv\stackrel{\nu\to\infty}{\longrightarrow}\prod_{i=1}^N\delta(v_i-1)\ud\bv.\label{EquationNuLimit}
\end{equation}
\end{proposition}

These results follow from Thms.~\ref{TheoremSteadyState} and \ref{TheoremFreezingLimit} with the exception of the limit $\nu\to\infty$ for the interacting Bessel processes. Through numerical simulations one may find evidence to support these claims. 

Note that the limit $\nu\to\infty$ for the interacting Bessel processes corresponds to the physical case of a QCD Dirac operator where the topological charge is very large \cite{katoritanemura04, verbaarschot94}.

\subsection{Simulations of the interacting Brownian motions}

First, consider the interacting Brownian motions. Denoting the positions of the particles in this system by $\{X_{i,t}\}_{i=1}^N$, and denoting by $\{B_{i,t}\}_{i=1}^N$ a set of $N$ independent one-dimensional Brownian motions, the SDEs for $\{X_{i,t}\}_{i=1}^N$ are given by
\begin{equation}
\ud X_{i,t}=\ud B_{i,t}+\frac{\beta}{2}\sum_{\substack{j=1\\j\neq i}}^N\frac{\ud t}{X_{i,t}-X_{j,t}}\label{SDETypeA}
\end{equation}
for $\beta> 0$ \cite{demni08B}. In the case where $\beta=0$, there are collisions between particles and there are local times associated with them which are not included in Eq.~\eqref{SDETypeA}. Keeping this in mind, the Euler method can be used to integrate these SDEs and find the particle density (or one-point correlation function). The results are as follows.

\begin{figure}[!ht]
\centering
\includegraphics[width=0.8\textwidth]{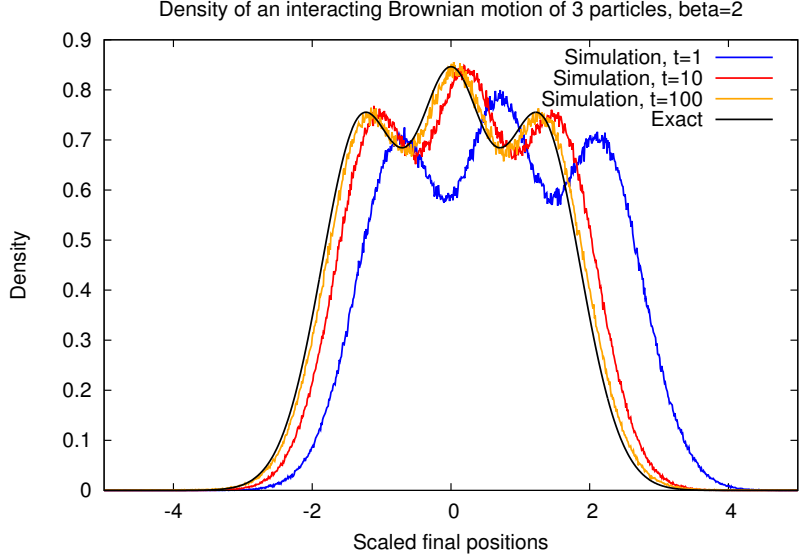}
\caption{Simulated density of three interacting Brownian motions for $\beta=2$ at varying time durations $t$, with final positions scaled down by a factor of $\sqrt{\beta t}=\sqrt{2t}$. The black solid line represents the exact density when all particles start at the origin.}\label{FigureRelaxationA}
\end{figure}

Consider the relaxation process that leads to Eq.~\eqref{EquationRelaxationA}. Figure~\ref{FigureRelaxationA} depicts the particle densities obtained by integrating the SDEs~\eqref{SDETypeA} for a system of three particles starting from the positions $\bar{\bx}_\mu=(0,1,2)^T$ with a time step of $2\times10^{-4}$ and $\beta=2$, at a resolution of $10^{-2}$. Here, the final positions were scaled down by a factor of $\sqrt{\beta t}=\sqrt{2 t}$. One can observe that as the time duration $t$ grows, the densities calculated from the simulation data tend toward the exact particle density
\begin{equation}
\sigma_N^{(A)}(y,t)=\frac{\rme^{-y^2/2t}}{2^N(N-1)!\sqrt{2\pi t}}\Big[H_N^2\Big(\frac{y}{\sqrt{2t}}\Big)-H_{N+1}\Big(\frac{y}{\sqrt{2t}}\Big)H_{N-1}\Big(\frac{y}{\sqrt{2t}}\Big)\Big],
\end{equation}
obtained from setting $x=y/\sqrt{\beta t}=y/\sqrt{2t}$ in Eq.~(6.2.10) of \cite{mehta04} (see also \cite[p.~3063]{katoritanemura04}, fourth line) and from using the Christoffel-Darboux formula \cite{szego}. This density corresponds to the case where all particles start at the origin. At $t=1$, the general shape of the simulated density function is close to the exact curve, but it is displaced because the initial position of the particles still has an effect on the state of the system. As the time duration grows, the density function converges slowly to $\sigma_N^{(A)}(y,t)$; at $t=100$, the density function shows a good agreement with the exact curve. 

In this case, Eq.~\eqref{GeneralRelaxationTime} is interpreted as follows. The initial condition considered here gives $\bar{\bx}_\mu=(0,1,2)^T$ and $s_\mu=0$. Because $A$ is a root system of rank $N-1$, one has $\bar{\bx}_{\mu\parallel}=(-1,0,1)^T$ and $\bar{\bx}_{\mu\perp}=(1,1,1)^T$. For this situation, Eq.~\eqref{GeneralRelaxationTime} requires the condition $t\gg10$ to guarantee that the system has reached the steady state.

However, it is necessary to consider the order of magnitude of the largest correction, which is given by Eq.~\eqref{AttentionSlowRelaxation}. Assuming the worst-case scenario where $t=10$, the condition for $\eta$ in Eq.~\eqref{EquationFinalConditionsEpsilonEta} reads $\eta\gg\sqrt{0.1}\approx 0.32$. With these estimations, the leading-order correction given by Eq.~\eqref{AttentionSlowRelaxation} is of order larger than $0.7$, which is comparable to 1. This estimation of the leading-order correction is quite large, and while the curves for $t=10$ and for the exact steady-state in Fig.~\ref{FigureRelaxationA} seem somewhat close, there is still a visible difference between them. By comparison, taking $t=100$ gives a best-case estimation of $\eta=0.1$ and the correction becomes of order larger than $7\eta/\sqrt{10}\approx 0.22$, a much smaller correction. This means that Eq.~\eqref{GeneralRelaxationTime} gives a very rough estimation for the relaxation time, and that the relaxation time for a given precision may surpass the value $(s_\mu^2 +\bar{x}_\mu^2)\max[1,\beta]$ greatly.

\begin{figure}[!ht]
\centering
\includegraphics[width=0.8\textwidth]{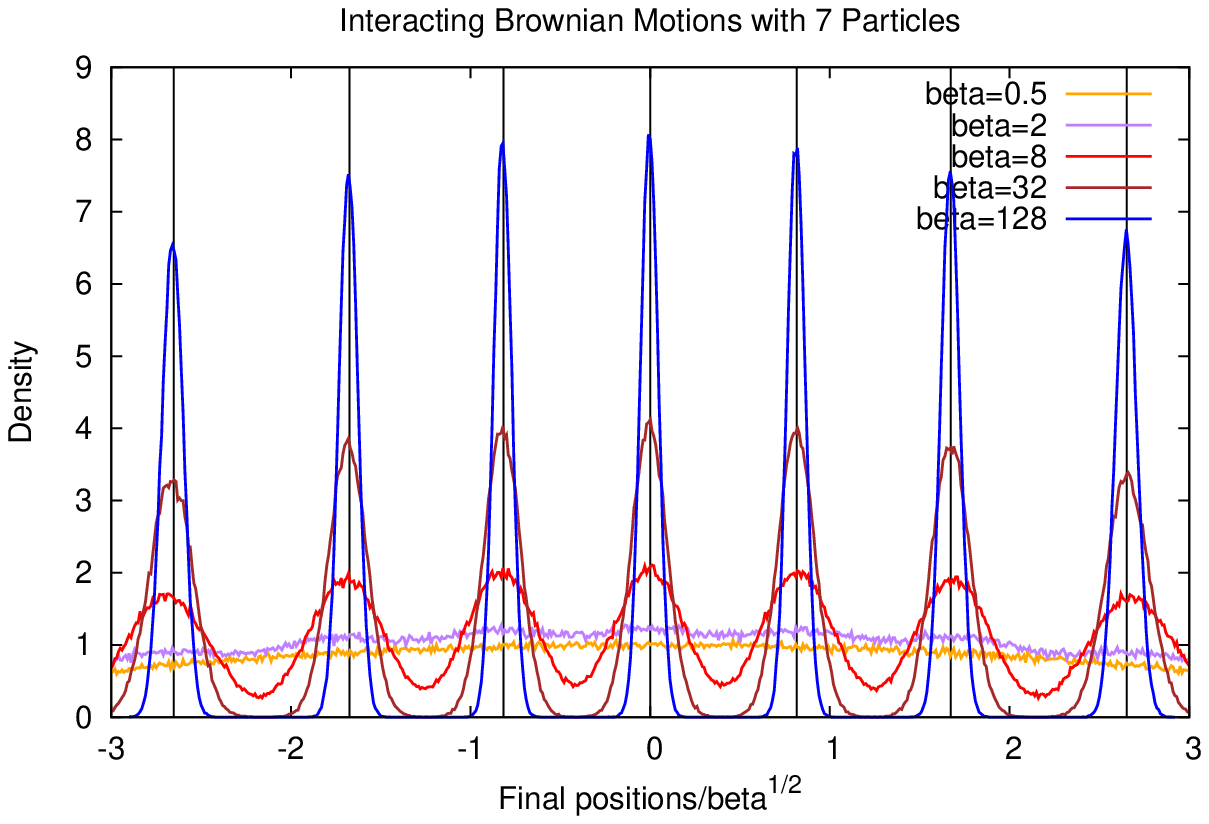}
\caption{Density of seven interacting Brownian motions for various values of $\beta$ at $t=1$, with final positions scaled down by a factor of $\sqrt{\beta t}=\sqrt{\beta}$. The vertical lines represent the zeroes of $H_7(x).$}\label{FigureFreezingA}
\end{figure}

Consider now the freezing limit $\beta\to\infty$ from Eq.~\eqref{EquationFreezingA}. Figure~\ref{FigureFreezingA} depicts the particle density of a series of interacting Brownian motions of seven particles for various values of $\beta$. To obtain these curves, the SDEs~\eqref{SDETypeA} were integrated a total of $10^6$ times for each curve, with a step size of $5\times10^{-5}$ at a resolution of $10^{-2}$ and a final time of $t=1$; the final positions were scaled down by a factor of $\sqrt{\beta t}=\sqrt{\beta}$. In order to guarantee a fast convergence to the steady state, the initial positions were chosen as $0,\pm10^{-2},\pm2\times10^{-2}$ and $\pm3\times10^{-2}$. Under these conditions, $\bar{\bx}_{\mu\perp}=\bzero$, $\bar{x}_\mu^2=2.8\times10^{-3}$ and $s_\mu=0$, which means that the interacting Brownian motions should reach the steady state for $t\gg 2.8\times10^{-3}\max[1,\beta]$, which means that for $\beta\approx 300$ or less, the distribution at $t=1$ must coincide with the steady-state distribution. This assertion is supported by the simulations performed, as the curves in Fig.~\ref{FigureFreezingA} preserve this shape at $t=2$ (not shown).

These curves are consistent with the large-$\beta$ asymptotics considered in \cite{dumitriuedelman05} after multiplying the scaled final positions by a factor of $1/\sqrt{2N}$. This fact is further evidence indicating the validity of the claim in Eq.~\eqref{EquationRelaxationA}. Furthermore, it is clear that as $\beta$ grows, the peaks in the density curves become narrower and remain centered to the vertical lines, which represent the location of the zeroes of the Hermite polynomial $H_7(x)$, $\bh_7$. This supports the claim in Eq.~\eqref{EquationFreezingA}.

\subsection{Simulations of the interacting Bessel processes}

Similarly, one may consider the SDEs for the interacting Bessel processes, denoted here by $\{Y_{i,t}\}_{i=1}^N$, for $\beta\geq 1$,
\begin{equation}
\ud Y_{i,t}=\ud B_{i,t}+\frac{\beta}{2}\Big[\frac{2\nu+1}{2}\frac{1}{Y_{i,t}}+\sum_{\substack{j=1\\j\neq i}}^N\Big(\frac{1}{Y_{i,t}-Y_{j,t}}+\frac{1}{Y_{i,t}+Y_{j,t}}\Big)\Big]\ud t.\label{SDETypeB}
\end{equation}
Note that the two terms in the sum represent a repulsive interaction between particles and their ``mirror images'' with respect to the origin, and that if the particle-particle interaction terms were absent, these SDEs would reduce to $N$ independent Bessel processes \cite{karatzasshereve91}. Note also that the processes $\{Y_{i,t}^2\}_{i=1}^N$ obey the SDEs that define the interacting squared Bessel processes which result as the eigenvalue processes of the Wishart and Laguerre processes. Therefore, the interacting Bessel processes are, for the purposes of this work, equivalent to the interacting squared Bessel processes. 

As with the interacting Brownian motions, let us consider the relaxation of the interacting Bessel processes to their steady state. To investigate this regime, the SDEs~\eqref{SDETypeB} were integrated numerically using the Euler method. The time evolution of a system of three particles starting at the positions $\bar{\bx}_\mu=(1,2,3)^T$ was simulated a total number of $10^6$ times, with $\beta 
= 2$, a time step of $2\times 10^{-4}$ and various final times; the final positions were scaled down by a factor of $\sqrt{\beta t}=\sqrt{2t}$. The resulting particle density functions (with a resolution of $10^{-2}$) are depicted in Fig.~\ref{FigureRelaxationBBeta}. In the figure, the solid black line represents the exact particle density in the case where all particles start at the origin. The curve corresponds to the function
\begin{IEEEeqnarray}{rCl}
\sigma_N^{(B)}(y,t)&=&\frac{N!}{\Gamma(\nu+N)}\Big\{N\Big[L_N^{(\nu)}\Big(\frac{y^2}{2t}\Big)\Big]^2+L_N^{(\nu)}\Big(\frac{y^2}{2t}\Big)L_{N-1}^{(\nu)}\Big(\frac{y^2}{2t}\Big)\nonumber\\
&-&(N+1)L_{N+1}^{(\nu)}\Big(\frac{y^2}{2t}\Big)L_{N-1}^{(\nu)}\Big(\frac{y^2}{2t}\Big)\Big\} \frac{2}{y}\Big(\frac{y^2}{2t}\Big)^{\nu} \rme^{-y^2/2t},
\end{IEEEeqnarray}
which is obtained from Eq.~(31) of \cite{katoritanemura04} and Eq.~(5.13) of \cite{forrester10} applied to the Laguerre case, while making the substitution $x=y^2/(\beta t)=y^2/(2t)$. 

As in Fig.~\ref{FigureRelaxationA}, the density curves obtained from the simulation data converge to the exact density as the time duration of the process grows, but they do so faster than the interacting Brownian motions. Indeed, at $t=1$ the effect of the initial positions of the particles is visible, but at $t=14$ the resulting curve is already very close to the exact density. This is in spite of the fact that Eq.~\eqref{GeneralRelaxationTime} requires $t\gg \beta(\bar{x}_\mu^2+s_\mu^2)=28$. The leading-order correction from Eq.~\eqref{AttentionSlowRelaxation} for this system and the initial condition considered is of order $2\eta/\sqrt{7}$, and because $\eta\gg t^{-1/2}=(14)^{-1/2}\approx 0.27$, this correction is of order larger than $0.2$. Note that $t=100$ was needed in the case of the interacting Brownian motions in Fig.~\ref{FigureRelaxationA} to obtain a similar correction. Therefore, it seems clear that the bound in Eq.~\eqref{GeneralRelaxationTime} may not work very well to estimate the time needed to reach the steady state in some cases. In practice, a proper estimation of the relaxation time must depend on the characteristics of the initial distribution and on the properties of the root system, because of the form of the leading-order correction from Eq.~\eqref{AttentionSlowRelaxation}.

\begin{figure}[!ht]
\centering
\includegraphics[width=0.8\textwidth]{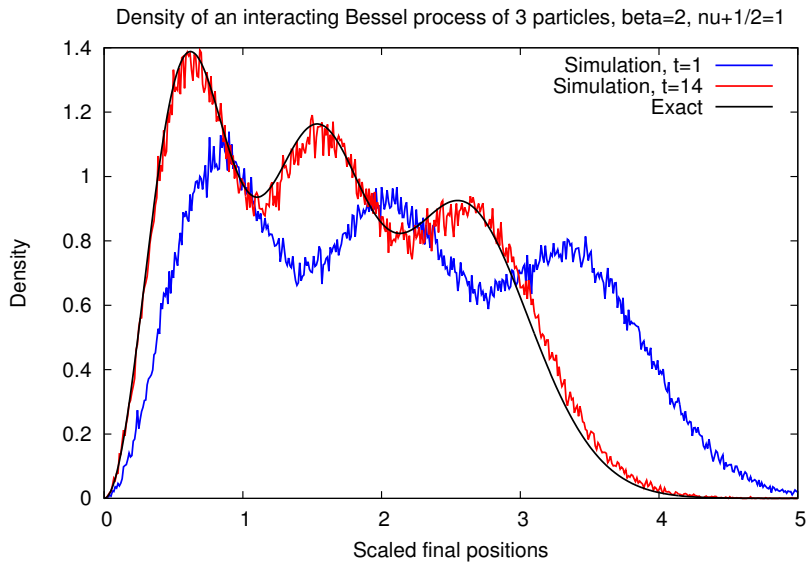}
\caption{Density of three interacting Bessel processes for $t=1$ and 14 at $\beta=2$ and $\nu=1/2$, with final positions scaled down by a factor $\sqrt{\beta t}=\sqrt{2t}$. The solid black line depicts the exact density when all particles start at the origin.}\label{FigureRelaxationBBeta}
\end{figure}

Next, the freezing limit where $\beta\to\infty$ is considered. As in Fig.~\ref{FigureFreezingA}, the initial configuration of the system is chosen so that the relaxation time is short. In this case, the initial configuration is $\bar{\bx}_\mu=(1,2,3,4,5,6,7)^T\times10^{-2}$, and the SDEs~\eqref{SDETypeB} were integrated up to $t=1$ with a step size of $2\times 10^{-4}$ and $\nu=1/2$ a total of $10^6$ times, producing a plot with a resolution of $10^{-2}$. Here, the final positions were scaled down by a factor of $\sqrt{\beta t}=\sqrt{\beta}$. The result is given in Fig.~\ref{FigureFreezingB}. From Eq.~\eqref{GeneralRelaxationTime}, it follows that the system relaxes to the steady state at $t=\beta(\bar{x}_\mu^2+s_\mu^2)=2.8\times10^{-2}\ll1$ for this particular initial configuration. Therefore, the curves depicted in the figure represent the steady-state particle density, and they are consistent with the eigenvalue density of the $\beta$-Laguerre ensembles of random matrices \cite{dumitriuedelman05}, after the variable substitution $y_i=\sqrt{\lambda_i t}$ and the parameter transformation $a=\beta(N+\nu-1/2+1/\beta)/2$. Much like in Fig.~\ref{FigureFreezingA}, one can observe that as $\beta$ grows, the peaks become narrower while staying centered around the vertical lines, which represent the square root of the zeroes of the Laguerre polynomial $L_7^{(0)}(x)$, $\bl_{0,7}$. This evidence is clearly in favor of the claim in Eq.~\eqref{EquationFreezingB}.

\begin{figure}[!ht]
\centering
\includegraphics[width=0.8\textwidth]{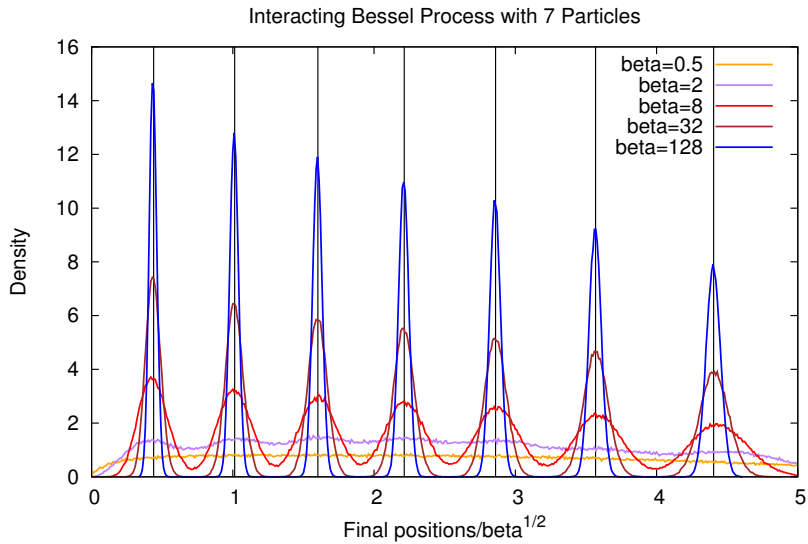}
\caption{Density of seven interacting Bessel processes for various values of $\beta$ at $t=1$ and $\nu=1/2$, with final positions scaled down by a factor of $\sqrt{\beta t}=\sqrt{\beta}$. The vertical lines represent the zeroes of $L^{(0)}_7(x).$}\label{FigureFreezingB}
\end{figure}

In the same way, one can investigate the limit $\nu\to\infty$ of the interacting Bessel processes. As before, integrating the SDEs \eqref{SDETypeB} a total of $10^6$ times, up to a time $t=1/2$ with $\beta=2$ and initial configuration $\bar{x}_{\mu,i}=i\times10^{-2}$, with time increments of $2\times10^{-4}$, Fig.~\ref{FigureLimitNuB} can be produced. The final positions were scaled down by a factor of $\sqrt{\beta \nu t}=\sqrt{\nu}$. This time, it is clear that as $\nu$ grows, all particles tend to group up at the same scaled position. This is a consequence of the fact that the interaction between particles (the two terms in the sum in Eq.~\eqref{SDETypeB}) is rendered negligible due to the large value of $\nu$, and the only part that remains is a set of independent Bessel processes with a large Bessel index.

\begin{figure}[!ht]
\centering
\includegraphics[width=0.8\textwidth]{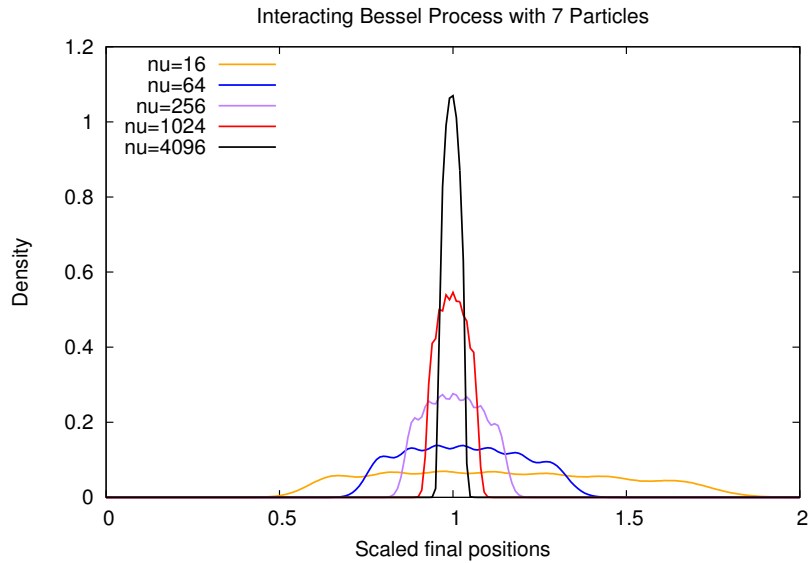}
\caption{Density of seven interacting Bessel processes for various values of $\nu$ at $t=1/2$ and $\beta=2$, with final positions scaled down by a factor of $\sqrt{\beta\nu t}=\sqrt{\nu}$.}\label{FigureLimitNuB}
\end{figure}

With the result of these simulations, the claims from Props.~\ref{LimitingRegimesA} and \ref{LimitingRegimesB} are illustrated and supported. An alternative derivation of these propositions will be tackled in the rest of this chapter.

\section{The intertwining operator for symmetric polynomials}

As stated in Chapter~\ref{general_steady}, the intertwining operator plays an important role in the relaxation process that leads to the steady-state distributions of the interacting Brownian motions and Bessel processes, as well as in their behavior in the freezing regime. Here, the expressions for the intertwining operators of type $A$ and $B$ for the case where they are applied on symmetric polynomials are derived. In addition, simple examples of their effect on symmetric polynomials are studied, and their form in the freezing limit $\beta\to\infty$ is calculated. Finally, the forms of the generalized Bessel functions of type $A$ and $B$ are calculated and shown to be consistent with the results from Chapter~\ref{general_freezing}.

\subsection{Derivation}

Using the fact that the interacting Brownian motions and Bessel processes are equivalent to the radial Dunkl processes of type $A$ and $B$, one may equate their respective transition densities to deduce that \cite{rosler98,gallardoyor05,bakerforrester97}
\begin{IEEEeqnarray}{rCl}
\sum_{\rho\in S_N}V_A\rme^{\rho\bx\cdot\by}&=&N!\FZ{2/\beta}\left(\bx,\by\right),\label{RadialDunklKernelA}\\ 
\sum_{\rho\in W_{B}}V_B \rme^{\rho\bx\cdot\by}&=&2^N N! \FF{2/\beta}\left(\frac{\beta}{2}(\nu+ N-1/2) +\frac{1}{2};\frac{(\bx)^2}{2},\frac{(\by)^2}{2}\right).\nonumber\\
\label{RadialDunklKernelB}
\end{IEEEeqnarray}
Here, the symbol $(\bx)^2$ denotes the vector $(x_1^2,\ldots,x_N^2)^T$, and the functions $\FZ{\alpha}\left(\bx,\by\right)$ and $\FF{\alpha}(b;\bx,\by)$ are the generalized hypergeometric functions, which are defined in terms of several quantities.

Consider the integer partitions, $\lambda$, $\tau$ and $\mu$, and the real parameter $\alpha>0$; also, denote by $|\lambda|$ and $l(\lambda)$ the total sum and the length of the partition $\lambda$, respectively. The monomial symmetric functions are given by
\begin{equation}\label{monomialsymmetric}
m_\tau(\bx)=\sum_\sigma \prod_{j=1}^N x_j^{\tau_{\sigma (j)}},
\end{equation}
where the sum is taken over all permutations $\sigma$ such that each monomial $\prod_{j=1}^N x_j^{\tau_{\sigma (j)}}$ is distinct.

The Jack polynomials \cite[p.~379]{macdonald} are denoted by
\begin{equation}
\PP{\lambda}{\alpha}(\bx)=\sum_{\substack{\mu:|\mu|=|\lambda|\\l(\mu)\leq N}}u_{\lambda\mu}(\alpha)m_\mu(\bx),\label{JackP}
\end{equation}
where $u_{\lambda\mu}(\alpha)$ is a triangular matrix, in the sense that $u_{\lambda\mu}(\alpha)$ is nonzero only when $\mu\leq\lambda$ in the sense of the \emph{partial ordering} of partitions defined in \cite[p.~7]{macdonald}. In particular, $u_{\mu\mu}(\alpha)=1$ for all $\alpha>0$. It is known that the Jack functions are part of the eigenfunctions of the periodic type-$A$ Calogero-Moser-Sutherland model \cite{bakerforrester97B}. They are also used to calculate the symmetric eigenfunctions of the type-$A$ Calogero-Moser system. The Jack function of parameter $\alpha>0$, $\PP{\tau}{\alpha}(\bx)$, is defined as the polynomial eigenfunction of the operator \cite{stanley89}
\begin{equation}
\left(\sum_{i=1}^Nx_i^2\frac{\partial^2}{\partial x_i^2}+\frac{2}{\alpha}\sum_{1\leq i\neq j\leq N}\frac{x_i^2}{x_i-x_j}\frac{\partial}{\partial x_i}\right)\PP{\tau}{\alpha}(\bx)=E_{\tau,\alpha}\PP{\tau}{\alpha}(\bx)
\end{equation}
with eigenvalue 
\begin{equation}
E_{\tau,\alpha}=\sum_{j=1}^N\tau_j[\tau_j-1-2(j-1)/\alpha]+|\tau|(N-1).
\end{equation}

For an integer $0\leq n\leq N$, the elementary symmetric polynomial $e_n(\bx)$ is given by
\begin{equation}\label{edefinition}
e_n(\bx)=\sum_{1\leq i_1<\cdots< i_n\leq N}\prod_{j=1}^nx_{i_j}.
\end{equation}
For example, $e_0(\bx)=1$, $e_1(\bx)=\sum_{i=1}^Nx_i$, $e_2(\bx)=\sum_{1\leq i<j \leq N}x_ix_j$ and $e_N(\bx)=\prod_{i=1}^N x_i$. When the subscript of $e$ is a partition, it is given by
\begin{equation}
e_{\tau}(\bx)=\prod_{i=1}^{l(\tau)}e_{\tau_i}(\bx).
\end{equation}

The Schur polynomial $s_{\tau}(\bx)$ is given by the Jacobi-Trudi formula \cite{fulton, macdonald}
\begin{equation}
s_{\tau}(\bx)=\frac{\det_{1\leq i,j\leq N}[x_j^{\tau_i+N-i}]}{\det_{1\leq i,j\leq N}[x_j^{N-i}]}.
\end{equation}

Depending on the value of the parameter $\alpha,$ the Jack polynomials become the Schur, monomial and elementary symmetric polynomials as summarized in Table~\ref{TableSymmetricPolynomials}.
\begin{table}[!ht]
\begin{center}
\begin{tabular}{llll}
\hline
Parameter	&Function				&Matrix					&Function name\\
\hline
$\alpha$		&$\PP{\tau}{\alpha}(\bx)$	&$u_{\tau\lambda}(\alpha)$	&Jack\\
0			&$e_{\tau^\prime}(\bx)$	&$a_{\tau\lambda}$			&Elementary symmetric\\
1			&$s_{\tau}(\bx)$		&$K_{\tau\lambda}$			&Schur\\
$\infty$		&$m_\tau(\bx)$			&$\delta_{\tau\lambda}$		&Monomial symmetric\\
\hline
\end{tabular}
\end{center}
\caption{Some special cases of the Jack polynomials. The matrices in the third column relate the polynomials of the second column with the monomial symmetric polynomials as in Eq.~\eqref{JackP}.}
\label{TableSymmetricPolynomials}
\end{table}

If $\tau$ is an integer partition, then the expression $(i,j)\in\tau$ implies that $1\leq i\leq l(\tau)$ and $1\leq j \leq \tau_i$. The functions $c_\tau (\alpha), c_\tau^\prime (\alpha)$ are given by
\begin{IEEEeqnarray}{rCl}
c_\tau(\alpha)=\prod_{(i,j)\in \tau}(\alpha(\tau_i-j)+\tau_j^\prime-i+1),\label{hooksnocorner}\\
c_\tau^\prime(\alpha)=\prod_{(i,j)\in \tau}(\alpha(\tau_i-j+1)+\tau_j^\prime-i),\label{hookscorner}
\end{IEEEeqnarray}
and the generalized Pochhammer symbol $(b)_\tau^{(\alpha)}$ (with $b>0$) is defined by
\begin{equation}\label{GPS}
(b)_\tau^{(\alpha)}=\prod_{i=1}^{l(\tau)}\frac{\Gamma(b-(i-1)/\alpha+\tau_i)}{\Gamma(b-(i-1)/\alpha)}.
\end{equation}
With this, one can write the generalized hypergeometric functions as follows.
\begin{IEEEeqnarray}{rCl}
\FF{\alpha}(b;\bx,\by)&=&\sum_{n=0}^\infty\sum_{\substack{\tau:l(\tau)\leq N\cr |\tau|=n}}\frac{c_\tau (\alpha)}{c_\tau^\prime (\alpha)}\frac{\PP{\tau}{\alpha}(\bx)\PP{\tau}{\alpha}(\by)}{(b)_\tau^{(\alpha)}(N/\alpha)_\tau^{(\alpha)}}\label{genhypergeob},\\
\FZ{\alpha}(\bx,\by)&=&\sum_{n=0}^\infty\sum_{\substack{\tau:l(\tau)\leq N\cr |\tau|=n}}\frac{c_\tau (\alpha)}{c_\tau^\prime (\alpha)}\frac{\PP{\tau}{\alpha}(\bx)\PP{\tau}{\alpha}(\by)}{(N/\alpha)_\tau^{(\alpha)}}.\label{genhypergeoa}
\end{IEEEeqnarray} 

The form of the intertwining operator is given in terms of the quantities defined above as follows:

\begin{proposition}\label{PropositionVBetaOnSymmetricPolynomials}
The intertwining operators $V_A$ and $V_B$ have the following explicit forms when they act on symmetric polynomials:
\begin{IEEEeqnarray}{rCl}
\label{VTypeASymmetric}V_A m_\lambda(\bx)=\lambda!M(\lambda,N)\sum_{\substack{\tau:l(\tau)\leq N\cr |\tau|= |\lambda|}}\frac{c_\tau (2/\beta)}{c_\tau^\prime (2/\beta)}\frac{u_{\tau\lambda}(2/\beta)}{(\beta N/2)_\tau^{(2/\beta)}}\PP{\tau}{2/\beta}(\bx),
\end{IEEEeqnarray}
\begin{IEEEeqnarray}{rCl}
&&V_B m_\lambda[(\bx)^2]=\frac{(2\lambda)!M(\lambda,N)}{2^{2|\lambda|}}\!\!\!\!\sum_{\substack{\tau:l(\tau)\leq N\cr |\tau|= |\lambda|}}\frac{c_\tau (2/\beta)}{c_\tau^\prime (2/\beta)}\nonumber\\
&&\quad\times\frac{u_{\tau\lambda}(2/\beta)\PP{\tau}{2/\beta}[(\bx)^2]}{(\beta N/2)_\tau^{(2/\beta)}(\beta[\nu+N-1/2]/2+1/2)_\tau^{(2/\beta)}}.\label{VTypeBSymmetric}
\end{IEEEeqnarray}
\end{proposition}

\begin{proof}
Expanding both sides of Eqs.~\eqref{RadialDunklKernelA} and \eqref{RadialDunklKernelB} in terms of symmetric polynomials, and using the orthogonality relations obeyed by the Jack polynomials \cite[p.~379]{macdonald} one may extract an explicit form for both $V_A$ and $V_B$. 

Consider $V_A$ first. Using the notation $\mu!=\prod_{j=1}^{l(\mu)}\mu_j!$, the expansion of the l.h.s.\ of Eq.~\eqref{RadialDunklKernelA} reads
\begin{IEEEeqnarray}{rCl}
\sum_{\rho\in S_N}\rme^{\rho\bx\cdot\by}&=&\sum_{\rho\in S_N}\sum_{n=0}^\infty\sum_{\substack{\mu:l(\mu)\leq N\cr |\mu|=n}}\frac{1}{\mu!}\sum_{\substack{\tau\in S_N:\cr \tau(\mu)\ \textrm{distinct}}}\prod_{j=1}^N(x_{\rho(j)}y_j)^{\mu_{\tau(j)}}\nonumber\\
&=&\sum_{\mu:l(\mu)\leq N}\frac{1}{\mu!}\sum_{\substack{\tau\in S_N:\cr \tau(\mu)\ \textrm{distinct}}}\left\{\sum_{\rho\in S_N}\prod_{j=1}^Nx_{\rho(j)}^{\mu_{\tau(j)}}\right\}\prod_{j=1}^Ny_j^{\mu_{\tau(j)}}\nonumber\\
&=&\sum_{\mu:l(\mu)\leq N}\frac{1}{\mu!}\left\{\sum_{\rho^\prime\in S_N}\prod_{j^\prime=1}^Nx_{j^\prime}^{\mu_{\rho^\prime(j^\prime)}}\right\}\sum_{\substack{\tau\in S_N:\cr \tau(\mu)\ \textrm{distinct}}}\prod_{j=1}^Ny_j^{\mu_{\tau(j)}}.
\end{IEEEeqnarray}
The last line results from the substitutions $j^\prime=\rho(j)$ and $\rho^\prime(j^\prime)=\tau[\rho^{-1}(j^\prime)]$. The last term on the right is, by definition, $m_\mu(\by)$. The term inside the braces is equal to $m_\mu(\bx)$ multiplied by the number of non-distinct permutations of $\mu$. 

Let $l_j^{\mu}$ represent the multiplicity of the $j$th (distinct) part of $\mu$, where the subscript $P$ in $l_P^\mu$ refers to the number of distinct parts of $\mu$. In the cases where $l(\mu)<N$, there are $N-l(\mu)$ zero parts in $\mu$ and therefore $l_P^\mu=N-l(\mu)$ accounts for the multiplicity of zero parts in the first $N$ parts of $\mu$. For example, if $N=6$ and $\mu=(5,3,2,2)=(5,3,2,2,0,0)$, then $P=4$ and $l_1^{\mu}=1$, $l_2^{\mu}=1$, $l_3^{\mu}=2$ and $l_4^{\mu}=2$. Using this notation one may define the following multinomial coefficient:
\begin{equation}
M(\mu,N)=\frac{N!}{l_1^{\mu}!\cdots l_P^{\mu}!}.\label{functionM}
\end{equation}
This function represents the number of distinct permutations of $\mu$ when it is considered as an $N$-dimensional vector. 

Using \eqref{functionM}, the l.h.s.\ of Eq.~\eqref{RadialDunklKernelA} becomes
\begin{equation}
\sum_{\rho\in S_N}V_A\rme^{\rho\bx\cdot\by}=\sum_{\mu:l(\mu)\leq N}\frac{N!m_\mu(\by)}{\mu!M(\mu,N)}V_Am_\mu(\bx).\label{SymDKExp}
\end{equation}

The next step is to eliminate the variable $\by$ using the orthogonality of Jack polynomials. Insertion of the inverse of \eqref{JackP} in \eqref{SymDKExp} after applying $V_A$ yields
\begin{equation}
\sum_{\rho\in S_N}V_A\rme^{\rho\bx\cdot\by}=\sum_{\mu:l(\mu)\leq N}\frac{N![V_{A}m_\mu(\bx)]}{\mu!M(\mu,N)}\sum_{\substack{\nu:\nu\leq\mu\cr |\nu|= |\mu|}}(u^{-1})_{\mu\nu}(2/\beta)\PP{\nu}{2/\beta}(\by).\label{SymmDunklKP}
\end{equation}
Equating Eq.~\eqref{SymmDunklKP} and the r.h.s.\ of Eq.~\eqref{RadialDunklKernelA} gives
\begin{IEEEeqnarray}{rCl}
\sum_{\mu:l(\mu)\leq N}\frac{[V_{A}m_\mu(\bx)]}{\mu!M(\mu,N)}\sum_{\substack{\nu:\nu\leq\mu\cr |\nu|= |\mu|}}(u^{-1})_{\mu\nu}(2/\beta)\PP{\nu}{2/\beta}(\by)&&\nonumber\\
\qquad\qquad\qquad\qquad\qquad 
=\sum_{\tau:l(\tau)\leq N}\frac{c_\tau (2/\beta)}{c_\tau^\prime (2/\beta)}\frac{\PP{\tau}{2/\beta}(\bx)\PP{\tau}{2/\beta}(\by)}{(\beta N/2)_\tau^{(2/\beta)}}.&&
\end{IEEEeqnarray}
From the orthogonality of Jack functions \cite{macdonald} and the linearity of $V_A$, which acts only on $\bx$, we can equate the coefficients of the same Jack functions of $\by$ to obtain
\begin{equation}
\sum_{\substack{\mu:l(\mu)\leq N\cr |\mu|=|\tau|}}\frac{(u^{-1})_{\mu\tau}(2/\beta)}{\mu!M(\mu,N)}V_{A}m_\mu(\bx)=\frac{c_\tau (2/\beta)}{c_\tau^\prime (2/\beta)}\frac{\PP{\tau}{2/\beta}(\bx)}{(\beta N/2)_\tau^{(2/\beta)}}.\label{tobeinverted}
\end{equation}
This relation is solved for $V_Am_\lambda(\bx)$ if we apply the sum $\sum_\tau u_{\tau\lambda}(2/\beta)$ on both sides. The result is Eq.~\eqref{VTypeASymmetric}.

Similarly, an expression for $V_B$ may be extracted from Eq.~\eqref{RadialDunklKernelB}. The first step is to expand $\sum_{\rho\in W_B}\exp(\bx\cdot\rho\by)$ in terms of symmetric polynomials. All the elements of $W_B$ can be written as compositions of variable permutations and sign changes. Therefore,
\begin{equation}
\sum_{\rho\in W_B}\rme^{\bx\cdot\rho\by}=\sum_{\rho\in S_N}\sum_{\substack{\mu:l(\mu)\leq N}}\frac{1}{\mu!}\sum_{\substack{\tau\in S_N:\cr \tau(\mu)\ \textrm{distinct}}}\prod_{j=1}^N\sum_{s_j=\pm 1}s_j^{\mu_{\tau(j)}}(y_{\rho(j)} x_j)^{\mu_{\tau(j)}}.
\end{equation}
The product over $j$ vanishes when at least one of the parts of $\mu$ is odd, so it suffices to consider partitions with even parts. Then,
\begin{IEEEeqnarray}{rCl}
\sum_{\rho\in W_B}\rme^{\bx\cdot\rho\by}&=&\sum_{\substack{\mu:l(\mu)\leq N}}\frac{2^N}{(2\mu)!}\sum_{\substack{\tau\in S_N:\cr \tau(\mu)\ \textrm{distinct}}}\left\{\sum_{\rho\in S_N}\prod_{j=1}^N(y_{\rho(j)})^{2\mu_{\tau(j)}}\right\} \prod_{j=1}^Nx_j^{2\mu_{\tau(j)}}\nonumber\\
&=&\sum_{\substack{\mu:l(\mu)\leq N}}\frac{2^NN!}{(2\mu)!}\frac{m_{\mu}[(\bx)^2]m_{\mu}[(\by)^2]}{M(\mu,N)}.\label{exponentialexpansion}
\end{IEEEeqnarray}

Applying $V_B$ on this result and inserting into \eqref{RadialDunklKernelB} yields
\begin{equation}
V_B\sum_{\substack{\mu:l(\mu)\leq N}}\frac{m_{\mu}[(\bx)^2]m_{\mu}[(\by)^2]}{(2\mu)!M(\mu,N)}=\FF{2/\beta}\left(\frac{\beta}{2}(\nu+N-1/2) +\frac{1}{2};\frac{(\bx)^2}{2},\frac{(\by)^2}{2}\right).
\end{equation}
From \eqref{genhypergeob} and the fact that the Jack polynomial $\PP{\tau}{\alpha}(\bx)$ is homogeneous of degree $|\tau|$, one obtains
\begin{IEEEeqnarray}{rCl}
&&V_B\!\!\!\!\sum_{\substack{\mu:l(\mu)\leq N}}\!\!\!\frac{m_{\mu}[(\bx)^2]m_{\mu}[(\by)^2]}{(2\mu)!M(\mu,N)}=\!\!\!\!\sum_{\tau:l(\tau)\leq N}\frac{c_\tau (2/\beta)}{c_\tau^\prime (2/\beta)}\nonumber\\
&&\quad\times\frac{\PP{\tau}{2/\beta}[(\bx)^2]\PP{\tau}{2/\beta}[(\by)^2]}{2^{2|\tau|}(\beta[\nu+N-1/2]/2+1/2)_\tau^{(2/\beta)}(N\beta/2)_\tau^{(2/\beta)}}.
\end{IEEEeqnarray}

Next, using the inverse of the expansion of Jack polynomials into monomial symmetric polynomials, Eq.~\eqref{JackP}, on the l.h.s.\ and equating the coefficients of $\PP{\tau}{2/\beta}[(\by)^2]$ gives
\begin{IEEEeqnarray}{rCl}
&&\sum_{\substack{\mu:l(\mu)\leq N\cr |\mu|=|\tau|}}\!\!\frac{V_Bm_{\mu}[(\bx)^2]}{(2\mu)!M(\mu,N)}(u^{-1})_{\mu\tau}(2/\beta)\nonumber\\
&&\quad=\frac{c_\tau (2/\beta)}{c_\tau^\prime (2/\beta)}\frac{\PP{\tau}{2/\beta}[(\bx)^2/4]}{(\frac{\beta}{2}[\nu+N-1/2]+\frac{1}{2})_\tau^{(2/\beta)}(N\frac{\beta}{2})_\tau^{(2/\beta)}}.
\end{IEEEeqnarray}
Multiplying by $\sum_\tau u_{\tau\lambda}(2/\beta)$ on both sides yields, finally, Eq.~\eqref{VTypeBSymmetric}.
\end{proof}

Equations~\eqref{VTypeASymmetric} and \eqref{VTypeBSymmetric} can be used to obtain a clearer idea of the action of the intertwining operators on symmetric polynomials. In the simplest non-trivial case, one can consider the case where the partition $\lambda$ is equal to $(2,0,\ldots)$ and study the effect of $V_A$ on the equation $m_{2}(\bx)= \sum_{j=1}^N x_j^2=1$ and the effect of $V_B$ on the equation $m_{2}[(\bx)^2]= \sum_{j=1}^N x_j^4=1$. In Chap.~\ref{general_freezing}, the variables considered were often scaled up by a factor of $\sqrt{\beta}$. The same scaling will be taken into account here. The transformed equations are
\begin{IEEEeqnarray}{rCl}
&&1=V_A \beta m_{2}(\bx)=\frac{\beta(\beta+2)}{\beta N+2}\sum_{j=1}^N x_j^2+\frac{2\beta^2}{\beta N+2}\sum_{1\leq i<j\leq N}x_ix_j,\label{ExampleA}\\
&&1=V_B \beta^2 m_{2}[(\bx)^2]=\frac{3\beta^2[(\beta+2)\sum_{i=1}^Nx_i^4+2\beta\sum_{1\leq i<j\leq N}x_i^2 x_j^2]}{(\beta[\nu+N-1/2]+1)(\beta[\nu+N-1/2]+3)(\beta N+2)}.\nonumber\\\label{ExampleB}
\end{IEEEeqnarray}

The change in the surfaces defined by the l.h.s.\ of Eqs.~\eqref{ExampleA} and \eqref{ExampleB} is depicted in Fig.~\ref{FigureQuadraticExamples}.
\begin{figure}[!t]
  \centering
  \subfloat[$m_{2}(\bx)=1$]{\includegraphics[width=0.3\textwidth]{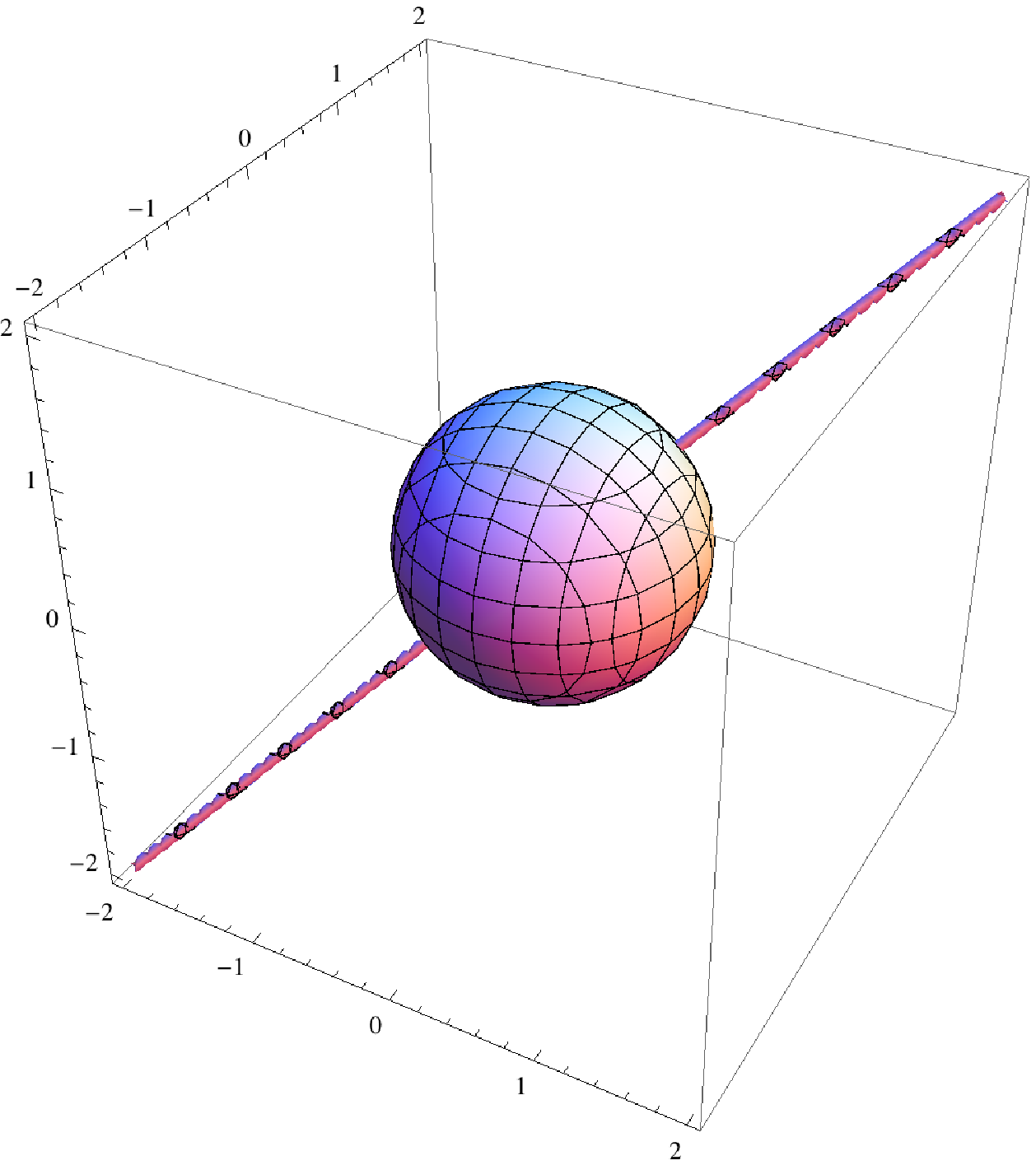}}\
  \subfloat[Type $A_2$, $\beta=2$]{\includegraphics[width=0.3\textwidth]{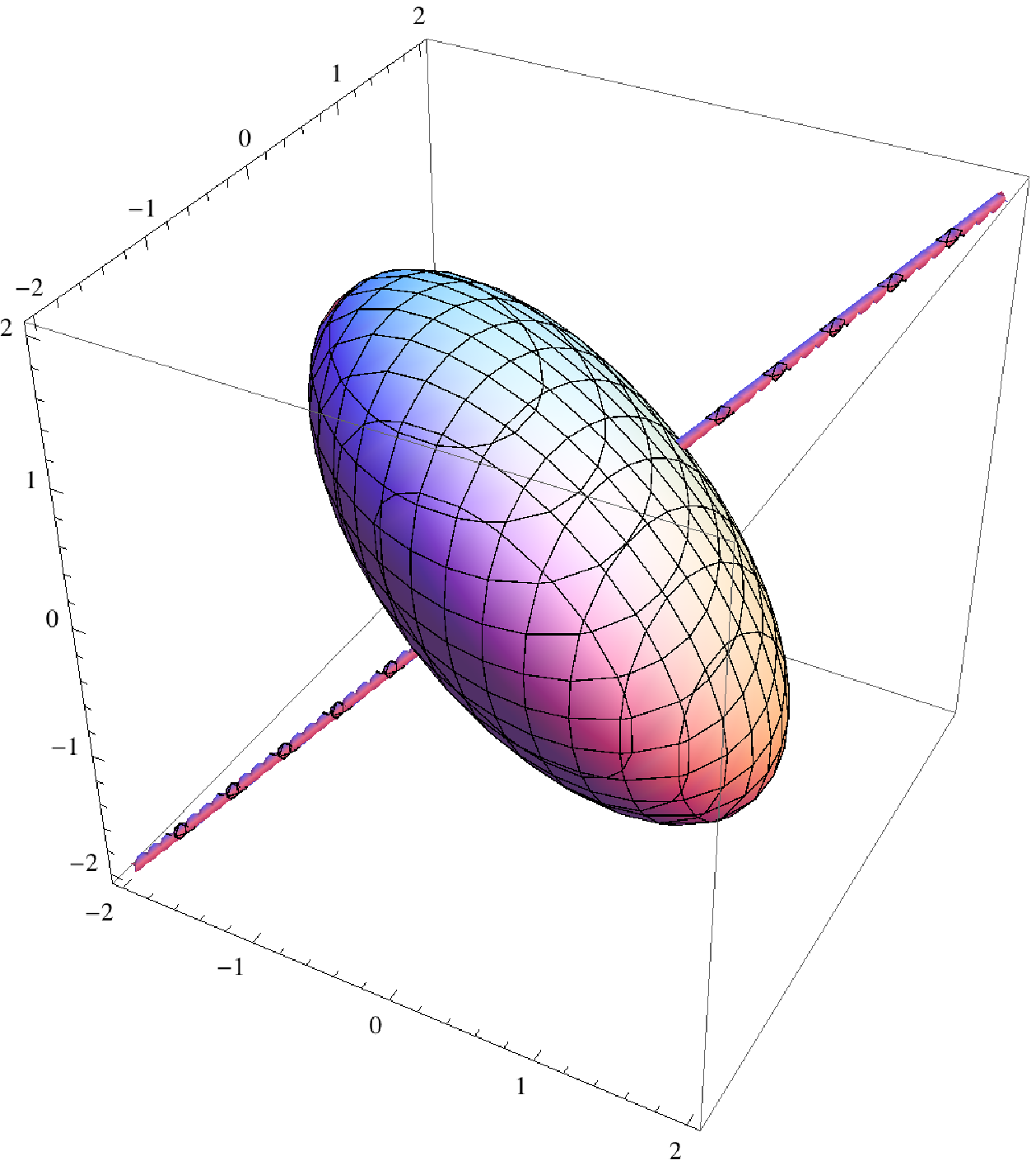}}\
  \subfloat[Type $A_2$, $\beta\to\infty$]{\includegraphics[width=0.3\textwidth]{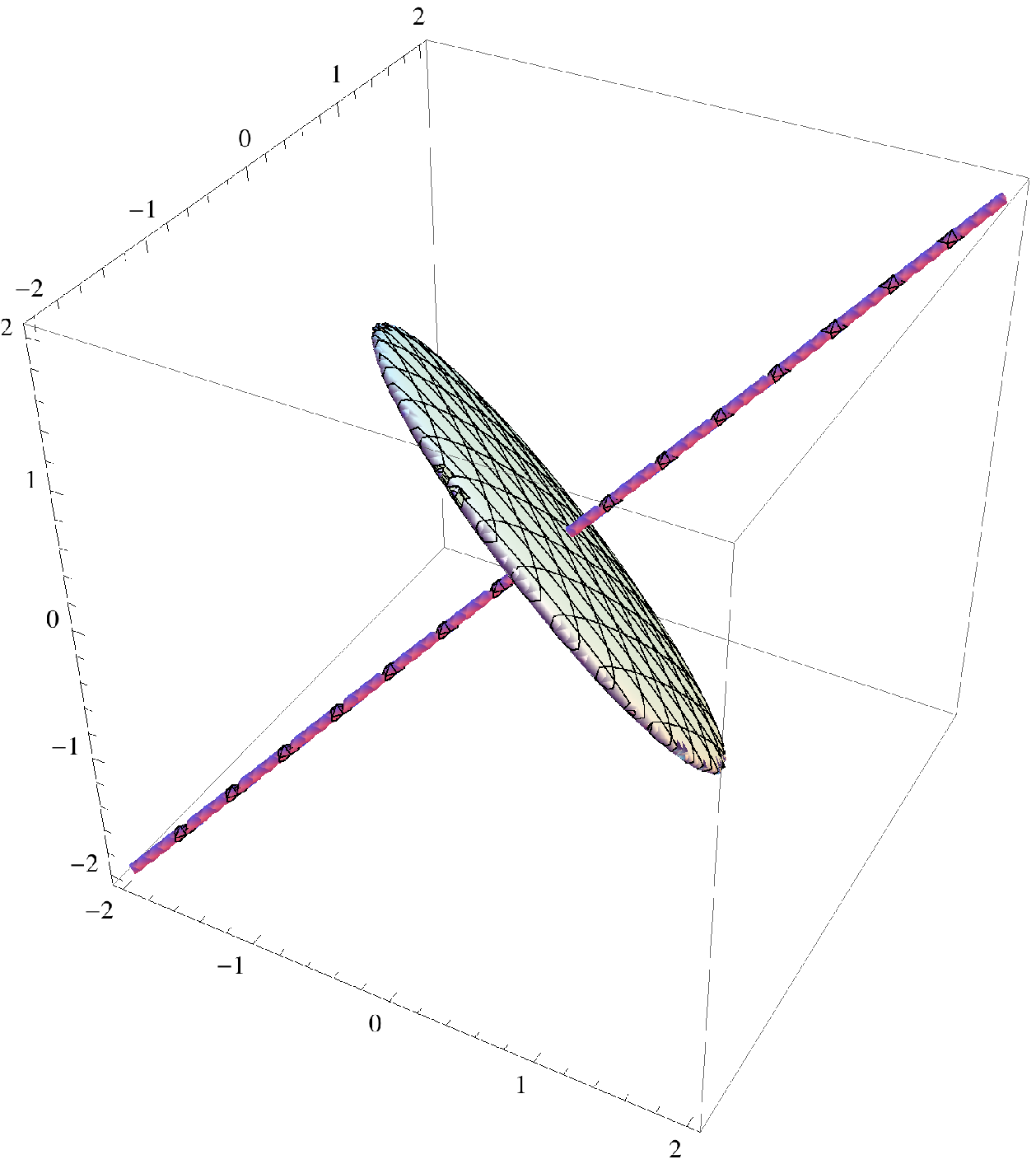}}\\

  \subfloat[$m_2{[(\bx)^2]}=1$]{\includegraphics[width=0.3\textwidth]{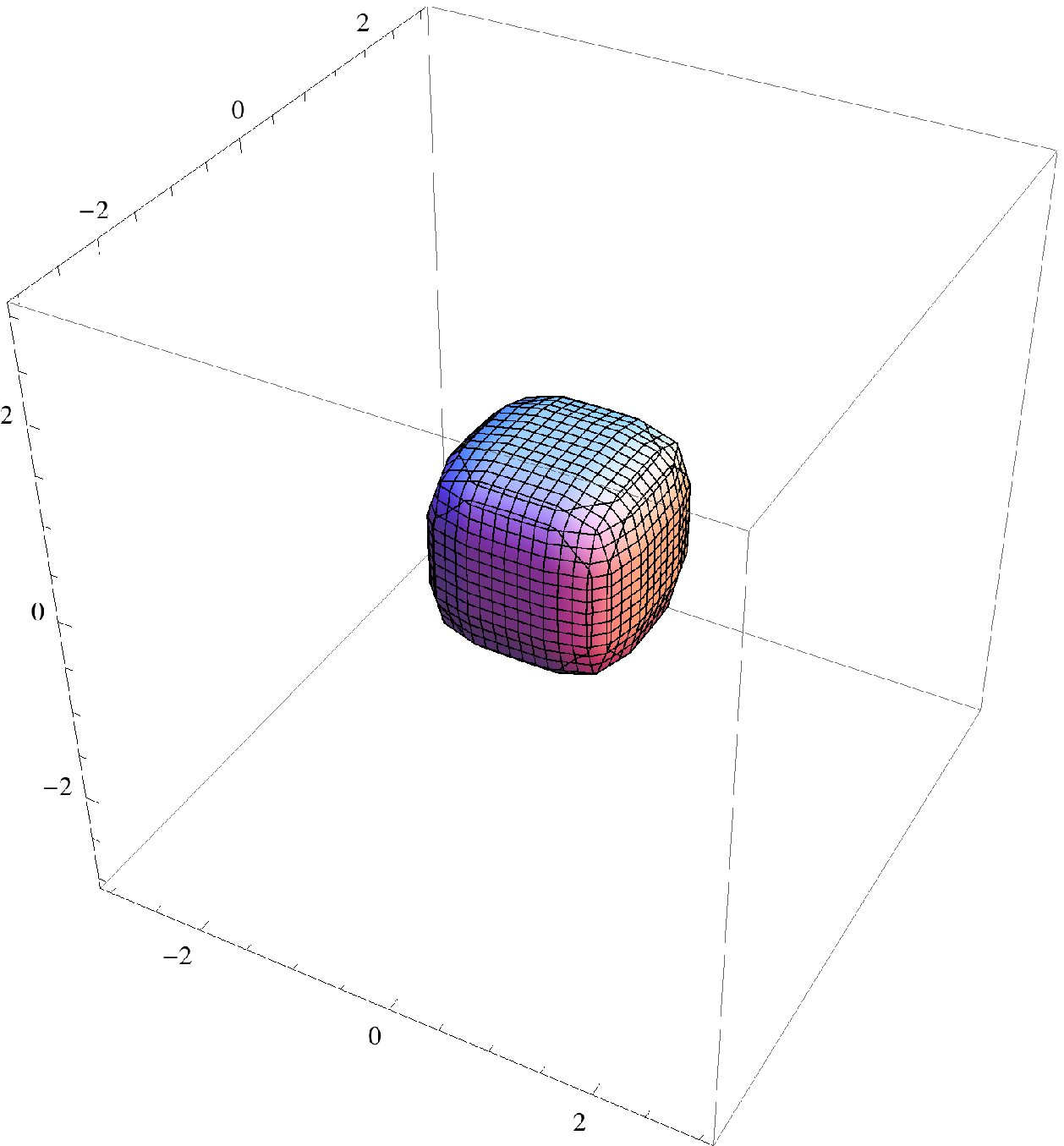}}\                 
  \subfloat[Type $B_3$, $\beta=2$]{\includegraphics[width=0.3\textwidth]{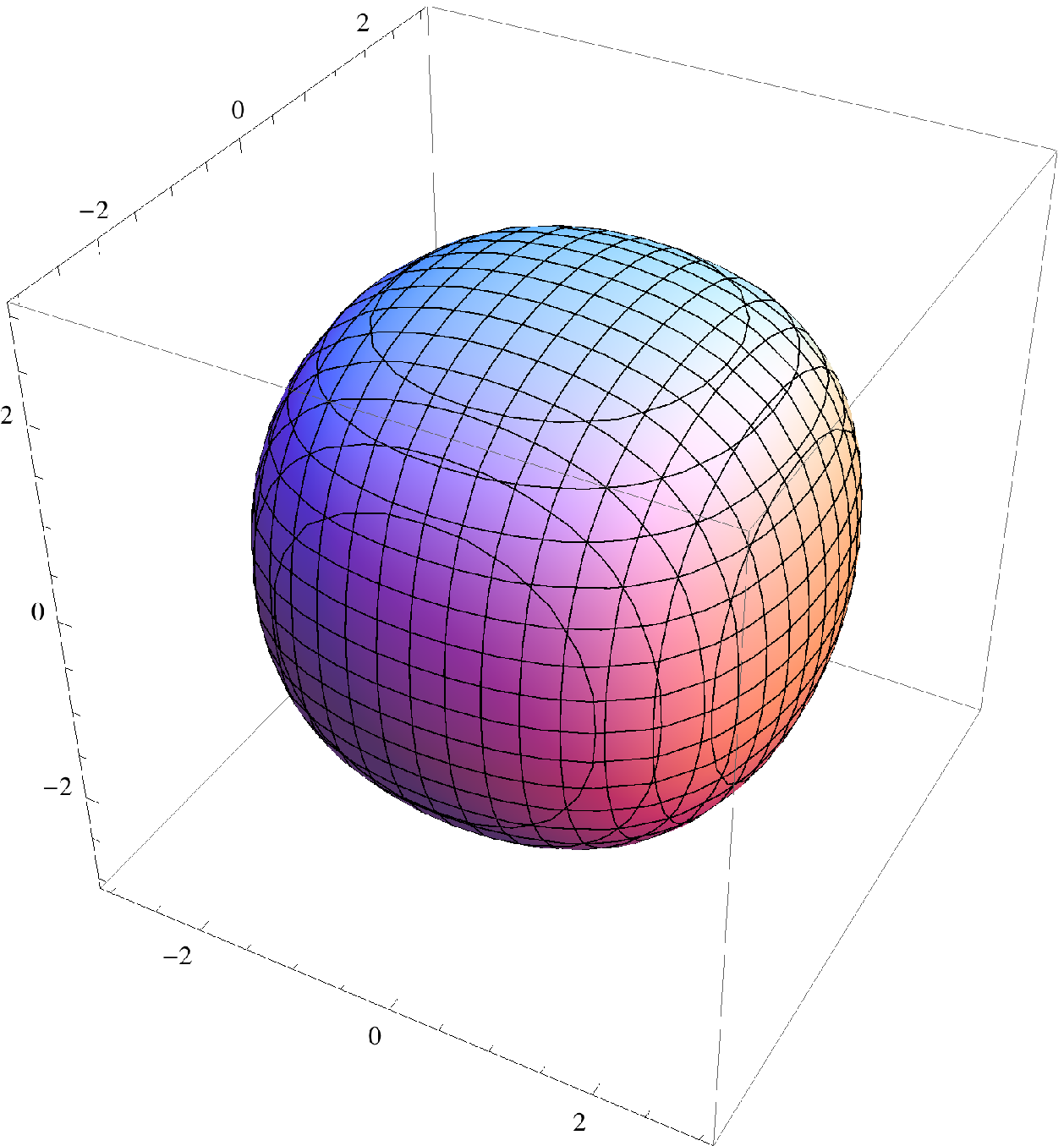}}\
  \subfloat[Type $B_3$, $\beta\to\infty$]{\includegraphics[width=0.3\textwidth]{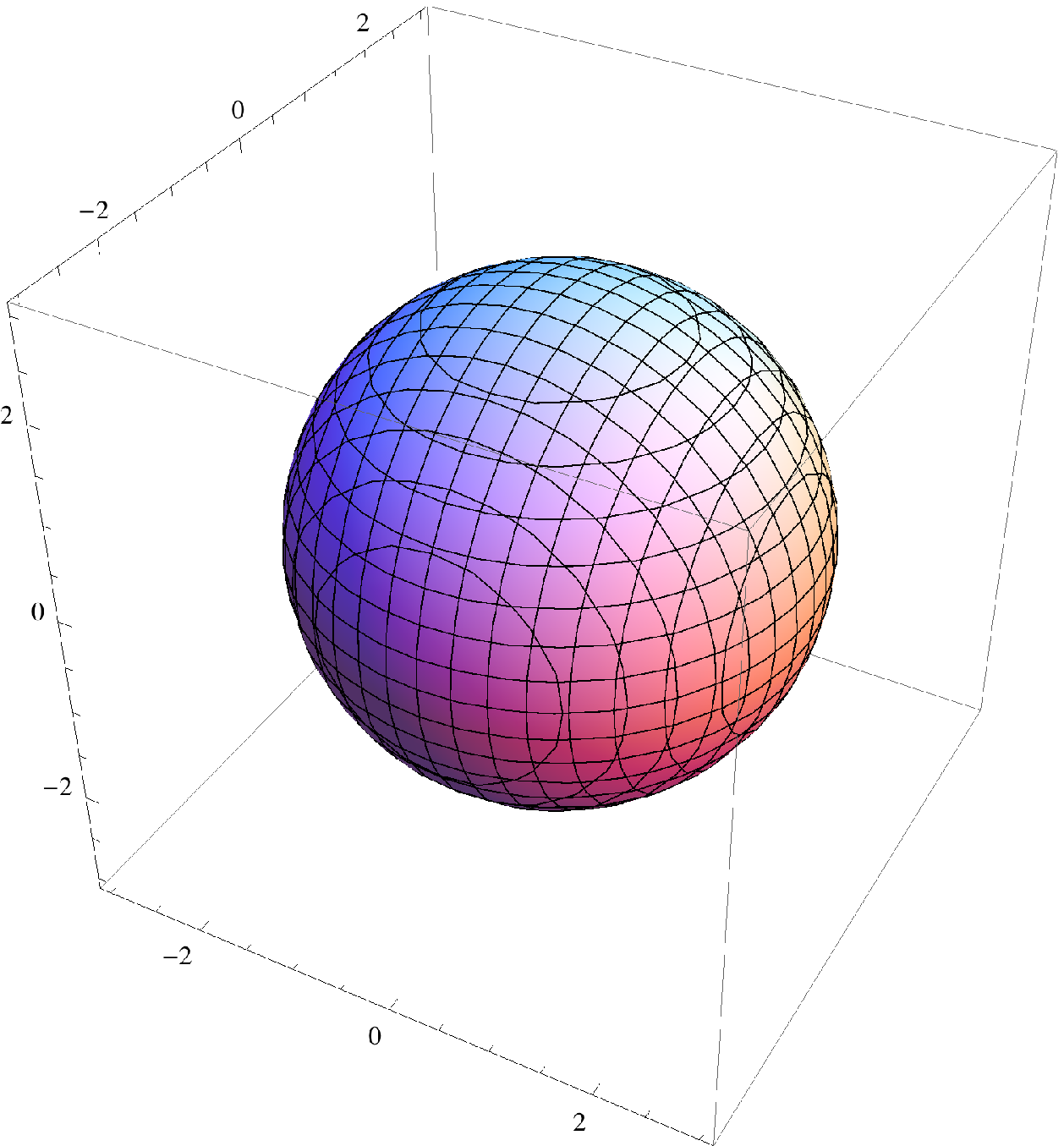}}\\
  
  \caption{First row: effect of the intertwining operator of type $A_2$ on the equation $\beta m_2(\bx)=\beta \sum_{i=1}^3x_i^2=1$. Second row: effect of the intertwining operator of type $B_3$ on the equation $\beta^2m_2[(\bx)^2]=\beta^2\sum_{i=1}^3x_i^4=1$ with $\nu=1/2$.}
  \label{FigureQuadraticExamples}
\end{figure}
From the figure, two observations can be made. 

The first is that the intertwining operators stretch the surfaces in the figure in the directions covered by the span of the related Weyl group. Clearly, $W_B$ is of full rank, and consequently, $V_B$ stretches the cube-like surface defined by the equation $\sum_{i=1}^3x_i^4=1$ into a sphere, as shown in Fig.~\ref{FigureQuadraticExamples}(f). At the same time, because $W_A=S_N$ is not of full rank, it only stretches the sphere in Fig.~\ref{FigureQuadraticExamples}(a) in the plane orthogonal to the diagonal line in Figs.~\ref{FigureQuadraticExamples}(a), (b) and (c). This line represents all the points $(a,a,a),\ a\in\RR$, which are orthogonal to the span of $S_3$. 

The second observation is that this stretching action becomes symmetric as $\beta$ grows to infinity. In Figs.~\ref{FigureQuadraticExamples}(c) and (f), the resulting surfaces are a circle and a sphere, respectively, meaning that the space within the span of the root system is mapped to a $d_R$-dimensional ball, while the space that is orthogonal to the root system remains unchanged. These observations are true in general, and they will have an effect on the form of the generalized Bessel functions as $\beta\to\infty$.

\subsection{Intertwining operators in the limiting regimes}

The intertwining operators $V_A$ and $V_B$ show the following limiting behavior.
\begin{proposition}\label{PropositionLimitsVSymmetric}
As $\beta\to\infty$, $V_A$ and $V_B$ become
\begin{IEEEeqnarray}{rCl}
\lim_{\beta\to\infty}V_A m_\lambda(\bx)&=&\frac{M(\lambda,N)}{N^{|\lambda|}}\left(\sum_{j=1}^Nx_j\right)^{|\lambda|},\label{EquationVTypeALimit}\\
\lim_{\beta\to\infty}V_B\beta^{|\lambda|}m_\lambda[(\bx)^2]&=&\frac{(2\lambda)!M(\lambda,N)}{2^{|\lambda|}\lambda!N^{|\lambda|}(\nu+N-1/2)^{|\lambda|}}\left(\sum_{j=1}^Nx_j^2\right)^{|\lambda|}\!\!\!\!\nonumber\\
&=&\frac{(2\lambda)!}{\lambda!}\lim_{\beta\to\infty}V_Am_\lambda(\bu),\label{EquationVTypeBLimitBeta}
\end{IEEEeqnarray}
with $\bu=(\bx)^2/[2(\nu+N-1/2)]$ on the r.h.s. Furthermore, as $\nu\to\infty$, 
\begin{equation}\label{EquationVTypeBLimitNu}
\lim_{\nu\to\infty}V_B \nu^{|\lambda|}m_\lambda[(\bx)^2]=\frac{(2\lambda)!}{\lambda!}V_A m_\lambda(\bu),
\end{equation}
with $\bu=(\bx)^2/(2\beta)$.
\end{proposition}

\begin{proof}
The most important part of the proof consists of evaluating the product
\begin{multline}\label{EquationFilterProduct}
\frac{c_\tau(2/\beta)}{c_\tau^\prime (2/\beta)(\beta N/2)_\tau^{(2/\beta)}}=\\ \prod_{(i,j)\in \tau}\frac{(\tau_i-j+\beta(\tau_j^\prime-i+1)/2)}{(\beta(N-i+1)/2+j-1)(\tau_i-j+1+\beta(\tau_j^\prime-i)/2)}
\end{multline}
in the freezing limit $\beta\to\infty$. In general, this quantity tends to zero, and the only case in which it does not vanish is when $\tau_j^\prime=i$ for any $(i,j)\in\tau$. Now, if one considers the representation of the integer partition $\tau$ as a Young diagram (see, e.g., \cite{fulton}), $\tau^\prime_j$ indicates the number of boxes on the $j$th column of the diagram. This is equivalent to counting the number of rows in the diagram that have at least $j$ boxes. It follows that $\tau^\prime_j$ is the number of parts greater than or equal to $j$ in $\tau$. Thus, one can conclude that the condition $\tau_j^\prime=i\ \forall(i,j)\in\tau$ is satisfied only when $\tau_j^\prime=1$, i.e. $\tau=(|\tau|,0,\ldots)$. In other words, only partitions of length equal to one satisfy this condition. Therefore,
\begin{equation}\label{EquationFilterProduct0}
\frac{c_\tau(2/\beta)}{c_\tau^\prime (2/\beta)(\beta N/2)_\tau^{(2/\beta)}}\stackrel{\beta\to\infty}{\longrightarrow}0
\end{equation}
whenever $l(\tau)>1$ and
\begin{IEEEeqnarray}{rCl}
\frac{c_\tau(2/\beta)}{c_\tau^\prime (2/\beta)(\beta N/2)_\tau^{(2/\beta)}}&=&\prod_{j=1}^{|\tau|}\frac{(|\tau|-j+\beta/2)}{(\beta N/2+j-1)(|\tau|-j+1)}\nonumber\\
&&\stackrel{\beta\to\infty}{\longrightarrow}\prod_{j=1}^{|\tau|}\frac{1}{N(|\tau|-j+1)}=\frac{1}{N^{|\tau|}|\tau|!}\label{EquationFilterProduct1}
\end{IEEEeqnarray}
when $l(\tau)=1$. 

Consider first Eq.~\eqref{VTypeASymmetric}. As $\beta\to\infty$, $V_A$ becomes
\begin{equation}
V_A m_\lambda(\bx)\stackrel{\beta\to\infty}{\longrightarrow}\frac{\lambda!M(\lambda,N)}{N^{|\lambda|}|\lambda|!}a_{\lambda^*\lambda}(e_{1}(\bx))^{|\lambda|},\nonumber
\end{equation}
where $\lambda^*=(|\lambda|,0,\ldots)$. 
This is due to the fact that if $l(\tau)=1$, then $\tau^\prime=(1,1,\ldots,1)$, a partition composed of ones with $|\lambda|$ parts. The calculation for $V_A$ is completed by giving an explicit form of the matrix components $a_{\lambda^*\lambda}$ by expanding $(e_{1}(\bx))^{|\lambda|}$ in terms of $m_\tau(\bx)$:
\begin{equation}
(e_{1}(\bx))^{|\lambda|}=\left(\sum_{j=1}^Nx_j\right)^{|\lambda|}=\sum_{\substack{\tau:l(\tau)\leq N\cr |\tau|=|\lambda|}}\frac{|\lambda|!}{\tau!}m_\tau(\bx)=\sum_{\substack{\tau:l(\tau)\leq N\cr |\tau|=|\lambda|}}a_{\lambda^*\tau}m_\tau(\bx).
\end{equation}
Therefore,
\begin{equation}
a_{\lambda^*\lambda}=\frac{|\lambda|!}{\lambda!},\label{EquationLimitMatrix}
\end{equation}
which yields Eq.~\eqref{EquationVTypeALimit}.

Consider now Eq.~\eqref{VTypeBSymmetric} with the variables $\bx$ scaled up by a factor of $\sqrt{\beta}$. In view of Eqs.~\eqref{EquationFilterProduct0} and \eqref{EquationFilterProduct1}, it suffices to consider the limit 
\begin{IEEEeqnarray}{rCl}
\lim_{\beta\to\infty}\frac{\beta^{|\tau|}}{(\frac{\beta}{2}[\nu+N-1/2]+\frac{1}{2})_\tau^{(2/\beta)}}&=&\lim_{\beta\to\infty}\prod_{(i,j)\in\tau}\frac{\beta}{\frac{\beta}{2}(\nu+1/2+N-i)+j-\frac{1}{2}}\nonumber\\
&=&\prod_{(i,j)\in\tau}\frac{2}{\nu+1/2+N-i}.
\end{IEEEeqnarray}
However, because only the term where $l(\tau)=1$ survives as $\beta\to\infty$, the only term required is
\begin{equation}
\lim_{\beta\to\infty}\frac{\beta^{|\tau|}}{(\beta[\nu+N-1/2]/2+\frac{1}{2})_{(|\tau|,0\ldots)}^{(1/k)}}=\frac{2^{|\tau|}}{(\nu+N-1/2)^{|\tau|}}.\label{EquationPartialLimitBetaVB}
\end{equation}
Using Eqs.~\eqref{EquationFilterProduct1}, \eqref{EquationLimitMatrix} and \eqref{EquationPartialLimitBetaVB} on Eq.~\eqref{VTypeBSymmetric} yields Eq.~\eqref{EquationVTypeBLimitBeta}.

To complete the proof, it only remains to compute the limit $\nu\to\infty$ of Eq.~\eqref{VTypeBSymmetric} after scaling the vector $\bx$ by a factor of $\sqrt{\nu}$. The parameter $\nu$ only occurs in the ratio (recall that $|\lambda|=|\tau|$),
\begin{multline}
\lim_{\nu\to\infty}\frac{(\beta[\nu+N-1/2]+1)_\tau^{(2/\beta)}}{(2\nu)^{|\tau|}}=\\
\lim_{\nu\to\infty}\prod_{(i,j)\in\tau}\frac{\beta(\nu+1/2+N-i)+2j-1}{2\nu}=\Big(\frac{\beta}{2}\Big)^{|\tau|}.
\end{multline}
This yields
\begin{multline}
\lim_{\nu\to\infty}V_B \nu^{|\lambda|}m_\lambda[(\bx)^2]=\\(2\lambda)!M(\lambda,N)\sum_{\substack{\tau:l(\tau)\leq N\cr |\tau|= |\lambda|}}\frac{c_\tau (2/\beta)}{c_\tau^\prime (2/\beta)}\frac{u_{\tau\lambda}(2/\beta)}{(\beta N/2)_\tau^{(2/\beta)}}\frac{\PP{\tau}{2/\beta}[(\bx)^2]}{(2\beta)^{|\tau|}},
\end{multline}
as desired.
\end{proof}

\subsection{Limiting regimes of the Generalized Bessel functions}

After calculating the limiting behavior of $V_A$ and $V_B$, the results can be reassembled to obtain the limiting behavior of Eqs.~\eqref{RadialDunklKernelA} and \eqref{RadialDunklKernelB}.

\begin{proposition}\label{FreezingLimitOfGeneralizedBesselFunctions}
The limiting regime of the generalized Bessel function of type $A$ without scaling is given by
\begin{equation}
\lim_{\beta\to\infty}\sum_{\rho\in S_N}V_A\rme^{\rho\bx\cdot\by}=N!\exp\left(\frac{(\bx\cdot\bone)(\by\cdot\bone)}{N}\right),
\end{equation}
and the two scaled limiting regimes of the generalized Bessel function of type $B$ are given by
\begin{IEEEeqnarray}{rCl}
\lim_{\beta\to\infty}\sum_{\rho\in W_{B}}V_B \rme^{\sqrt{\beta}\by\cdot\rho\bx}&=&2^NN!\exp\left(\frac{y^2 x^2}{2N(\nu+N-1/2)}\right),\label{RadialDunklKernelTypeBFreezing}\\
\lim_{\nu\to\infty}\sum_{\rho\in W_{B}}V_B \rme^{\sqrt{\nu}\by\cdot\rho\bx}&=&2^NN!\FZ{2/\beta}\Bigg(\frac{(\bx)^2}{\sqrt{2\beta}},\frac{(\by)^2}{\sqrt{2\beta}}\Bigg).\label{RadialDunklKernelTypeBNuInfinity}
\end{IEEEeqnarray}
\end{proposition}

\begin{proof} The proof of this proposition is fairly straightforward. First, each of the generalized Bessel functions is expanded in terms of symmetric polynomials. Then, the corresponding intertwining operator is applied and the parameter limit is taken using Proposition~\ref{PropositionLimitsVSymmetric}. Finally, the resulting sums are reassembled into exponential functions. 

The freezing limit $\beta\to\infty$ of Eq.~\eqref{RadialDunklKernelA} is as follows.
\begin{IEEEeqnarray}{rCl}
\lim_{\beta\to\infty}\sum_{\rho\in S_N}V_A\rme^{\rho\bx\cdot\by}&=&\sum_{\mu:l(\mu)\leq N}\frac{N!m_\mu(\by)}{\mu!M(\mu,N)}\lim_{\beta\to\infty}V_{A}m_\mu(\bx)\nonumber\\
&=&\sum_{\mu:l(\mu)\leq N}\frac{N!m_\mu(\by)}{\mu!N^{|\mu|}}\left(\sum_{j=1}^Nx_j\right)^{|\mu|}\nonumber\\
&=&\sum_{\mu:l(\mu)\leq N}\frac{N!}{\mu!}m_\mu\left(\frac{(\bx\cdot\bone)\by}{N}\right)\nonumber\\
&=&N!\exp\left(\frac{(\bx\cdot\bone)(\by\cdot\bone)}{N}\right).
\end{IEEEeqnarray}
Here, $\bone$ denotes the vector $(1,\ldots,1)$. 

The freezing limit of Eq.~\eqref{RadialDunklKernelB} is taken similarly,
\begin{IEEEeqnarray}{rCl}
\lim_{\beta\to\infty}\sum_{\rho\in W_{B}}V_B \rme^{\sqrt{\beta}\by\cdot\rho\bx}&=&2^NN!\sum_{\mu:l(\mu)\leq N}\frac{m_\mu[(\by)^2]}{2^{|\mu|}\mu!N^{|\mu|}(\nu+N-1/2)^{|\mu|}}\left(\sum_{j=1}^Nx_j^2\right)^{|\mu|}\nonumber\\
&=&2^NN!\sum_{\mu:l(\mu)\leq N}\frac{1}{\mu!}m_\mu\left[\frac{(\by)^2x^2}{2N(\nu+N-1/2)}\right]\nonumber\\
&=&2^NN!\exp\left(\frac{y^2x^2}{2N(\nu+N-1/2)}\right).
\end{IEEEeqnarray}

The limit $\nu\to\infty$ of Eq.~\eqref{RadialDunklKernelB} is taken below:
\begin{IEEEeqnarray}{rCl}
&&\lim_{\nu\to\infty}\sum_{\rho\in W_{B}}V_B \rme^{\sqrt{\nu}\by\cdot\rho\bx}\nonumber\\
&&\quad=\sum_{\substack{\mu:l(\mu)\leq N}}\frac{2^NN!}{(2\mu)!}\frac{m_{\mu}[(\by)^2]}{M(\mu,N)}\lim_{\nu\to\infty}V_B\nu^{|\mu|}m_{\mu}[(\bx)^2]\nonumber\\
&&\quad=\sum_{\substack{\mu:l(\mu)\leq N}}\!\!\!\!\frac{2^NN!m_{\mu}[(\by)^2]}{2^{2|\mu|}}\!\!\!\!\sum_{\substack{\tau:l(\tau)\leq N\cr |\tau|= |\mu|}}\frac{c_\tau (2/\beta)}{c_\tau^\prime (2/\beta)}\frac{u_{\tau\mu}(2/\beta)}{(\beta N/2)_\tau^{(2/\beta)}}\frac{\PP{\tau}{2/\beta}[(\bx)^2]}{(\beta/2)^{|\tau|}}\nonumber\\
&&\quad=2^NN!\sum_{\substack{\tau:l(\tau)\leq N}}\frac{c_\tau (2/\beta)}{c_\tau^\prime (2/\beta)}\frac{\PP{\tau}{2/\beta}[(\by)^2]\PP{\tau}{2/\beta}[(\bx)^2]}{2^{|\tau|}\beta^{|\tau|}(\beta N/2)_\tau^{(2/\beta)}}.
\end{IEEEeqnarray}
This last calculation completes the proof.
\end{proof}

From the first of the three limiting regimes proved above, it follows that the scaled low-temperature behavior of the generalized Bessel function of type $A$ is given by
\begin{equation}
\sum_{\rho\in S_N}V_A\rme^{\sqrt{\beta}\rho\bx\cdot\by}\stackrel{\beta\text{ large}}{\approx}N!\exp\Big(\sqrt{\beta}\frac{(\bx\cdot\bone)(\by\cdot\bone)}{N}\Big)
\end{equation}
to leading order in $\beta$. The next-order terms could be calculated using the procedure from the proof of Proposition~\ref{PropositionLimitsVSymmetric}. However, because there are a large number of integer partitions that produce terms of order greater than or equal to $O(\beta^0)$ from the product in Eq.~\eqref{EquationFilterProduct}, the resulting sum over partitions becomes difficult to calculate. With this in mind, the correct expression is obtained by using Lemma~\ref{FreezingLimitDunklKernel}, that is,
\begin{equation}\label{CorrectFreezingRegimeDunklKernelA}
\sum_{\rho\in S_N}V_A\rme^{\sqrt{\beta}\rho\bx\cdot\by}\stackrel{\beta\text{ large}}{\approx}N!\exp\Big(\sqrt{\beta}\frac{(\bx\cdot\bone)(\by\cdot\bone)}{N}+\frac{x_\parallel^2y_\parallel^2}{2(\gamma_A+\varepsilon_\beta)}\Big).
\end{equation}

\section{Derivation of propositions~\ref{LimitingRegimesA} and \ref{LimitingRegimesB}}

Here, an alternate derivation is proposed for the steady-state distribution and freezing limit of the interacting Brownian motions and interacting Bessel processes. These derivations use the additional assumption that the parameter $\beta$ is large but finite unless otherwise noted.

\subsection{Steady-state part}

Consider the interacting Brownian motions. Using the expression for the transition density $p_A(t,\by|\bx)$ given in Eq.~\eqref{TransitionDensityRadialDunkl} as well as the first column of Table~\ref{RadialDunklProcessesQuantities},  Eq.~\eqref{CorrectFreezingRegimeDunklKernelA} and Stirling's approximation \cite{feller}, for large $\beta$ one may write the following approximation for the scaled transition probability density,
\begin{multline}
\log [p_A(t,\sqrt{\beta t}\bv|\bx)(\beta t)^{N/2}]\\
\approx-\beta\left[F_A(\bv)+\frac{1}{2}\sum_{i=1}^Ni\log i-\frac{N}{4}(N-1)(1+\log 2)\right]\\
+\sqrt{\frac{\beta}{t}}x_\perp v_\perp+\frac{x_\parallel^2 v_\parallel^2}{2t(\gamma_A +\varepsilon_\beta)}+\frac{N}{2}\log \frac{\beta}{2}+\log N!-\frac{N}{2}\log\pi-\frac{x^2}{2t},\label{LogDensityTypeA}
\end{multline}
where $\gamma_A$ is given in Table~\ref{RadialDunklProcessesQuantities} and $x_\perp=(\bx\cdot\bone)/\sqrt{N}$, $v_\perp=(\bv\cdot\bone)/\sqrt{N}$.

Similarly, the large-$\beta$ behavior of the scaled transition probability density $p_B(t,\sqrt{\beta}\bv|\bx)(\beta t)^{N/2}$ of the interacting Bessel processes can be written as follows,
\begin{multline}
\log [p_B(t,\sqrt{\beta t}\bv|\bx)(\beta t)^{N/2}]\approx -\beta\Bigg[F_B(\bv,\nu)-\frac{N}{2}(N+\nu-1/2)\\+\frac{1}{2}\sum_{i=1}^N i\log i
+\frac{1}{2}\sum_{i=1}^N(\nu+i-1/2) \log (\nu+i-1/2)\Bigg] +\frac{N}{2}\log\beta\\
-\frac{x^2}{2t}+\frac{N}{2}\log 2+\log N!+\frac{v^2x^2}{2t(\gamma_B+\varepsilon_B)},\label{LogDensityTypeB}
\end{multline}
where $\gamma_B$ is given in Table~\ref{RadialDunklProcessesQuantities}.

Consider the initial distributions $\mu_A(\bx)$ and $\mu_B(\bx)$, defined in the Weyl chambers $C_A$ and $C_B$ respectively, which are assumed to have finite second-order moments. Then, the relations
\begin{multline}\label{ScaledDistributionTypeALargeBeta}
f_A(t,\sqrt{\beta t}\bv)(\beta t)^{N/2}\ud\bv\approx\rme^{-\beta [F_A(\bv)-K_A]}\Big(\frac{\beta}{2\pi}\Big)^{N/2}N!\\
\times\int_{C_A}\exp\Big[-\frac{x^2}{2t}+\frac{x_\parallel^2 v^2_\parallel}{2t(\gamma_A +\varepsilon_\beta)}+\sqrt{\frac{\beta}{t}}x_\perp v_\perp\Big]\mu_A(\bx)\ud\bx\ud\bv
\end{multline}
and
\begin{multline}\label{ScaledDistributionTypeBLargeBeta}
f_B(t,\sqrt{\beta t}\bv)(\beta t)^{N/2}\ud\bv\approx\rme^{-\beta [F_B(\bv,\nu)-K_B]}(2\beta)^{N/2}N!\\
\times\int_{C_B}\exp\Big[-\frac{x^2}{2t}\Big(1-\frac{v^2}{\gamma_B+\varepsilon_\beta}\Big)\Big]\mu_B(\bx)\ud\bx\ud\bv
\end{multline}
give the scaled probability distributions of the interacting Brownian motions and Bessel processes. The constants $K_A$ and $K_B$ are given by Eqs.~\eqref{ConstantLargeBetaKA} and \eqref{ConstantLargeBetaKB} respectively, and they arise naturally when one considers the leading order terms in $\beta$ of the constants $c_A$ and $c_B$ (see Table~\ref{RadialDunklProcessesQuantities}).

Let us consider the expectation of a test function $h(\bv)$ under the distributions \eqref{ScaledDistributionTypeALargeBeta} and \eqref{ScaledDistributionTypeBLargeBeta},
\begin{IEEEeqnarray}{rCl}
\langle h\rangle_{A,t}&=&\int_{C_A}h(\bv)f_A(t,\sqrt{\beta t}\bv)(\beta t)^{N/2}\ud\bv,\\
\langle h\rangle_{B,t}&=&\int_{C_B}h(\bv)f_B(t,\sqrt{\beta t}\bv)(\beta t)^{N/2}\ud\bv.
\end{IEEEeqnarray}
The objective is to show that after a suitably long time, these expectations converge to
\begin{IEEEeqnarray}{rCl}
\langle h\rangle_{A}&=&\int_{C_A}h(\bv)\frac{\rme^{-\beta F_A(\bv)}}{z_{\beta,A}}\ud\bv\approx\Big(\frac{\beta}{2\pi}\Big)^{N/2}N!\int_{C_A}h(\bv)\rme^{-\beta [F_A(\bv)-K_A]}\ud\bv,\\
\langle h\rangle_{B}&=&\int_{C_B}h(\bv)\frac{\rme^{-\beta F_B(\bv)}}{z_{\beta,B}}\ud\bv\approx(2\beta)^{N/2}N!\int_{C_B}h(\bv)\rme^{-\beta [F_B(\bv,\nu)-K_B]}\ud\bv.\IEEEeqnarraynumspace
\end{IEEEeqnarray}
It is assumed that $h(\bv)$ has, at most, polynomial growth as $v\to\infty.$

The expectation $\langle h\rangle_{B,t}$ will be examined first, as it is simpler than $\langle h\rangle_{A,t}$. Making the substitution $\bu=\bx/\sqrt{t}$ yields the integral
\begin{multline}\label{ApproximatedTExpectationTypeB}
\langle h\rangle_{B,t}=(2\beta t)^{N/2}N!\int_{C_B}\int_{C_B}h(\bv)\rme^{-\beta [F_B(\bv,\nu)-K_B]}\\
\times\rme^{-u^2/2}\rme^{u^2v^2/[2(\gamma_B+\varepsilon_\beta)]}\mu_B(\sqrt{t}\bu)\ud\bu\ud\bv.
\end{multline}

Recalling the asymptotic behavior of the function $\epsilon_\beta$ from Eq.~\eqref{EquationBehaviorEpsilonBeta}, when $uv/\sqrt{\beta}\gg 1$ the product of exponentials becomes
\begin{equation}
\exp\Big[-\beta[F_B(\bv,\nu)-K_B]+\frac{\beta v^2}{8}-\Big(u-\frac{\sqrt{\beta}v}{2}\Big)^2/2\Big].
\end{equation}
Then, because
\begin{equation}
F_B(\bv,\nu)-\frac{v^2}{8}=\frac{3v^2}{8}-\frac{2\nu+1}{2}\sum_{i=1}^N\log |v_i|-\sum_{1\leq i<j\leq N}\log|v_j^2-v_i^2|,
\end{equation}
the argument of the exponential is dominated by the terms $-3\beta v^2/8$, $-u^2/2$ and $-(u-\sqrt{\beta}v/2)^2/2$. This means that in the region where $uv/\sqrt{\beta}\gg 1$, the integrand is exponentially decreasing and this part of the integral can be neglected. 

When $uv/\sqrt{\beta}\ll1$, the expectation $\langle h\rangle_{B,t}$ is approximated by
\begin{multline}\label{ApproximatedTExpectationTypeBInside}
\langle h\rangle_{B,t}\approx(2\beta t)^{N/2}N!\int_{C_B: u<M_1}\int_{C_B: v<M_2}h(\bv)\rme^{-\beta [F_B(\bv,\nu)-K_B]}\\
\times\rme^{-u^2/2}\rme^{u^2v^2/[2(\gamma_B+\varepsilon_\beta)]}\mu_B(\sqrt{t}\bu)\ud\bu\ud\bv,
\end{multline}
where the positive real constants $M_1,M_2$ are chosen so that $M_1 M_2\ll\sqrt{\beta}$. Because $\varepsilon_\beta>0$, the following bound holds in the region of integration in Eq.~\eqref{ApproximatedTExpectationTypeBInside},
\begin{equation}
1\leq \exp\{u^2v^2/[2(\gamma_B+\varepsilon_\beta)]\}\leq \exp\{M_1^2M_2^2/2\gamma_B\}.
\end{equation}

Thus, $\langle h\rangle_{B,t}$ is bounded by the following expressions,
\begin{multline}
\int_{C_B: u<M_1}\rme^{-u^2/2}\mu_B(\sqrt{t}\bu)\ud\bu\int_{C_B: v<M_2}h(\bv)\rme^{-\beta [F_B(\bv,\nu)-K_B]}\ud\bv\\
 \lesssim\frac{\langle h\rangle_{B,t}}{(2\beta t)^{N/2}N!}\lesssim \\
\rme^{\frac{M_1^2M_2^2}{2\gamma_B}}\int_{C_B: u<M_1}\rme^{-u^2/2}\mu_B(\sqrt{t}\bu)\ud\bu\int_{C_B: v<M_2}h(\bv)\rme^{-\beta [F_B(\bv,\nu)-K_B]}\ud\bv.
\end{multline}

Also, the integral over $\bu$ is bounded as follows,
\begin{multline}\label{ApproximatedTExpectationTypeBPiece1}
\rme^{-M_1^2/2}\int_{C_B: u<M_1}\mu_B(\sqrt{t}\bu)\ud\bu\leq \\
\int_{C_B: u<M_1}\rme^{-u^2/2}\mu_B(\sqrt{t}\bu)\ud\bu\leq \\
\int_{C_B: u<M_1}\mu_B(\sqrt{t}\bu)\ud\bu,
\end{multline}
and because $\mu_B(\bx)$ is a probability measure with finite second moments,
\begin{equation}\label{ApproximatedTExpectationTypeBPiece2}
\int_{C_B: u<M_1}\mu_B(\sqrt{t}\bu)\ud\bu=\frac{1}{t^{N/2}}\{1+O[(M_1\sqrt{t})^{-(r-N+2)}]\},
\end{equation}
where $r>N-2$. 

In addition, due to Eq.~\eqref{EquationIntegralOutside} one has (neglecting the polynomial-growing coefficient of the exponential in the correction term)
\begin{equation}\label{ApproximatedTExpectationTypeBPiece3}
\int_{C_B: v<M_2}h(\bv)\rme^{-\beta [F_B(\bv,\nu)-K_B]}\ud\bv=\langle h\rangle_B [1+O(\rme^{-\beta M_2^2/2})],
\end{equation}
where $\langle h\rangle_B$ is the steady-state expectation of $h(\bv)$. 

Inserting Eqs.~\eqref{ApproximatedTExpectationTypeBPiece1}, \eqref{ApproximatedTExpectationTypeBPiece2} and \eqref{ApproximatedTExpectationTypeBPiece3} into Eq.~\eqref{ApproximatedTExpectationTypeB} gives
\begin{multline}
\langle h\rangle_B\rme^{-M_1^2/2}[1+O(\rme^{-\beta M_2^2/2})]\{1+O[(M_1\sqrt{t})^{-(r-N+2)}]\}\\
\lesssim\langle h\rangle_{B,t}\lesssim\\
\langle h\rangle_B\rme^{M_1^2M_2^2/2\gamma_B}[1+O(\rme^{-\beta M_2^2/2})]\{1+O[(M_1\sqrt{t})^{-(r-N+2)}]\}.
\end{multline}
From this expression, it follows that
\begin{equation}\label{AlmostFinalSteadyStateConvergenceTypeB}
\langle h\rangle_{B,t}=\langle h\rangle_B\{1+O[(M_1\sqrt{t})^{-(r-N+2)}]+O[M_1^2]+O[M_1^2M_2^2]\},
\end{equation}
by assuming that the conditions
\begin{equation}\label{ConditionsSteadyStateTypeB}
M_1^2\ll 1,\quad t M_1^2\gg 1,\quad\beta M_2^2\gg 1,\text{ and }M_1^2M_2^2\ll 1
\end{equation}
are satisfied. Setting $M_1\propto t^{-\alpha}$ and $M_2\propto \beta^{-\alpha}$ with $0<\alpha<1/2$, one can always find sufficiently large values of $t$ and $\beta$ such that all the conditions \eqref{ConditionsSteadyStateTypeB} are satisfied. With this choice of $M_1$ and $M_2$, Eq.~\eqref{AlmostFinalSteadyStateConvergenceTypeB} becomes
\begin{equation}\label{FinalSteadyStateConvergenceTypeB}
\langle h\rangle_{B,t}=\langle h\rangle_B\{1+O[t^{-(1/2-\alpha)(r-N+2)}]+O[\max(t^{-2\alpha},(\beta t)^{-2\alpha})]\},
\end{equation}
where $0<\alpha<1/2$, $r>N-2$ and $(\beta t)^{-\alpha}\ll 1$.

The expectation $\langle h\rangle_{A,t}$ can be treated similarly. First, it must be noted that the particle-particle interaction term in Eq.~\eqref{EquationLogSteadyDistributionTypeA} does not depend on the component of $\bv$ that is perpendicular to the root system of type $A$, namely $v_\perp=\bv\cdot\bone/\sqrt{N}$. Recall also that $\bv_\parallel = \bv-(\bv\cdot\bone)\bone/N$. Then, one may write
\begin{equation}\label{PotentialTypeADecomposition}
F_A(\bv)=\frac{v_\parallel^2}{2}-\sum_{1\leq i<j\leq N}\log|v_{j\parallel}-v_{i\parallel}|+\frac{v_\perp^2}{2}=F_A(\bv_\parallel)+\frac{v_\perp^2}{2},
\end{equation}
because $v^2=v_\parallel^2+v_\perp^2$. 

The argument of the exponentials in Eq.~\eqref{ScaledDistributionTypeALargeBeta} is transformed as follows,
\begin{multline}
-\beta [F_A(\bv)-K_A]-\frac{x^2}{2t}+\frac{x_\parallel^2 v^2_\parallel}{2t(\gamma_A +\varepsilon_\beta)}+\sqrt{\frac{\beta}{t}}x_\perp v_\perp\\
=-\beta [F_A(\bv_\parallel)-K_A]-\frac{x^2_\parallel}{2t}+\frac{x_\parallel^2 v^2_\parallel}{2t(\gamma_A +\varepsilon_\beta)}-\frac{1}{2}\Big(\sqrt{\beta}v_\perp-\frac{x_\perp}{\sqrt{t}}\Big)^2.
\end{multline}

With the variable substitution $\bu=\bx/\sqrt{t}$, $\langle h\rangle_{A,t}$ becomes
\begin{multline}\label{CompleteExpectationTypeA}
\langle h\rangle_{A,t}=\Big(\frac{\beta t}{2\pi}\Big)^{N/2}N!\int_{C_A}\int_{C_A}h(\bv)\rme^{-\beta [F_A(\bv_\parallel)-K_A]-u_\parallel^2/2+u_\parallel^2 v_\parallel^2/2(\gamma_A+\varepsilon_\beta)}\\
\times\exp\Big[-\frac{1}{2}\Big(\sqrt{\beta}v_\perp-u_\perp\Big)^2\Big]\mu_A(\sqrt{t}\bu)\ud\bu\ud\bv.
\end{multline}

Let us focus on the integral over $u_\perp$ first. For this purpose, consider a positive and integrable function $f(x)$ with polynomial decay at infinity,
\begin{equation}
f(x)\stackrel{|x|\text{ large}}{\approx} \frac{C}{x^r},
\end{equation}
where $C$ is a positive constant and $r>0$. Consider now the integral
\begin{equation}\label{EquationSimpleIntegral}
\int_{\RR}f(\sqrt{t}x)\rme^{-(\sqrt{\beta}y-x)^2/2}\ud x
\end{equation}
as a simple representation of the integral over $u_\perp$ in Eq.~\eqref{CompleteExpectationTypeA}. Using the variable substitution $x=\sqrt{\beta} u$ and setting $0<\epsilon\ll1$ gives
\begin{equation}
\int_{\RR}f(\sqrt{t}x)\rme^{-(\sqrt{\beta}y-x)^2/2}\ud x=\sqrt{\beta}\Big[\int_{|u|\geq \epsilon}+\int_{|u|< \epsilon}\Big]f(\sqrt{\beta t}u)\rme^{-\beta(y-u)^2/2}\ud u.
\end{equation}

Because $f(x)$ is a positive function, the outer integral can be bounded as follows,
\begin{equation}
\int_{|u|\geq \epsilon}f(\sqrt{\beta t}u)\rme^{-\beta(y-u)^2/2}\ud u\leq \int_{|u|\geq \epsilon} \frac{C}{(\sqrt{\beta t}u)^r}\ud u=O[(\beta t)^{-r/2}\epsilon^{-r+1}].
\end{equation}

By the mean value theorem, there exists a number $u_*$ such that $|u_*|<\epsilon$ and
\begin{equation}
\int_{-\epsilon}^\epsilon f(\sqrt{\beta t}u)\rme^{-\beta(y-u)^2/2}\ud u=\rme^{-\beta(y-u_*)^2/2}\int_{-\epsilon}^\epsilon f(\sqrt{\beta t}u)\ud u.
\end{equation}

Then, by making $\beta$ and $t$ large enough that $\sqrt{\beta t}\epsilon\gg 1$ while keeping $\epsilon\ll 1$, one obtains
\begin{multline}
\int_{\RR}f(\sqrt{t}x)\rme^{-(\sqrt{\beta}y-x)^2/2}\ud x=\rme^{-\beta(y-u_*)^2/2}\int_{-\epsilon\sqrt{\beta}}^{\epsilon\sqrt{\beta}} f(\sqrt{t}x)\ud x\\
+O[t^{-r/2}(\beta \epsilon^2)^{-(r-1)/2}].
\end{multline}
Setting $\epsilon\propto (\beta t)^{-\alpha}$ with $0<\alpha<1/2$ satisfies the assumptions $\sqrt{\beta t}\epsilon\gg 1$ and $\epsilon\ll 1$ for $\beta$ and $t$ sufficiently large. 

Finally, the integral \eqref{EquationSimpleIntegral} gives
\begin{equation}\label{EquationPerpendicularCorrectionTypeA}
\int_{\RR}f(\sqrt{t}x)\rme^{-(\sqrt{\beta}y-x)^2/2}\ud x=\frac{\rme^{-\beta y^2/2}}{\sqrt{t}}\int_{\RR}f(x^\prime)\ud x^\prime
+O[t^{-1/2}(\beta t)^{-(1/2-\alpha)(r-1)}]
\end{equation}
after making $u_*\approx 0$, extending the domain of integration over $x$ to $\RR$ at the expense of a correction term of order $O[(\beta t)^{-r/2}\epsilon^{-r+1}]$ and using the variable substitution $\sqrt{t}x=x^\prime$. With this, the integral over $u_\perp$ in Eq.~\eqref{CompleteExpectationTypeA} can be evaluated directly, and the integral over $\bu_\parallel$ is evaluated in the same way as in the derivation of Eq.~\eqref{FinalSteadyStateConvergenceTypeB}. The only difference is that the integral over $\bu_\parallel$ is an $(N-1)$-dimensional integral, so $N$ must be replaced by $N-1$ in Eq.~\eqref{FinalSteadyStateConvergenceTypeB}. 

Combining the integrals over $\bu_\parallel$ and $u_\perp$ through the use of Eq.~\eqref{FinalSteadyStateConvergenceTypeB} with $N\to N-1$ and Eq.~\eqref{EquationPerpendicularCorrectionTypeA} gives
\begin{equation}
\langle h\rangle_{A,t}=\langle h\rangle_{A}\{1+O[t^{-(1/2-\alpha)(r-N+3)}]+O[t^{-1/2}(\beta t)^{-(1/2-\alpha)(r-1)}]\}.
\end{equation}
With this, the derivation is complete.\qquad\qquad\qquad\qquad\qquad\qquad\qquad\qquad$\square$

\subsection{$\beta\to\infty$ and $\nu\to\infty$ regimes}

Let us focus first on the freezing distribution for the interacting Brownian motions. The extrema of the function $F_A(\bv)$ are located at $\bv=\bh_N$ or any of its permutations, and all extrema are local minima. This is proved as follows. The first-order partial derivatives of $F_A(\bv)$ relative to $\bv$ are
\begin{equation}
\frac{\partial}{\partial v_i}F_A(\bv)=v_i-\sum_{\substack{j:j\neq i\cr j=1}}^N\frac{1}{v_i-v_j}
\end{equation}
Hence, the extrema of $F_A(\bv)$ occur at points $\bv$ which obey the relation
\begin{equation}
v_i=\sum_{\substack{j:j\neq i\cr j=1}}^N\frac{1}{v_i-v_j}.\label{extcond}
\end{equation}
The second order derivatives of $F_A(\bv)$ are
\begin{multline}
\frac{\partial^2}{\partial v_j\partial v_i}\Big[\frac{v^2}{2}-\sum_{1\leq i<j\leq N}\log|v_j-v_i|\Big]=\frac{\partial}{\partial v_j}\Big[v_i-\sum_{l:l\neq i}\frac{1}{v_i-v_l}\Big]\\
=
\left\{
\begin{array}{cl}
1+\sum_{l:l\neq i}\frac{1}{(v_i-v_l)^2} & \textrm{if }i=j,\\
-\frac{1}{(v_i-v_j)^2} & \textrm{if }i\neq j.
\end{array}
\right.\label{secondder}
\end{multline}
The Hessian matrix formed by the $N\times N$ second order derivatives above is positive definite for all vectors $\bv$ such that $v_i\neq v_j$ for $i\neq j$. To show this, we consider an arbitrary real vector $\bu$ and calculate the quadratic form associated to \eqref{secondder}.
\begin{equation}
\sum_{1\leq i,j\leq N}u_i\frac{\partial^2F_A(\bv)}{\partial v_j\partial v_i}u_j=\sum_{i=1}^Nu_i^2+\frac{1}{2}\sum_{1\leq i<j\leq N}\frac{(u_i-u_j)^2}{(v_i-v_j)^2}\geq 0\label{quadraticform}
\end{equation}
Here, the equality holds only when all the ${u_i}$ are equal to zero. Hence, all extrema given by \eqref{extcond} are minima. Suppose that there exists a vector $\bz$ which satisfies Eq.~\eqref{extcond}. Given $\bz$, any of its permutations solve Eq.~\eqref{extcond}:
\begin{equation}
z_{\rho(i)}=\sum_{\substack{j:j\neq i\cr j=1}}^N\frac{1}{z_{\rho(i)}-z_{\rho(j)}}
\end{equation}
for any $\rho\in S_N$. 

Equation~\eqref{extcond} implies that $\{z_i\}_{i=1,\ldots,N}$ must be the roots of the $N$th Hermite polynomial, a fact that is shown as follows. Multiply Eq.~\eqref{extcond} by $\prod_{\substack{l:l\neq i\cr l=1}}^N(z_i-z_j)$. The result is
\begin{equation}
z_{i}\prod_{\substack{l:l\neq i\cr l=1}}^N(z_i-z_l)=\sum_{\substack{j:j\neq i\cr j=1}}^N\prod_{\substack{l:l\neq i,j\cr l=1}}^N(z_i-z_l).\label{eqforzprod}
\end{equation}
Now, define the polynomial whose roots are $\{z_i\}_{i=1,\ldots,N}$ by
\begin{equation}
p_1(x)=c_1 \prod_{n=1}^N(x-z_n),
\end{equation}
with $c_1$ a non-zero constant. The first two derivatives of this polynomial are:
\begin{equation}
p_1^\prime(x)=\frac{\ud}{\ud x}p_1(x)=c_1\sum_{j=1}^N\prod_{\substack{n:n\neq j\cr n=1}}^N(x-z_n)
\end{equation}
and
\begin{equation}
p_1^{\prime\prime}(x)=\frac{\ud^2}{\ud x^2}p_1(x)=2c_1\sum_{1\leq j<l\leq N}\prod_{\substack{n:n\neq j,l\cr n=1}}^N(x-z_n).
\end{equation}
At any of the values $z_i$, $p_1^{\prime\prime}(x)$ behaves as follows.
\begin{equation}
p_1^{\prime\prime}(z_i)=2c_1\sum_{\substack{j:j\neq i\cr j=1}}^N\prod_{\substack{n:n\neq i,j\cr n=1}}^N(z_i-z_n)
\end{equation}
We insert \eqref{eqforzprod} to obtain
\begin{equation}
p_1^{\prime\prime}(z_i)=2c_1z_i\prod_{\substack{n:n\neq i\cr n=1}}^N(z_i-z_n)=2z_ip_1^\prime(z_i).\label{hzeros}
\end{equation}
It is known \cite{szego} that the differential relation on the zeros of the polynomial $p_1(x)$ is only fulfilled by the $N$th Hermite polynomial. Indeed, it solves the differential equation
\begin{equation}
H^{\prime\prime}_N(x)-2xH^\prime_N(x)+2NH_N(x)=0,
\end{equation}
which reduces to \eqref{hzeros} when $x=h_{i,N},$ with $i=1,\ldots,N$ and $h_{i,N}$ is the $i$th root of $H_N(x)$. Hence, $p_1(x)\propto H_N(x),$ and $\bz=\bh_{N}$, as desired.

The value of the function $F_A(\bv)-K_A$ at its minima is
\begin{multline}\label{MinimumPotentialTypeAExpression}
F_A(\bh_N)-K_A=\frac{h_N^2}{2}-\sum_{1\leq i<j\leq N}\log|h_{j,N}-h_{i,N}|\\
-\frac{N}{4}(N-1)(1+\log 2)+\frac{1}{2}\sum_{i=1}^Ni\log i.
\end{multline}
From Eq.~\eqref{extcond}, the squared norm of $\bh_N$ can be calculated as follows,
\begin{equation}\label{MinimumPotentialTypeAPartOne}
h_N^2=\sum_{\substack{1\leq i\neq j\leq N}}\frac{h_{i,N}}{h_{i,N}-h_{j,N}}=\frac{1}{2}\sum_{\substack{1\leq i\neq j\leq N}}\frac{h_{i,N}-h_{j,N}}{h_{i,N}-h_{j,N}}=\frac{N}{2}(N-1)=\gamma_A.
\end{equation}

The term $\sum_{1\leq i<j\leq N}\log|h_{j,N}-h_{i,N}|$ is calculated following Szeg{\"o} \cite{szego}. Using 
\begin{equation}\label{hermitederivative}
H_N^\prime(h_{i,N})=\lim_{x\to h_{i,N}}\frac{H_N(x)}{x-h_{i,N}}=2^N\prod_{\substack{j:j\neq i\cr j=1}}^N(h_{i,N}-h_{j,N}),
\end{equation}
one may write
\begin{multline}
\prod_{1\leq i<j\leq N}(h_{j,N}-h_{i,N})^2=(-1)^{N(N-1)/2}\prod_{1\leq i\neq j\leq N}(h_{j,N}-h_{i,N})\\
=\frac{(-1)^{N(N-1)/2}}{2^{N^2}}\prod_{i=1}^NH_N^\prime(h_{i,N})\\
=\frac{(-1)^{N(N-1)/2}N^N}{2^{N(N-1)}}\prod_{i=1}^NH_{N-1}(h_{i,N}).
\end{multline}
The last equality follows from the derivative relation
\begin{equation}\label{hermitederivativex}
H_N^\prime(x)=2NH_{N-1}(x).
\end{equation}
Let us focus on the last product:
\begin{multline}
\prod_{i=1}^NH_{N-1}(h_{i,N})=2^{N(N-1)}\prod_{i=1}^N\prod_{j=1}^{N-1}(h_{i,N}-h_{j,N-1})\\
=2^{N(N-1)}\prod_{j=1}^{N-1}\prod_{i=1}^{N}(h_{j,N-1}-h_{i,N})=\prod_{j=1}^{N-1}H_N(h_{j,N-1}).
\end{multline}
From the recurrence relation
\begin{equation}
H_N(x)=2xH_{N-1}(x)-2(N-1)H_{N-2}(x),
\end{equation}
it follows that $H_N(h_{j,N-1})=-2(N-1)H_{N-2}(h_{j,N-1})$. Therefore, the product above becomes
\begin{equation}
\prod_{i=1}^NH_{N-1}(h_{i,N})=[-2(N-1)]^{N-1}\prod_{j=1}^{N-1}H_{N-2}(h_{j,N-1}).
\end{equation}
Mathematical induction on the last relation gives
\begin{eqnarray}
\prod_{i=1}^NH_{N-1}(h_{i,N})&=&(-2)^{N(N-1)/2}\prod_{j=1}^{N-1}j^j.
\end{eqnarray}
Therefore,
\begin{multline}
\prod_{1\leq i<j\leq N}(h_{j,N}-h_{i,N})^2=\frac{(-1)^{N(N-1)/2}N^N}{2^{N(N-1)}}\prod_{i=1}^NH_{N-1}(z_{i,N})\\
=\frac{1}{2^{N(N-1)/2}}\prod_{j=1}^{N}j^j.
\end{multline}
Taking the logarithm of the above expression gives
\begin{equation}\label{MinimumPotentialTypeAPartTwo}
2\sum_{1\leq i<j\leq N}\log|h_{j,N}-h_{i,N}|=\sum_{i=1}^Ni\log i-\frac{N}{2}(N-1)\log 2.
\end{equation}

Inserting Eqs.~\eqref{MinimumPotentialTypeAPartOne} and \eqref{MinimumPotentialTypeAPartTwo} into Eq.~\eqref{MinimumPotentialTypeAExpression} finally yields
\begin{equation}
F_A(\bh_N)-K_A=0.
\end{equation}

Using all of the previous information about $F_A(\bv)$, one may take the freezing limit $\beta\to\infty$ from the low-temperature approximation of the scaled distribution $f_A(t,\sqrt{\beta t}\bv)(\beta t)^{N/2}$ given in Eq.~\eqref{ScaledDistributionTypeALargeBeta}, which starts from the arbitrary probability distribution $\mu(\bx)$ defined on $C_A$. Because $F_A(\bv)-K_A$ has its minima at any possible permutation of the vector $\bh_N$ (which we assume is arranged in increasing order), the following limit is obtained,
\begin{equation}\label{FreezingLimitSteadyStateDistributionTypeB}
\lim_{\beta\to\infty}\rme^{-\beta [F_A(\bv)-K_A]}\Big(\frac{\beta}{2\pi}\Big)^{N/2}N!= \sum_{\rho\in S_N}\delta^{(N)}(\bv-\rho\bh_N).
\end{equation}

The only calculation that remains is the freezing limit $\beta\to\infty$ of the integral over $\bx$ in Eq.~\eqref{ScaledDistributionTypeALargeBeta}. The integral is given by the expression
\begin{equation}
\int_{C_A}\exp\Big[-\frac{x^2}{2t}+\frac{x_\parallel^2 v^2_\parallel}{2t(\gamma_A +\varepsilon_\beta)}+\sqrt{\frac{\beta}{t}}x_\perp v_\perp\Big]\mu_A(\bx)\ud\bx.
\end{equation}
Note that $h_{N\perp}=\bh_N\cdot\bone/\sqrt{N}=0$ because the roots of the Hermite polynomials are distributed symmetrically with respect to the origin, meaning that
\begin{equation}
\sum_{i=1}^N h_{i,N}=0. 
\end{equation}
Using Eq.~\eqref{PotentialTypeADecomposition}, one may take a term $-\beta v_\perp^2/2$ from the function $F_A(\bv)$ to rewrite the integral as follows,
\begin{equation}
\int_{C_A}\exp\Big[-\frac{x^2_\parallel}{2t}+\frac{x_\parallel^2 v^2_\parallel}{2t(\gamma_A +\varepsilon_\beta)}-\frac{1}{2}\Big(\frac{x_\perp}{\sqrt t}-\sqrt\beta v_\perp\Big)\Big]\mu_A(\bx)\ud\bx.
\end{equation}
From Eqs.~\eqref{MinimumPotentialTypeAPartOne} and \eqref{CorrectedGammaAtInfinity}, it follows that
\begin{equation}
-\frac{x^2_\parallel}{2t}+\frac{x_\parallel^2 v^2_\parallel}{2t(\gamma_A +\varepsilon_\beta)}\stackrel{\beta\to\infty}{\longrightarrow}0,
\end{equation}
because Eq.~\eqref{FreezingLimitSteadyStateDistributionTypeB} forces $\bv_\parallel$ to be equal to $\bh_{N\parallel}=\bh_N$ in the freezing limit $\beta\to\infty$. 

Consequently, only the freezing limit of the terms involving $v_\perp$ and $x_\perp$ remains to be calculated. From Eq.~\eqref{FreezingLimitGeneralPerpendicularPart} it follows that, in the sense of distributions,
\begin{equation}
\sqrt{\frac{\beta}{2\pi}}\rme^{-\beta (v_\perp-x_\perp/\sqrt{\beta t})^2/2}\stackrel{\beta\to\infty}{\longrightarrow}\delta(v_\perp).
\end{equation}
Therefore, the scaled distribution becomes
\begin{equation}
f_A(t,\sqrt{\beta t}\bv)(\beta t)^{N/2}\stackrel{\beta\to\infty}{\longrightarrow}\sum_{\rho\in S_N}\delta^{(N)}(\bv-\rho\bh_N)=\delta(v_\perp)\sum_{\rho\in S_N}\delta^{(N-1)}(\bv_\parallel-\rho\bh_N).
\end{equation}
However, $f_A(t,\by)$ is normalized in $C_A$, meaning that only the delta function located at $\bh_N$ needs to be taken into account and all the other terms can be neglected. Therefore,
\begin{equation}
f_A(t,\sqrt{\beta t}\bv)(\beta t)^{N/2}\stackrel{\beta\to\infty}{\longrightarrow}\delta^{(N)}(\bv-\bh_N),
\end{equation}
and the calculation is complete.

The freezing distribution for the interacting Bessel processes is given in a similar manner. The extrema of $F_B(\bv,\nu)$ are located at $\bor_{\nu-1/2,N}$, its permutations and sign changes, where $(\bor_{\nu-1/2,N})^2=\bl_{\nu-1/2,N}$. Furthermore, all its extrema are local minima. To prove this, one must first consider the extrema of $F_B(\bv,\nu)$, which are given by
\begin{equation}\label{minconditionb}
v_i^2=\nu+1/2+\sum_{\substack{j:j\neq i\cr j=1}}^N\frac{2v_i^2}{v_i^2-v_j^2}.
\end{equation}
The second derivatives of $F_B(\bv,\nu)$ are given by
\begin{multline}
\frac{\partial^2}{\partial v_j \partial v_i}F_B(\bv,\nu)=\delta_{ij}\Big(1+\frac{2\nu+1}{2v_i^2}\Big)\\+2\Big[\delta_{ij}\sum_{\substack{l:l\neq i\cr l=1}}^N\frac{v_i^2+v_l^2}{(v_i^2-v_l^2)^2}-(1-\delta_{ij})\frac{2v_iv_j}{(v_i^2-v_j^2)^2}\Big].
\end{multline}
The Hessian of $F_B(\bv,\nu)$ is positive definite, because for an arbitrary vector $\bu\in\RR^N$ the following expression is non-negative,
\begin{multline}
\sum_{1\leq i,j\leq N}u_iu_j\frac{\partial^2}{\partial v_j \partial v_i}F_B(\bv,\nu)=\sum_{i=1}^Nu_i^2\Big(1+\frac{2\nu+1}{2v_i^2}\Big)\\
+\sum_{1\leq i\neq j\leq N}\frac{(u_iv_i-u_jv_j)^2+(u_iv_j-u_jv_i)^2}{(v_i^2-v_j^2)^2}\geq 0.
\end{multline}
Therefore, all extrema are minima. 

Now, setting $\bos=(\bu)^2$ in \eqref{minconditionb} yields
\begin{equation}\label{minimizationonr}
s_i=\nu+1/2+\sum_{\substack{j:j\neq i\cr j=1}}^N\frac{2s_i}{s_i-s_j}.
\end{equation}
Let us define the following polynomial,
\begin{equation}
p_2(x)=c_2\prod_{j=1}^N(x-s_j),
\end{equation}
and denote by $p_2^\prime(x)$ and $p_2^{\prime\prime}(x)$ its first and second derivatives, respectively. Evaluating them at $x=s_i$, they become
\begin{IEEEeqnarray}{rCl}
p_2^\prime(s_i)&=&c_2\prod_{\substack{n:n\neq i\cr n=1}}^N(s_i-s_n)\ \textrm{and}\\
p_2^{\prime\prime}(s_i)&=&2c_2\sum_{\substack{j:j\neq i\cr j=1}}^N\prod_{\substack{n:n\neq i,j\cr n=1}}^N(s_i-s_n).
\end{IEEEeqnarray}
Multiplying \eqref{minimizationonr} by $p_N^\prime(s_i)$ yields
\begin{equation}
s_ip_2^{\prime\prime}(s_i)+(\nu+1/2 -s_i)p_2^\prime(s_i)=0,\quad i=1,\ldots,N.
\end{equation}
Comparing this equation with the differential equation obeyed by the Laguerre polynomials \cite{szego},
\begin{equation}
xL_N^{(\alpha)\prime\prime}(x)+(\alpha+1-x)L_N^{(\alpha)\prime}(x)+NL_N^{(\alpha)}(x)=0,
\end{equation}
it follows that $p_2(x)$ must be proportional to $L_N^{(\nu-1/2)}(x)$, and the set $\{s_i\}_{i=1}^N$ must be the set of roots of $L_N^{(\nu-1/2)}(x)$, $\{l_{i,\nu-1/2,N}\}_{i=1}^N$. This, in turn, means that the minima of $F_B(\bv,\alpha+1/2)$ lie at $\bv=(\sqrt{l_{1,\alpha,N}},\ldots,\sqrt{l_{N,\alpha,N}})$ with $\alpha=\nu-1/2$.

Let us define $\bor_{\nu-1/2,N}$ such that
\begin{equation}
(\bor_{\nu-1/2,N})^2=\bl_{\nu-1/2,N}.
\end{equation}
At its minima, the function $F_B(\bv,\nu)-K_B$ takes the value
\begin{multline}\label{minimumvalue}
F_B(\bor_{\nu-1/2,N},\nu)-K_B=\frac{r^2}{2}-\frac{2\nu+1}{4}\sum_{i=1}^N\log r_i^2-\sum_{1\leq i<j\leq N}\log|r_j^2-r_i^2|\\
-\frac{N}{2}(N+\nu-1/2)+\frac{1}{2}\sum_{i=1}^N i\log i+\frac{1}{2}\sum_{i=1}^N(\nu+i-1/2) \log (\nu+i-1/2).
\end{multline}
In this expression, the subindices $N$ and $\nu-1/2$ have been omitted for the sake of brevity, and they will be omitted henceforth except for the cases in which confusion may arise. Because $\bor$ obeys Eq.~\eqref{minconditionb}, it follows that its squared norm is
\begin{equation}\label{partone}
r_{\nu-1/2,N}^2=\sum_{i=1}^N\Big(\nu+1/2+\sum_{\substack{j:j\neq i\cr j=1}}^N\frac{2r_i^2}{r_i^2-r_j^2}\Big)=N(\nu+1/2)+N(N-1)=\gamma_B.
\end{equation}

The second term can be calculated from
\begin{equation}
\sum_{i=1}^N\log r_i^2=\sum_{i=1}^N\log l_{i,\alpha,N}=\log(N!L_N^{(\alpha)}(0)),
\end{equation}
where $\alpha=\nu-1/2$ and $L_N^{(\alpha)}(0)=\frac{1}{N!}\prod_{i=1}^N(\alpha+i)$ \cite{szego}. This gives
\begin{equation}
\sum_{i=1}^N\log r_i^2=\sum_{i=1}^N\log(\alpha+i).\label{parttwo}
\end{equation}

Finally, the third term is calculated following \cite{szego} and in a manner similar to the case of the function $F_A(\bv)$. Consider the expression
\begin{equation}
\prod_{1\leq i<j\leq N}(l_{j,\alpha,N}-l_{i,\alpha,N})^2=(-1)^{N(N+1)/2}(N!)^N\prod_{i=1}^NL_{N}^{(\alpha)\prime}(l_{i,\alpha,N}).
\end{equation}
Using the derivative relation $xL_{N}^{(\alpha)\prime}(x)=NL_{N}^{(\alpha)}(x)-(N+\alpha)L_{N-1}^{(\alpha)}(x)$ combined with Eq.~\eqref{parttwo}, one obtains
\begin{equation}
\prod_{1\leq i<j\leq N}(l_{j,\alpha,N}-l_{i,\alpha,N})^2=\frac{(-1)^{N(N-1)/2}(N!)^N(\alpha+N)^N}{\prod_{j=1}^N(\alpha+j)}\prod_{i=1}^NL_{N-1}^{(\alpha)}(l_{i,\alpha,N}).\label{tobeinserted}
\end{equation}
The product of Laguerre polynomials can be rewritten as
\begin{equation}
\prod_{i=1}^NL_{N-1}^{(\alpha)}(l_{i,\alpha,N})=\frac{N^N}{N!}\prod_{i=1}^{N-1}L_{N}^{(\alpha)}(l_{i,\alpha,N-1}),
\end{equation}
and using the recursion relation $NL_{N}^{(\alpha)}(x)=(-x+2N+\alpha-1)L_{N-1}^{(\alpha)}(x)-(N+\alpha-1)L_{N-2}^{(\alpha)}(x)$ on the expression above yields
\begin{equation}
\prod_{i=1}^NL_{N-1}^{(\alpha)}(l_{i,\alpha,N})=\frac{(-1)^{N-1}(N-1+\alpha)^{N-1}}{(N-1)!}\prod_{i=1}^{N-1}L_{(N-1)-1}^{(\alpha)}(l_{i,\alpha,N-1}).
\end{equation}
Mathematical induction on the last equation gives
\begin{equation}
\prod_{i=1}^NL_{N-1}^{(\alpha)}(l_{i,\alpha,N})=(-1)^{N(N-1)/2}\prod_{i=1}^{N-1}\left(\frac{\alpha+i}{N-i}\right)^i.
\end{equation}
Inserting this expression into Eq.~\eqref{tobeinserted} and taking logarithms on both sides gives
\begin{equation}
2\sum_{1\leq i<j\leq N}\log|l_{j,\alpha,N}-l_{i,\alpha,N}|=\sum_{i=1}^N[(i-1)\log(\alpha+i)+i\log i].\label{partthree}
\end{equation}

Finally, substituting Eqs.~\eqref{partone}, \eqref{parttwo} and \eqref{partthree} into Eq.~\eqref{minimumvalue} gives
\begin{equation}
F_B(\bor_{\nu-1/2,N},\nu)=0.
\end{equation}

Because the root system of type $B$ spans $\RR^N$, the calculation of the freezing limit $\beta\to\infty$ of the scaled distribution $f_B(t,\sqrt{\beta t}\bv)(\beta t)^{N/2}$ is simpler than the case of the root system of type $A$. In this case, the steady state distribution follows the freezing limit
\begin{equation}
\lim_{\beta\to\infty}\rme^{-\beta [F_B(\bv)-K_B]}(2\beta)^{N/2}N!= \sum_{\rho\in W_B}\delta^{(N)}(\bv-\rho\bor_{\nu-1/2,N}),
\end{equation}
because $F_B(\bv)-K_B>0$ whenever $\bv\neq\bor_{\nu-1/2,N}$ or its orbit in $W_B$. It is assumed that all the components of $\bor_{\nu-1/2,N}$ are positive and that they are arranged in increasing order so that $\bor_{\nu-1/2,N}\in C_B$. 

By Eq.~\eqref{RadialDunklKernelTypeBFreezing}, the integral over $\bx$ is given by
\begin{equation}
\int_{C_B}\exp\Big[-\frac{x^2}{2t}\Big(1-\frac{v^2}{\gamma_B+\varepsilon_\beta}\Big)\Big]\mu_B(\bx)\ud\bx,
\end{equation}
and as $\beta\to\infty$, $\varepsilon_\beta\to0$ and $v^2\to r_{\nu-1/2,N}^2=\gamma_B$. Therefore,
\begin{multline}
f_B(t,\sqrt{\beta t}\bv)(\beta t)^{N/2}=\rme^{-\beta [F_B(\bv)-K_B]}(2\beta)^{N/2}N!\\
\times\int_{C_B}\exp\Big[-\frac{x^2}{2t}\Big(1-\frac{v^2}{\gamma_B+\varepsilon_\beta}\Big)\Big]\mu_B(\bx)\ud\bx\\
\stackrel{\beta\to\infty}{\longrightarrow}\sum_{\rho\in W_B}\delta^{(N)}(\bv-\rho\bor_{\nu-1/2,N})\int_{C_B}\mu_B(\bx)\ud\bx\\
=\sum_{\rho\in W_B}\delta^{(N)}(\bv-\rho\bor_{\nu-1/2,N}).
\end{multline}
However, $f_B(t,\sqrt{\beta t}\bv)(\beta t)^{N/2}$ is only defined and normalized in $C_B$, meaning that the delta functions that are outside of $C_B$ can be neglected. Therefore,
\begin{equation}
\lim_{\beta\to\infty}f_B(t,\sqrt{\beta t}\bv)(\beta t)^{N/2}=\delta^{(N)}(\bv-\bor_{\nu-1/2,N}),
\end{equation}
as desired.

In the limit $\nu\to\infty$ of the distribution $f_B(t,\sqrt{\beta\nu t}\bv)(\beta \nu t)^{N/2}$, one must consider the behavior of the expression $F_B(\sqrt{\nu}\bv,\nu)-K_B$ when $\nu$ is much larger than $N-1/2$. The expression
\begin{multline}
F_B(\sqrt{\nu}\bv,\nu)-K_B=\nu\frac{v^2}{2}-\frac{2\nu+1}{4}\Big(N\log \nu+\sum_{i=1}^N\log v_i^2\Big)-\sum_{1\leq i<j\leq N}\log|v_j^2-v_i^2|\\
-\frac{N}{2}(N-1)\log \nu-\frac{N}{2}(N+\nu-1/2)+\frac{1}{2}\sum_{i=1}^N i\log i\\
+\frac{1}{2}\sum_{i=1}^N(\nu+i-1/2) \log (\nu+i-1/2)
\end{multline}
can be approximated by
\begin{multline}
F_B(\sqrt{\nu}\bv,\nu)-K_B\approx\nu\Big[\frac{v^2}{2}-\frac{1}{2}\sum_{i=1}^N\log v_i^2-\frac{N}{2}\Big]-\sum_{1\leq i<j\leq N}\log\nu|v_j^2-v_i^2|
\end{multline}
when $\nu\gg N-1/2$. The expression in parentheses is the function $\tilde{F}_B(\bv)$ defined in Eq.~\eqref{PotentialFTildeTypeB}. Therefore, using Eq.~\eqref{RadialDunklKernelTypeBNuInfinity}, the scaled distribution in this case is given by
\begin{multline}
f_B(t,\sqrt{\beta\nu t}\bv)(\beta \nu t)^{N/2}\approx\rme^{-\beta \nu\tilde{F}_B(\bv)}\prod_{1\leq i<j\leq N}|\nu(v_j^2-v_i^2)|^\beta(2\beta\nu)^{N/2}N!\\
\times\int_{C_B}\rme^{-x^2/2t}\FZ{2/\beta}\Bigg(\frac{(\bx)^2}{2t},(\bv)^2\Bigg)\mu_B(\bx)\ud\bx.
\end{multline}

The function $\tilde{F}_B(\bv)$ has the following first- and second-order derivatives,
\begin{IEEEeqnarray}{rCl}
\frac{\partial \tilde{F}_B}{\partial v_i}&=&v_i-\frac{1}{v_i},\\
\frac{\partial^2 \tilde{F}_B}{\partial v_j \partial v_i}&=&\delta_{ij}\Big(1+\frac{1}{v_i^2}\Big).
\end{IEEEeqnarray}
Consequently, the Hessian of $\tilde{F}_B(\bv)$ is positive definite and all extrema are minima. In addition, the minima lie on all vectors $\bv$ such that $v_i=\pm1$, and the minimum value of $\tilde{F}_B(\bv)$ is zero. Then, for large values of $\beta\nu$ the following approximation holds,
\begin{equation}
\rme^{-\beta \nu\tilde{F}_B(\bv)}\approx\prod_{i=1}^N\sum_{z_i=\pm1}\exp\Big[-\beta \nu(v_i-z_i)^2\Big].
\end{equation}

Consider now the integral
\begin{multline}
\mathcal{E}=\int_{\bar{C}_B}h(\bv)\prod_{i=1}^N\sum_{z_i=\pm1}\rme^{-\beta \nu(v_i-z_i)^2}\prod_{1\leq i<j\leq N}|\nu(v_j^2-v_i^2)|^\beta(2\beta\nu)^{N/2}N!\\
\times\int_{\bar{C}_B}\rme^{-x^2/2t}\FZ{2/\beta}\Bigg(\frac{(\bx)^2}{2t},(\bv)^2\Bigg)\mu_B(\bx)\ud\bx\ud\bv,
\end{multline}
where $h(\bv)$ is a test function with polynomial growth at infinity. Define the following subset of the closure of $C_B$, $\mathcal{D}_\epsilon=\{\by\in \bar{C}_B: 1-\epsilon\leq y_1\leq \ldots\leq y_N\leq 1+\epsilon\}$, where $0<\epsilon\ll 1$. At very large values of $\beta\nu$, one has
\begin{multline}
\int_{\bar{C}_B\backslash\mathcal{D}_\epsilon}h(\bv)\prod_{i=1}^N\sum_{z_i=\pm1}\rme^{-\beta \nu(v_i-z_i)^2}\prod_{1\leq i<j\leq N}|\nu(v_j^2-v_i^2)|^\beta(2\beta\nu)^{N/2}N!\\
\times\int_{\bar{C}_B}\rme^{-x^2/2t}\FZ{2/\beta}\Bigg(\frac{(\bx)^2}{2t},(\bv)^2\Bigg)\mu_B(\bx)\ud\bx\ud\bv\\
=O[\rme^{-\beta\nu\epsilon^2}],
\end{multline}
because the Gaussian term dominates the integrand away from $(\bv)^2=\bone$. Therefore, if $\epsilon$ is chosen small while keeping $\beta\nu\epsilon^2$ very large, this part of the integral can be neglected. For this purpose, set $\epsilon\propto \nu^{-\alpha}$ with $0<\alpha<1/2$. Then, the integral over $\mathcal{D}_\epsilon$ is simplified using the mean value theorem as
\begin{multline}
\int_{\mathcal{D}_\epsilon}h(\bv)\prod_{i=1}^N\sum_{z_i=\pm1}\rme^{-\beta \nu(v_i-z_i)^2}\prod_{1\leq i<j\leq N}|\nu(v_j^2-v_i^2)|^\beta(2\beta\nu)^{N/2}N!\\
\times\int_{\bar{C}_B}\rme^{-x^2/2t}\FZ{2/\beta}\Bigg(\frac{(\bx)^2}{2t},(\bv)^2\Bigg)\mu_B(\bx)\ud\bx\ud\bv\\
=(2\beta\nu)^{N/2}N!\prod_{1\leq i<j\leq N}|\nu(v_{j*}^2-v_{i*}^2)|^\beta\int_{\mathcal{D}_\epsilon}h(\bv)\prod_{i=1}^N\sum_{z_i=\pm1}\rme^{-\beta \nu(v_i-z_i)^2}\\
\times\int_{\bar{C}_B}\rme^{-x^2/2t}\FZ{2/\beta}\Bigg(\frac{(\bx)^2}{2t},(\bv)^2\Bigg)\mu_B(\bx)\ud\bx\ud\bv,
\end{multline}
where $\bv_*\in\mathcal{D}_\epsilon$. Then, the components of $\bv_*$ have the property that
\begin{equation}
v_{i*}=1+O(\epsilon),
\end{equation}
and consequently
\begin{equation}
v_{j*}^2-v_{i*}^2=2O(\epsilon)+O(\epsilon^2)=O(\epsilon).
\end{equation}
Thus, the order of magnitude of the product of differences is given by
\begin{equation}
\prod_{1\leq i<j\leq N}|\nu(v_{j*}^2-v_{i*}^2)|^\beta=\prod_{1\leq i<j\leq N}|O(\nu\epsilon)|^\beta=O(\nu^{(1-\alpha)\beta N(N-1)/2}).
\end{equation}
This means that as $\nu\to\infty$, the product of differences tends to infinity. Therefore, it makes sense to write 
\begin{multline}
\lim_{\nu\to\infty}\mathcal{E}\propto\int_{C_B}h(\bv)\prod_{i=1}^N\sum_{z_i=\pm1}\delta(v_i-z_i)\\
\times\int_{C_B}\rme^{-x^2/2t}\FZ{2/\beta}\Bigg(\frac{(\bx)^2}{2t},(\bv)^2\Bigg)\mu_B(\bx)\ud\bx\ud\bv\\
=\int_{C_B}h(\bv)\prod_{i=1}^N\sum_{z_i=\pm1}\delta(v_i-z_i)\ud\bv\int_{C_B}\rme^{-x^2/2t}\FZ{2/\beta}\Bigg(\frac{(\bx)^2}{2t},\bone\Bigg)\mu_B(\bx)\ud\bx\\
=h(\bone)\int_{C_B}\rme^{-x^2/2t}\FZ{2/\beta}\Bigg(\frac{(\bx)^2}{2t},\bone\Bigg)\mu_B(\bx)\ud\bx.
\end{multline}

From Eqs.~(2.8) and (3.2b) in \cite{bakerforrester97}, it is known that
\begin{equation}
\FZ{2/\beta}\Bigg(\frac{(\bx)^2}{2t},\bone\Bigg)=\exp\Bigg(\frac{x^2}{2t}\Bigg),
\end{equation}
which finally gives
\begin{equation}
\lim_{\nu\to\infty}\mathcal{E}\propto h(\bone)\int_{C_B}\rme^{-x^2/2t}\rme^{x^2/2t}\mu_B(\bx)\ud\bx=h(\bone),
\end{equation}
or, in the sense of distributions,
\begin{equation}
\lim_{\nu\to\infty}f_B(t,\sqrt{\beta \nu t}\bv)(\beta \nu t)^{N/2}\ud\bv\propto\delta^{(N)}(\bv-\bone)\ud\bv.
\end{equation}
The proportionality constant is one because both members of the expression are normalized to one in $C_B$.\qquad\qquad\qquad\qquad\qquad\qquad\qquad\qquad\qquad\quad\ $\square$
% !TEX encoding = UTF-8 Unicode
\chapter{Summary of results and future prospects}\label{conclusions}

In the present thesis, the behavior of the interacting Brownian motions and Bessel processes in the steady state and freezing regimes was investigated through the use of Dunkl operator theory. After the brief review of Dunkl theory given in Chapter~\ref{preliminaries}, the correspondence between the Calogero-Moser systems and Dunkl processes was proved in Chapter~\ref{CalogeroMoserCorrespondence} (Prop.~\ref{correspondencer}). This correspondence served as motivation for the fact that Dunkl processes, after given an appropriate scaling, converge to a steady state and have a well-defined freezing limit. 

In Chapters~\ref{general_steady} and \ref{general_freezing}, the main results of this thesis were proved. The first result is that the scaled final distribution of a Dunkl process that starts from an initial distribution with finite second moments will converge to a specific steady-state distribution (Thm.~\ref{TheoremSteadyState}). The second result is that the scaled final distribution of a Dunkl process that starts from an arbitrary initial distribution freezes to a configuration that is given by the peak set of the root system $R$ instantaneously (Thm.~\ref{TheoremFreezingLimit}). The proof of these results depended on several calculations involving the intertwining operator, in particular the action of $V_\beta$ on linear polynomials (Lemma~\ref{LemmaFirstOrderV}), and on the exponential function in the freezing limit (Lemma~\ref{FreezingLimitDunklKernel}) as well as other approximations. While a finite lower bound was given for the time required for Dunkl processes to converge to the steady state, it seems that there must be a better estimation of the relaxation time in view of the fact that the freezing configuration is achieved instantaneously. This fact suggests that the relaxation time should be inversely proportional to the inverse temperature. This improvement on the estimation of the relaxation time is left as an open problem.

Because both the interacting Brownian motions and Bessel processes are particular cases of Dunkl processes, it follows from Thms.~\ref{TheoremSteadyState} and \ref{TheoremFreezingLimit} that these two systems of interacting particles have well-defined steady-state and freezing regimes. Chapter~\ref{ParticularCases} was devoted to these particular cases. The interacting Brownian motions converge to a steady state in which their scaled distribution coincides with the $\beta$-Hermite ensembles of random matrices, and freeze to a scaled distribution given by delta functions centered at the zeroes of the Hermite polynomials (Prop.~\ref{LimitingRegimesA}). Similarly, the interacting Bessel processes converge to a steady state in which their scaled distribution coincides with the $\beta$-Laguerre ensembles of random matrices, and freeze to a scaled distribution given by delta functions centered at the zeroes of the Laguerre polynomials; in addition, in the limit where the Bessel index tends to infinity (which, translated to the Dirac field of QCD corresponds to the case where the topological charge tends to infinity), all the particles converge to the same scaled position (Prop.~\ref{LimitingRegimesB}).

Prior to the derivation of Props.~\ref{LimitingRegimesA} and \ref{LimitingRegimesB}, the behavior of the interacting Brownian motions and Bessel processes in the steady-state and freezing regimes was studied using numerical simulations. As setup for the proof of these propositions, the action of $V_\beta$ on symmetric polynomials was derived in Prop.~\ref{PropositionVBetaOnSymmetricPolynomials}, and the freezing limit of the generalized Bessel functions of type $A$ and $B$ was obtained in Prop.~\ref{FreezingLimitOfGeneralizedBesselFunctions}. Because the expressions for the action of $V_\beta$ found in Chapter~\ref{ParticularCases} only apply to symmetric polynomials, it is of interest to examine its action on non-symmetric polynomials. This is a problem that should be tackled in the near future. 

The general results in Thms.~\ref{TheoremSteadyState} and \ref{TheoremFreezingLimit} correspond to two regimes where the probability distribution of the process is balanced in such a way that the probability that the Dunkl process is in one particular Weyl chamber is evenly distributed among the chambers. This means that the process density is invariant under reflections along $\balpha\in R$, and the information about the jumps is lost. It is of great interest to study the physical nature of the jumps performed by Dunkl processes and their effect on the relaxation to the steady state, which is a problem that has not been solved yet. In particular, it is of interest to see if the behavior of the jumps in Dunkl processes has a relationship with a physical phenomenon.

In addition, the numerical results from Chapter~\ref{ParticularCases} (Figs.~\ref{FigureFreezingA} and \ref{FigureFreezingB} in particular) seem to suggest the existence of a transition from a disordered to an ordered phase as $\beta\to\infty$. However, in order to verify the existence of a phase transition the calculation of other physically relevant quantities (e.g., correlation functions) is required. Because the intertwining operator is responsible for the time evolution of Dunkl processes, the calculation of correlations in equilibrium should require different techniques from the ones used in this work. However, the results related to the intertwining operator obtained here open the possibility of studying dynamical and multi-time correlations for these processes. As a consequence, further study of the intertwining operator is essential to investigate whether these processes undergo phase transitions out of equilibrium. This is another topic that has not been addressed and that we would like to study as a continuation of this work.

Finally, it is worth noting that the interacting particle systems studied here seem to have little relationship with actual experiments. Because in the present most of the applications of random matrix theory correspond to ensembles where $\beta=2$, it is not unlikely that many other possible applications have been overlooked because the necessary tools for the study of cases where $\beta>0$ are incomplete. Hopefully, this work will be a stepping stone towards a better understanding of the models treated here for $\beta>0$, and towards finding out what makes the cases $\beta=2$ so special both in terms of their applications in physics and of their mathematical properties.

\appendix
\chapter{Proof of the kernel-reproducing formula}\label{TheUsefulIntegral}

The objective of this appendix is to give a proof of the integral \eqref{GaussianIntegralDunklKernel}, following \cite{dunkl91}, \cite{rosler08} and \cite{dejeu06}. This formula requires the proof of several facts, so the first section is concerned with the tools necessary for the proof. The actual proof of the formula is given in the second section.

\section{Preparations}

The first tool required for the proof is the inner product between polynomials known as the Fischer product, which is defined as follows. Consider two polynomials of $N$ variables, $p(\bx)$ and $q(\bx)$. The expression $p(\bnabla)$ denotes the operator that is obtained by replacing the coordinates $\{x_i\}_{1\leq i\leq N}$ with their partial derivatives $\{\partial/\partial x_i\}_{1\leq i\leq N}$. The Fischer product is defined as
\begin{equation}
(p,q)_0:=p(\bnabla)q(\bx)|_{\bx=\bzero}.\label{FischerProduct}
\end{equation}
Note that monomials are orthogonal under this product, because the expression
\begin{equation}
\prod_{i=1}^N\frac{\partial^{\lambda_i}}{\partial x_i^{\lambda_i}}x_i^{\mu_i}\Big|_{\bx=\bzero}
\end{equation}
vanishes unless the multi-indices $\lambda$ and $\mu$ are equal. Therefore, this product is symmetric, i.e., 
\begin{equation}
(p,q)_0=(q,p)_0.
\end{equation}
It follows that the Fischer product of homogeneous polynomials of different degrees is equal to zero.

%Another property of the Fischer product is that%one can write any $n$th order polynomial as an expansion in Fisher products:
%\begin{equation}
%p(\bx)=\sum_{j=0}^n(E^{(j)}(\bx,\cdot),p)_0=\sum_{j=0}^nE^{(j)}(\bx,\bnabla^{(y)})p(\by)|_{\by=\bzero}.
%\end{equation}

Denote the Dunkl gradient by $\bT=(T_1,\ldots, T_N)^T$. The Dunkl generalization of the Fischer product is given by
\begin{equation}
(p,q)_\beta:=p(\bT)q(\bx)|_{\bx=\bzero}.\label{DunklFischerProduct}
\end{equation}
Like the Fischer product of Eq.~\eqref{FischerProduct}, this Fischer product is symmetric, and homogeneous polynomials of different degrees are orthogonal under it.

An important property of this product is that, denoting the $j$th term of the Taylor expansion of the exponential $\exp(\bx\cdot\by)$ by
\begin{equation}
E^{(j)}(\bx,\by):=\frac{(\bx\cdot\by)^j}{j!},
\end{equation}
one has that for any homogeneous polynomial $p(\bx)$ of degree $n$, the expression
\begin{equation}
(V_\beta E^{(n)}(\bx,\cdot),p)_\beta=p(\bx)
\end{equation}
holds. To prove this, $V_\beta$ must be shown to be one-to-one. This fact follows from the existence of its inverse, which is given by
\begin{equation}
U_\beta f(\bx):=\exp(\bx\cdot\bT^{(y)})f(\by)|_{\by=\bzero}\label{InverseIntertwiningOperator}
\end{equation}
for an arbitrary analytical function $f(\bx)$. The superscript $(y)$ indicates the variable acted upon whenever confusions may arise. To prove that $U_\beta$ is the inverse of $V_\beta$, it suffices to verify that $U_\beta$ satisfies the equation
\begin{multline}
\frac{\partial}{\partial x_i}U_\beta f(\bx)=\frac{\partial}{\partial x_i}\exp(\bx\cdot\bT^{(y)})f(\by)|_{\by=\bzero}=\exp(\bx\cdot\bT^{(y)})T_i^{(y)}f(\by)|_{\by=\bzero}\\
=U_\beta [T_i f(\bx)].
\end{multline}
This is the inverse of Eq.~\eqref{EquationVDefinition}, meaning that
\begin{equation}
U_\beta V_\beta f(\bx)=V_\beta U_\beta f(\bx)=f(\bx).
\end{equation}
Consequently, $V_\beta$ is one-to-one, and $U_\beta$ is linear and preserves the degree of homogeneous polynomials. 

Now, the Taylor expansion of a function $f(\bx)$ at the point $\by$ can be written as follows:
\begin{equation}
f(\bx)=\exp[\bx\cdot\bnabla^{(y)}]f(\by).
\end{equation}
For the homogeneous polynomial $p(\bx)$, this becomes
\begin{equation}
p(\bx)=E^{(n)}(\bx,\bnabla^{(y)})p(\by).
\end{equation}
Applying $V_\beta^{(y)}$ and then $V_\beta^{(x)}$ on both sides gives
\begin{equation}
V_\beta^{(x)}p(\bx)=V_\beta^{(x)}E^{(n)}(\bx,\bT^{(y)})V_\beta^{(y)}p(\by).
\end{equation}
Because $V_\beta$ is one-to-one, one may replace $V_\beta p(\bx)$ with an arbitrary homogeneous polynomial, say, $q(\bx)$. Also, this equation is valid for any $\by$, so taking $\by=\bzero$ gives
\begin{equation}\label{PartialSumDunklKernel}
q(\bx)=V_\beta^{(x)}E^{(n)}(\bx,\bT^{(y)})q(\by)|_{\by=\bzero}=(V_\beta E^{(n)}(\bx,\cdot),q)_\beta,
\end{equation}
as claimed.

The second tool required for the proof of Eq.~\eqref{GaussianIntegralDunklKernel} is the following theorem due to Dunkl (\cite{dunkl91}, Thm.~3.10). Here, $\Delta_\beta$ denotes the Dunkl Laplacian $\sum_{i=1}^N T_i^2$.
\begin{proposition}\label{NiceProposition}
For arbitrary polynomials $p(\bx)$ and $q(\bx)$, the Fischer product \eqref{DunklFischerProduct} can be written as
\begin{equation}
(p,q)_\beta=\frac{1}{c_\beta}\int_{\RR^N}[\rme^{-\Delta_\beta/2}p(\bx)][\rme^{-\Delta_\beta/2}q(\bx)]\rme^{-x^2/2}w_\beta(\bx)\ud\bx.
\end{equation}
\end{proposition}
\begin{proof}
The proof of this proposition is not at all trivial, and the first proof given by Dunkl is rather long and technical. The simpler and shorter proof due to de Jeu (\cite{dejeu06}, pages 4230-4231) will be followed here. The content from this point until Eq.~\eqref{CoolSimpleRelationship} concerns several relationships that will be necessary for the proof.

First, let us show that, for $f(\bx)$ a rapidly decreasing function at infinity and $g(\bx)$ a continuous, bounded and differentiable function, and for all $i=1,\ldots,N$,
\begin{equation}\label{AntisymmetricIntegral}
\int_{\RR^N}[T_i f(\bx)]g(\bx)w_\beta(\bx)\ud\bx=-\int_{\RR^N}f(\bx)[T_i g(\bx)]w_\beta(\bx)\ud\bx.
\end{equation}
The integral on the l.h.s.\ gives
\begin{multline}
\int_{\RR^N}[T_i f(\bx)]g(\bx)w_\beta(\bx)\ud\bx\\
=\int_{\RR^N}\Big[\frac{\partial }{\partial x_i} f(\bx)+\frac{\beta}{2}\sum_{\balpha\in R_+}\alpha_i\kappa(\balpha)\frac{f(\bx)-f(\sigma_{\balpha}\bx)}{\balpha\cdot\bx}\Big]g(\bx)w_\beta(\bx)\ud\bx.
\end{multline}
The derivative term can be treated by integrating by parts:
\begin{multline}
\int_{\RR^N}\Big[\frac{\partial }{\partial x_i} f(\bx)\Big]g(\bx)w_\beta(\bx)\ud\bx=-\int_{\RR^N}f(\bx)\Big[\frac{\partial }{\partial x_i}[g(\bx)w_\beta(\bx)] \Big]\ud\bx\\
=-\int_{\RR^N}f(\bx)\Big[w_\beta(\bx)\frac{\partial }{\partial x_i}g(\bx)+w_\beta(\bx)g(\bx)\beta\sum_{\balpha\in R_+}\frac{\alpha_i\kappa(\balpha)}{\balpha\cdot\bx}\Big]\ud\bx,
\end{multline}
where
\begin{equation}
\frac{\partial}{\partial x_i}w_\beta(\bx)=w_\beta(\bx)\beta\sum_{\balpha\in R_+}\frac{\alpha_i\kappa(\balpha)}{\balpha\cdot\bx}.
\end{equation}
Using the substitution $\bx^\prime=\sigma_{\balpha} \bx$, the difference term becomes
\begin{multline}
\frac{\beta}{2}\sum_{\balpha\in R_+}\alpha_i\kappa(\balpha)\int_{\RR^N}\Big[\frac{f(\bx)-f(\sigma_{\balpha}\bx)}{\balpha\cdot\bx}\Big]g(\bx)w_\beta(\bx)\ud\bx=\\
\frac{\beta}{2}\sum_{\balpha\in R_+}\alpha_i\kappa(\balpha)\Big[\int_{\RR^N}\frac{f(\bx)g(\bx)w_\beta(\bx)}{\balpha\cdot\bx}\ud\bx-\int_{\RR^N}\frac{f(\sigma_{\balpha}\bx)g(\bx)w_\beta(\bx)}{\balpha\cdot\bx}\ud\bx\Big]=\\
\frac{\beta}{2}\sum_{\balpha\in R_+}\alpha_i\kappa(\balpha)\Big[\int_{\RR^N}\frac{f(\bx)g(\bx)w_\beta(\bx)}{\balpha\cdot\bx}\ud\bx+\int_{\RR^N}\frac{f(\bx)g(\sigma_{\balpha}\bx)w_\beta(\bx)}{\balpha\cdot\bx}\ud\bx\Big].
\end{multline}
In the last line, $\bx^\prime$ has been written as $\bx$ for simplicity. Also, the fact that $w_\beta(\sigma_{\balpha}\bx)=w_\beta(\bx)$ has been used, and  the Jacobian for this variable substitution is equal to 1 because the reflection operator $\sigma_{\balpha}$ is represented by an orthogonal matrix. Adding the derivative and the difference terms, yields the desired result.

The second step is to prove the following commutation relation:
\begin{equation}
\Big[x_i,\frac{1}{2}\Delta_\beta\Big]f(\bx)=\frac{x_i}{2}\Delta_\beta f(\bx)-\frac{1}{2}\Delta_\beta[x_if(\bx)]=-T_if(\bx)
\end{equation}
for $i=1,\ldots,N.$ From simple calculations one can obtain the following relations:
\begin{IEEEeqnarray}{rCl}
\balpha\cdot\bnabla[x_i f(\bx)]&=&\alpha_i f(\bx)+x_i\balpha\cdot\bnabla f(\bx),\\
\frac{1}{2}\Delta[x_i f(\bx)]&=&\frac{\partial}{\partial x_i}f(\bx)+\frac{x_i}{2}\Delta f(\bx),\\
\frac{\alpha^2}{2}\frac{(1-\sigma_{\balpha})[x_i f(\bx)]}{(\balpha\cdot\bx)^2}&=&x_i\frac{\alpha^2}{2}\frac{(1-\sigma_{\balpha})f(\bx)}{(\balpha\cdot\bx)^2}+\alpha_i\frac{f(\sigma_{\balpha}\bx)}{\balpha\cdot\bx}.
\end{IEEEeqnarray}
Equation~\eqref{EquationDunklLaplacian} combined with the three previous expressions gives
\begin{multline}
\frac{1}{2}\Delta_\beta[x_i f(\bx)]=\frac{\partial}{\partial x_i}f(\bx)+\frac{x_i}{2}\Delta f(\bx)\\
+\frac{\beta}{2}\sum_{\balpha\in R+}\kappa(\balpha)\Big[\alpha_i\frac{f(\bx)}{\balpha\cdot\bx}+x_i\frac{\balpha\cdot\bnabla f(\bx)}{(\balpha\cdot\bx)}-x_i\frac{\alpha^2}{2}\frac{(1-\sigma_{\balpha})f(\bx)}{(\balpha\cdot\bx)^2}-\alpha_i\frac{f(\sigma_{\balpha}\bx)}{\balpha\cdot\bx}\Big]\\
=T_i f(\bx)+\frac{x_i}{2}\Delta_\beta f(\bx),
\end{multline}
which is the desired result.

A consequence of this commutation relation is that
\begin{equation}\label{UsefulCommutator}
[x_i,\rme^{-\Delta_\beta/2}]=T_i \rme^{-\Delta_\beta/2}.
\end{equation}
This is because, using the relation $[x_i,\Delta_\beta/2]=-T_i$ and the mathematical induction method on $n$, it can be proved that 
\begin{equation}
\Big[x_i,\Big(\frac{\Delta_\beta}{2}\Big)^{n}\Big]=-nT_i\Big(\frac{\Delta_\beta}{2}\Big)^{n-1}.
\end{equation}
Taking the sum $\sum_{n=0}^\infty (-1)^n (n!)^{-1}$ on both sides yields
\begin{equation}
\Big[x_i,\sum_{n=0}^\infty \frac{1}{n!}\Big(-\frac{\Delta_\beta}{2}\Big)^{n}\Big]=T_i\sum_{n=1}^\infty \frac{1}{(n-1)!}\Big(-\frac{\Delta_\beta}{2}\Big)^{n-1},
\end{equation}
which gives the claimed result.

Using Eq.~\eqref{UsefulCommutator}, the following useful relationship can be derived:
\begin{multline}
T_i[\rme^{-x^2/2}(\rme^{-\Delta_\beta/2}p(\bx))]=T_i(\rme^{-x^2/2})[\rme^{-\Delta_\beta/2}p(\bx)]+\rme^{-x^2/2}T_i[\rme^{-\Delta_\beta/2}p(\bx)]\\
=-x_i\rme^{-x^2/2}[\rme^{-\Delta_\beta/2}p(\bx)]+\rme^{-x^2/2}T_i[\rme^{-\Delta_\beta/2}p(\bx)]\\
=-\rme^{-x^2/2}\{\rme^{-\Delta_\beta/2}[x_ip(\bx)]+T_i\rme^{-\Delta_\beta/2}p(\bx)\}+\rme^{-x^2/2}T_i[\rme^{-\Delta_\beta/2}p(\bx)]\\
=-\rme^{-x^2/2}\rme^{-\Delta_\beta/2}[x_ip(\bx)].\label{ThreeTermRotation}
\end{multline}
The first line follows from Eq.~\eqref{DunklProductRule} and the fact that $\rme^{-x^2/2}$ is $W$-invariant, and the third line follows from Eq.~\eqref{UsefulCommutator}. As a final preparation, let us consider the case where $q(\bx)$ is homogeneous and $p(\bx)$ is an arbitrary polynomial. Then, replacing $\bx$ for $\bT$ in $q(\bx)$ and using Eq.~\eqref{ThreeTermRotation}, it follows that
\begin{equation}\label{CoolRelationship}
q(\bT)[\rme^{-x^2/2}(\rme^{-\Delta_\beta/2}p(\bx))]=(-1)^{\text{deg }q}\rme^{-x^2/2}\rme^{-\Delta_\beta/2}[p(\bx)q(\bx)],
\end{equation}
and in particular, when $p(\bx)=1$,
\begin{equation}
q(\bT)\rme^{-x^2/2}=(-1)^{\text{deg }q}\rme^{-x^2/2}\rme^{-\Delta_\beta/2}q(\bx).\label{CoolSimpleRelationship}
\end{equation}

The actual proof is as follows. Consider now the following integral for two arbitrary polynomials $p(\bx)$ and $q(\bx)$:
\begin{equation}
\mathcal{M}_\beta(p,q)=\frac{1}{c_\beta}\int_{\RR^N}[\rme^{-\Delta_\beta/2}p(\bx)][\rme^{-\Delta_\beta/2}q(\bx)]\rme^{-x^2/2}w_\beta(\bx)\ud\bx.
\end{equation}
Let us show that its value coincides with $(p,q)_\beta$ when $p(\bx)=1$. On one hand, by definition,
\begin{equation}
(1,q)_\beta=q(\bzero).
\end{equation}
On the other hand, for $q(\bx)$ homogeneous,
\begin{multline}
\mathcal{M}_\beta(1,q)=\frac{1}{c_\beta}\int_{\RR^N}[\rme^{-\Delta_\beta/2}q(\bx)]\rme^{-x^2/2}w_\beta(\bx)\ud\bx\\
=\frac{1}{c_\beta}\int_{\RR^N}(-1)^{\text{deg }q}[q(\bT^{(x)})\rme^{-x^2/2}]w_\beta(\bx)\ud\bx,
\end{multline}
due to Eq.~\eqref{CoolSimpleRelationship}. This integral can be expressed as a Dunkl transform,
\begin{equation}
\mathcal{E}(\bxi)=\frac{1}{c_\beta}\int_{\RR^N}[V_\beta\rme^{-\rmi \bx\cdot\bxi}](-1)^{\text{deg }q}[q(\bT^{(x)})\rme^{-x^2/2}]w_\beta(\bx)\ud\bx,
\end{equation}
evaluated at $\bxi=\bzero.$ This function can be transformed as follows:
\begin{multline}
\mathcal{E}(\bxi)=\frac{1}{c_\beta}\int_{\RR^N}[q(\bT^{(x)})V_\beta\rme^{-\rmi \bx\cdot\bxi}]\rme^{-x^2/2}w_\beta(\bx)\ud\bx\\
=\frac{1}{c_\beta}\int_{\RR^N}[q(-\rmi\bxi)V_\beta\rme^{-\rmi \bx\cdot\bxi}]\rme^{-x^2/2}w_\beta(\bx)\ud\bx.
\end{multline}
The first equality follows from Eq.~\eqref{AntisymmetricIntegral}, and the second follows from the definition of the Dunkl kernel. Setting $\bxi=\bzero$ gives
\begin{equation}
\mathcal{M}_\beta(1,q)=\mathcal{E}(\bzero)=\frac{q(\bzero)}{c_\beta}\int_{\RR^N}\rme^{-x^2/2}w_\beta(\bx)\ud\bx=q(\bzero).
\end{equation}
The last equality is due to the definition of $c_\beta$. This equation can be extended to non-homogeneous polynomials because $\mathcal{M}_\beta(p,q)$ is bilinear. Therefore, for arbitrary $q(\bx)$,
\begin{equation}
\mathcal{M}_\beta(1,q)=(1,q)_\beta=q(\bzero).\label{FischerProductsAlmostEqual}
\end{equation}

The product $(p,q)_\beta$ has the following property by definition: setting $r_i(\bx)=x_i p(\bx)$, one has
\begin{equation}
(r_i,q)_\beta=r_i(\bT^{(y)})q(\by)|_{\by=\bzero}=p(\bT^{(y)})T_i^{(y)}q(\by)|_{\by=\bzero}=(p,T_iq)_\beta.
\end{equation}
This means that, because $(p,q)_\beta$ is bilinear,
\begin{equation}
(p,q)_\beta=(1,p(\bT)q)_\beta
\end{equation}
for arbitrary polynomials $p(\bx)$ and $q(\bx)$. The integral $\mathcal{M}_\beta(p,q)$ shares this property:
\begin{multline}
\mathcal{M}_\beta(r_i,q)=\frac{1}{c_\beta}\int_{\RR^N}\{\rme^{-\Delta_\beta/2}[x_ip(\bx)]\}[\rme^{-\Delta_\beta/2}q(\bx)]\rme^{-x^2/2}w_\beta(\bx)\ud\bx\\
=\frac{1}{c_\beta}\int_{\RR^N}\{-T_i[\rme^{-x^2/2}(\rme^{-\Delta_\beta/2}p(\bx))]\}[\rme^{-\Delta_\beta/2}q(\bx)]w_\beta(\bx)\ud\bx\\
=\frac{1}{c_\beta}\int_{\RR^N}\rme^{-x^2/2}[\rme^{-\Delta_\beta/2}p(\bx)]\{T_i[\rme^{-\Delta_\beta/2}q(\bx)]\}w_\beta(\bx)\ud\bx\\
=\mathcal{M}_\beta(p,T_i q).
\end{multline}
The second line requires Eq.~\eqref{ThreeTermRotation}, while the third line follows from Eq.~\eqref{AntisymmetricIntegral}. Once more, due to the bilinearity of $\mathcal{M}_\beta(p,q)$, one can write
\begin{equation}
\mathcal{M}_\beta(p,q)=\mathcal{M}_\beta(1,p(\bT)q),
\end{equation}
for any polynomials $p(\bx)$ and $q(\bx)$. Then, it follows that
\begin{equation}
\mathcal{M}_\beta(p,q)=\mathcal{M}_\beta(1,p(\bT)q)=(1,p(\bT)q)_\beta=(p,q)_\beta
\end{equation}
due to Eq.~\eqref{FischerProductsAlmostEqual}, proving the statement.
\end{proof}

\section{Proof}

Consider now the expansion to $n$th order of the Dunkl kernel $V_\beta\rme^{\bx\cdot\by}$:
\begin{equation}
L_\beta^{(n)}(\bx,\by):=\sum_{j=0}^n\frac{1}{j!}V_\beta(\bx\cdot\by)^j.
\end{equation}
For an arbitrary polynomial $p(\bx)$ of degree $m\leq n$, Prop.~\ref{NiceProposition} gives
\begin{multline}
(L_\beta^{(n)}(\by,\cdot),p)_\beta=\frac{1}{c_\beta}\int_{\RR^N}[\rme^{-\Delta_\beta/2}L_\beta^{(n)}(\bx,\by)][\rme^{-\Delta_\beta/2}p(\bx)]\rme^{-x^2/2}w_\beta(\bx)\ud\bx\\=p(\by)
\end{multline}
by Eq.~\eqref{PartialSumDunklKernel}. Then, one has
\begin{equation}
-\frac{\Delta_\beta}{2}L_\beta^{(n)}(\bx,\by)=-\frac{y^2}{2}L_\beta^{(n-2)}(\bx,\by).
\end{equation}
Taking the limit $n\to\infty$ yields
\begin{equation}
\lim_{n\to\infty}\rme^{-\Delta_\beta/2}L_\beta^{(n)}(\bx,\by)=\rme^{-y^2/2}V_\beta\rme^{\bx\cdot\by},
\end{equation}
which in turn gives
\begin{equation}
p(\by)=\frac{1}{c_\beta}\int_{\RR^N}[\rme^{-y^2/2}V_\beta\rme^{\bx\cdot\by}][\rme^{-\Delta_\beta/2}p(\bx)]\rme^{-x^2/2}w_\beta(\bx)\ud\bx.
\end{equation}
Setting $p(\bx)=L_\beta^{(m)}(\bx,\bz)$ and taking the limit $m\to\infty$ gives
\begin{equation}
V_\beta\rme^{\by\cdot\bz}=\frac{1}{c_\beta}\int_{\RR^N}[\rme^{-y^2/2}V_\beta\rme^{\bx\cdot\by}][\rme^{-z^2/2}V_\beta\rme^{\bx\cdot\bz}]\rme^{-x^2/2}w_\beta(\bx)\ud\bx.
\end{equation}
Moving the Gaussians of $\by$ and $\bz$ to the l.h.s. yields Eq.~\eqref{GaussianIntegralDunklKernel}.\qquad\qquad$\square$

% Bibliography:
% !TEX encoding = UTF-8 Unicode
\bibliographystyle{ieeetr}
%\cleardoublepage
%\addcontentsline{toc}{chapter}{Bibliography}
%\phantomsection
\bibliography{./tex/thesis_biblio}
 
\end{document}